\documentclass[12pt,english]{article}
\usepackage{mathtools,float}
\allowdisplaybreaks[4]
\usepackage[dvipsnames]{xcolor}
\usepackage[T1]{fontenc}
\usepackage[latin9]{inputenc}
\usepackage{geometry}
\usepackage{amsmath,amsthm,amssymb}
\usepackage[authoryear]{natbib}
\usepackage{booktabs}
\usepackage{graphicx}
\usepackage{subcaption}
\usepackage{titletoc}
\usepackage{comment}
\usepackage{tablefootnote}
\usepackage{setspace}

 \theoremstyle{definition}
  \newtheorem{example}{\protect\examplename}
   \theoremstyle{plain}
  \newtheorem{assumption}{\protect\assumptionname}
\theoremstyle{plain}
\newtheorem{thm}{\protect\theoremname}
\theoremstyle{plain}
\newtheorem{cor}{\protect\corollaryname}
  \theoremstyle{plain}

  \newtheorem{rem}{\protect\remarkname}
    \theoremstyle{plain}
  \newtheorem{lem}{\protect\lemmaname}
  
    \theoremstyle{plain}

\usepackage[backref=page]{hyperref}
\definecolor{darkblue}{rgb}{0, 0, 0.5}
\hypersetup{
     colorlinks = true,
     citecolor = Maroon,
     urlcolor = darkblue,
     linkcolor = darkblue
}

\makeatother

  \providecommand{\assumptionname}{Assumption}
  \providecommand{\examplename}{Example}
  
\providecommand{\theoremname}{Theorem}
\providecommand{\remarkname}{Remark}
\providecommand{\corollaryname}{Corollary}
  \providecommand{\lemmaname}{Lemma}
    \providecommand{\algorithmname}{Algorithm}

\newcommand{\indep}{\perp \!\!\! \perp}

\global\long\def\hmu{\hat{\mu}}
\global\long\def\hSigma{\hat{\Sigma}}
\global\long\def\hbeta{\hat{\beta}}
\global\long\def\RR{\mathbb{R}}
\global\long\def\TT{\mathbb{T}}
\global\long\def\tX{\tilde{X}}
\global\long\def\tZ{\tilde{Z}}

\global\long\def\tY{\tilde{Y}}
\global\long\def\htau{\hat{\tau}}
\global\long\def\hw{\hat{w}}
\global\long\def\tX{\tilde{X}}
\global\long\def\betamix{\beta_{{\rm mix}}}

\global\long\def\hw{\hat{w}}

\global\long\def\bbeta{\bar{\beta}}
\global\long\def\RR{\mathbb{R}}
\global\long\def\oneb{\mathbf{1}}%
\global\long\def\Mcal{\mathcal{M}}%
\global\long\def\bxi{\bar{\xi}}%
\global\long\def\Tcal{\mathcal{T}}%
\global\long\def\tw{\tilde{w}}%
\global\long\def\bw{\bar{w}}%
\begin{document}

\title{Debiasing and $t$-tests for synthetic control inference on average causal effects\footnote{This paper was previously circulated as ``Inference on average treatment effects in aggregate panel data settings'', ``Practical and robust $t$-test based inference for synthetic control and related methods'', and ``A $t$-test for synthetic controls''.  We are grateful to the Editor (Azeem Shaikh), anonymous referees, Bruno Ferman, Jerry Hausman, Guido Imbens, Anna Mikusheva, Yixiao Sun, Stefan Wager, and many seminar and conference participants for valuable comments.
We are grateful to Kathleen Li for sharing the code to implement their subsampling approach \citep{li2020statistical} and answering our questions regarding some implementation details.
K.W.\ is also affiliated with CESifo. V.C.\ gratefully acknowledges funding by the National Science Foundation. The usual disclaimer applies.}}

\author{Victor Chernozhukov\thanks{Department of Economics and Center for Statistics and Data Science, Massachusetts Institute of Technology. Email: \url{vchern@mit.edu}} \quad \quad Kaspar W\"{u}thrich\thanks{Department of Economics, University of Michigan, 238 Lorch Hall, 611 Tappan Ave, Ann Arbor, MI 48109-1220, USA. Email: \url{kasparwu@umich.edu}} \quad \quad Yinchu Zhu\thanks{Brandeis University; 415 South Street, Waltham, MA 02453, USA. Email: \url{yinchuzhu@brandeis.edu}}}

\date{First version on arXiv: December 27, 2018\quad This version: \today}

\maketitle

\begin{abstract}
We propose a practical and robust method for making inferences on average treatment effects estimated by synthetic controls. We develop a $K$-fold cross-fitting procedure for bias correction. To avoid the difficult estimation of the long-run variance, inference is based on a self-normalized $t$-statistic, which has an asymptotically pivotal $t$-distribution. Our $t$-test is easy to implement, provably robust against misspecification, and valid with stationary and non-stationary data. It demonstrates an excellent small sample performance in application-based simulations and performs well relative to other methods. We illustrate the usefulness of the $t$-test by revisiting the effect of carbon taxes on emissions. 

\medskip


\end{abstract}

\maketitle

\newpage

\section{Introduction}

We propose a simple and robust $t$-test for making inferences on average effects in aggregate panel data settings.  The $t$-test is based on a novel debiased synthetic control (SC) estimator.\footnote{The original SC method was introduced by \citet{abadie2003economic} and further developed by \citet{abadie10sc,abadie15sc}. See \citet{abadie2021using} for a review.}  
The $t$-test is easy to use. For example, one can obtain $(1-\alpha)$-confidence intervals as 
\begin{equation}
\text{\texttt{debiased SC estimate}}~\pm ~t(1-\alpha/2) \times \text{\texttt{standard error}},\label{eq:ci_intro}
\end{equation}
where $t(1-\alpha/2)$ is the $(1-\alpha/2)$ quantile of a $t$-distribution. Implementing the $t$-test only requires computing SC weights, which can be obtained from existing SC software packages.\footnote{The $t$-test is implemented in the \texttt{R}-package \texttt{scinference}  (\url{https://github.com/kwuthrich/scinference}).}

We consider a setting with one treated unit (labeled $i=0$), which is untreated for the first $T_0$ periods and treated for the remaining $T_1$ periods, and $N$ untreated control units (labeled $i=1,\dots,N$). Let $Y_{0t}(1)$ and $Y_{0t}(0)$ denote the (random) potential outcomes of the treated unit with and without the treatment. Our object of interest is the average treatment effect on the treated unit (ATT) in the post treatment period,
\[
\tau = \frac{1}{T_1}\sum_{t=T_0+1}^{T_0+T_1}\left(Y_{0t}(1)-Y_{0t}(0)\right).
\]
The ATT provides an interpretable and easy-to-communicate summary measure of the effect of the treatment. Average treatment effects on the treated units are popular target parameters in many causal inference problems because of their direct policy relevance. Like much of the SC literature,  we focus on settings with only one treated unit and many treated periods. In such settings, the average effect for the treated unit over time captures the impact of the treatment rather than the average effect across many units. We also show how to modify the proposed $t$-test to make inferences on the expected effect for the treated unit, $E\left(Y_{0t}(1)-Y_{0t}(0)\right)$, when researchers are willing to restrict treatment effect heterogeneity by imposing stationarity and weak dependence of the treatment effect sequence.

Many existing inference methods for SC with a single treated unit focus on sharp null hypotheses, such as the null hypothesis of no effect whatsoever, and per-period effects \citep[e.g.,][]{abadie10sc,firpo18synthetic,cattaneo2021prediction,chernozhukov2021exact}. As pointed out by \citet[][p.81]{imbens2015causal}, sharp null hypotheses are ``a very useful starting point, prior to any more sophisticated analysis,'' but rejecting sharp null hypotheses is ``not sufficient to inform policy decisions.'' Per-period effects provide useful information about the effect dynamics. But while we can make inferences about per-period effects, they cannot be consistently estimated when there is only one treated unit, and the resulting confidence intervals can be wide and uninformative. Moreover, reporting inferences on many per-period effects might not be an effective way to communicate the overall impact of the treatment when there are many treated periods, and often an interpretable one-number summary is called for. The ATT provides such a summary, is a natural inferential target in empirical work, and admits confidence intervals that take the ``standard'' form \eqref{eq:ci_intro}. By contrast, the existing methods for inference on sharp nulls and per-period effects are typically based on permutation or bootstrap approaches.

The main inferential challenge in SC applications is that $Y_{0t}(0)$ is unobserved for $t>T_0$. SC approximates this unobserved counterfactual using a weighted average of control outcomes, $\sum_{i=1}^{N}\hat{w}_iY_{it}$. The weights $(\hat{w}_1,\dots,\hat{w}_{N})$ are estimated based on the pre-treatment data and constrained to be positive and add up to one. 

The natural SC estimator of the ATT is
\begin{equation}
\htau^{\text{SC}} = \frac{1}{T_1}\sum_{t=T_0+1}^{T_0+T_1}\left(Y_{0t}-\sum_{i=1}^{N}\hat{w}_iY_{it}\right).\label{eq:naive_estimator}
\end{equation}
In many SC applications, $T_0$ is smaller than or comparable to $N$, and the data exhibit persistence and dynamics. In such settings, there are two major challenges when using $\htau^{\text{SC}} $ to make inferences about $\tau$. First, $\htau^{\text{SC}} $ is biased due the error from estimating the high-dimensional (relative to $T_0$) weights, even under correct specification\footnote{In our context, ``correct specification'' refers to settings where the population SC estimator is unbiased.}, and the bias can be substantial under misspecification (Figure \ref{fig:bias}).\footnote{The constraints imposed by SC amount to  $\ell_1$-regularization in estimating the weights. This essentially shrinks unbiased estimates and therefore induces bias. In general, regularization aims to strike a balance in the bias-variance tradeoff; see textbooks such as \citet{hastie2009elements}. Section 2.2 of the seminal paper by  \citet{tibshirani96regression} offers  a detailed explanation in the case of $\ell_1$-regularization. A more modern explanation of the regularization bias can be found in Section 1 of \citet{chernozhukov2018double}.} This bias of $\htau^{\text{SC}}$ precludes the application of standard inference procedures.

Second, even in the ideal case where the true SC weights are known, estimating standard errors is difficult since they depend on the long-run variance (LRV). It is well-known that classical LRV estimators \citep[e.g.,][]{newey1987simple,Andrews1991} are not accurate enough for reliable inference in small samples. Figure \ref{fig:undercoverage} shows that using the popular Newey-West standard errors yields substantial undercoverage, especially when $T_0$ is small.\footnote{As pointed out in \citet{muller2007theory}, any consistent estimator of LRV might suffer from robustness issues.} 

We develop an inference method that addresses both challenges and is motivated by an asymptotic framework where $T_0,T_1,N\rightarrow \infty$. To remove the bias of $\htau^{\text{SC}}$, we propose a $K$-fold cross-fitting procedure. To avoid the difficult LRV estimation, we propose a self-normalized $t$-statistic with an asymptotically pivotal $t$-distribution with $K-1$ degrees of freedom, which makes our method easy to implement. Self-normalization further induces theoretical higher-order improvements and yields excellent small sample properties in our simulations. Such higher-order improvements and excellent small sample properties are important for a method relying on asymptotic approximations since $T_0$ and $T_1$ are often not very large in SC applications. 

Figure \ref{fig:bias} shows that the distribution of the debiased estimator is centered at the true value ($\tau=0$) and well-approximated by a normal distribution, even under misspecification. Figure \ref{fig:undercoverage} shows that our $t$-test exhibits an excellent coverage accuracy, despite the small sample sizes. 

\begin{figure}[ht]
\caption{Bias of synthetic control and impact of debiasing}
	\label{fig:bias}
	\begin{center}
		\includegraphics[width=0.4\textwidth,trim = {0 1cm 0 2cm}]{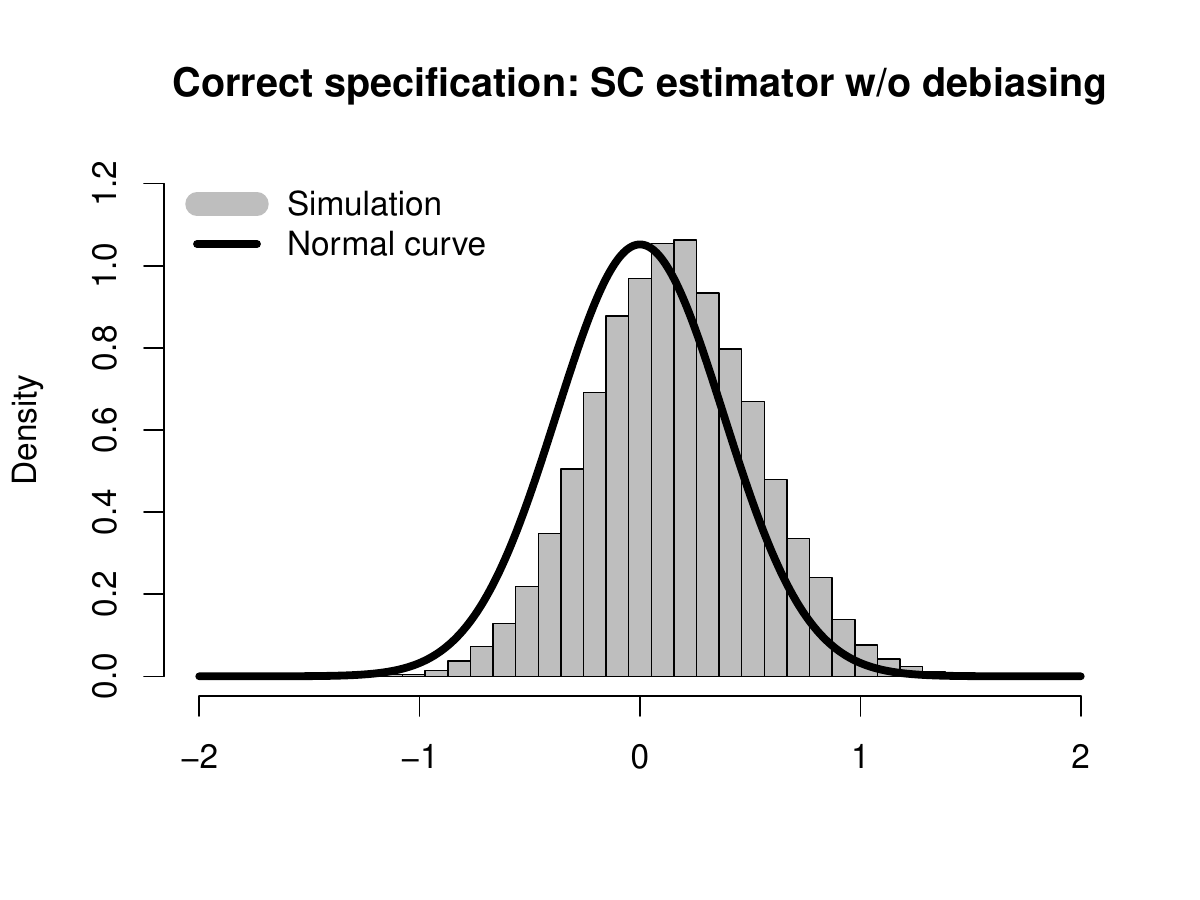}
		\includegraphics[width=0.4\textwidth,trim = {0 1cm 0 2cm}]{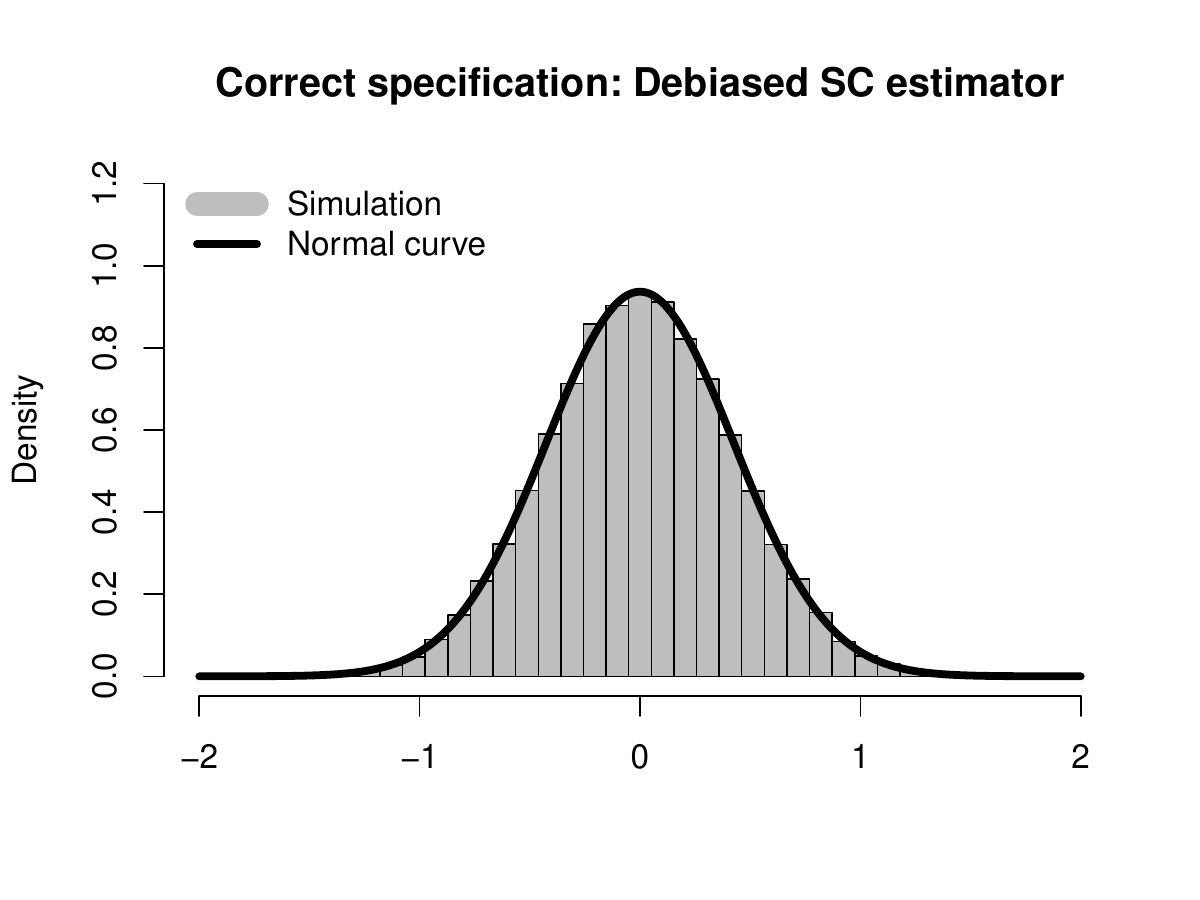}
		\includegraphics[width=0.4\textwidth,trim = {0 1.75cm 0 1cm}]{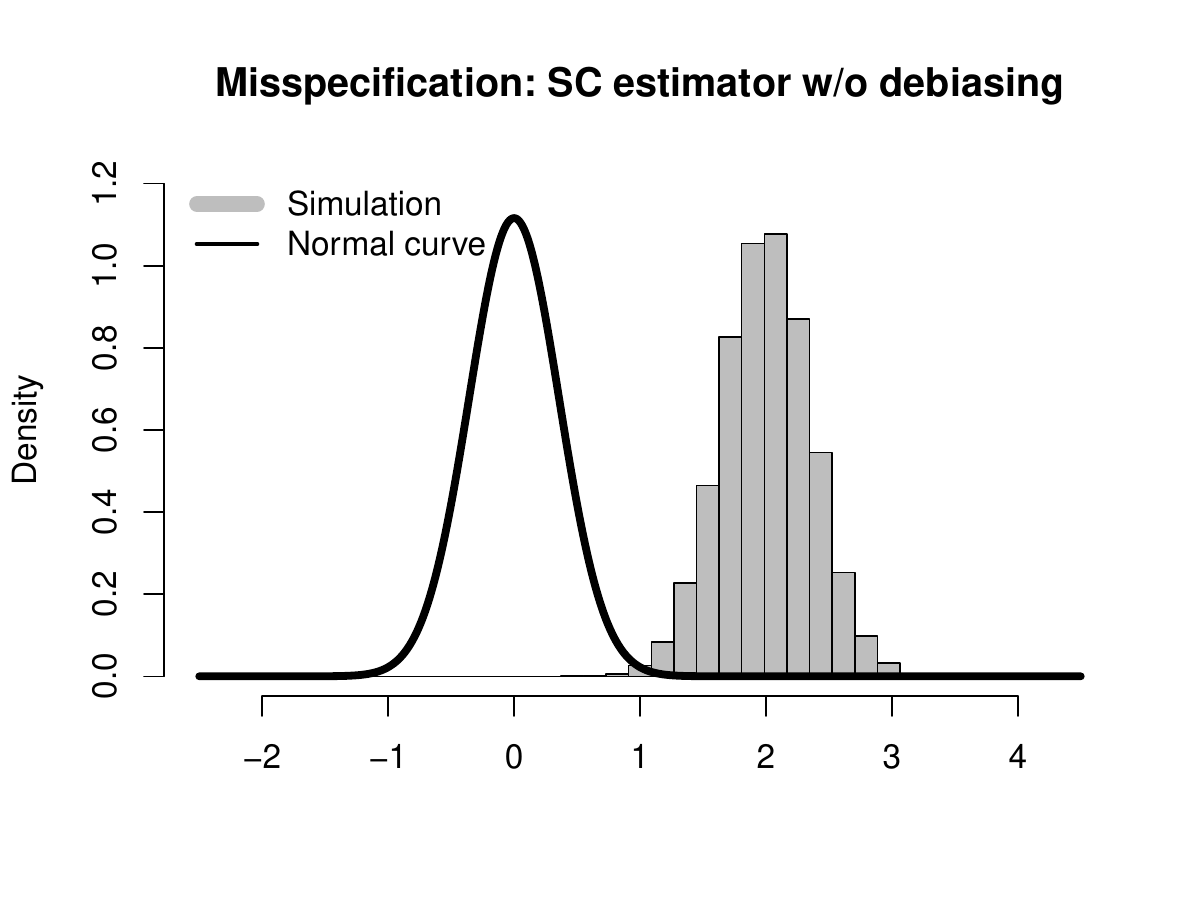}
		\includegraphics[width=0.4\textwidth,trim = {0 1.75cm 0 1cm}]{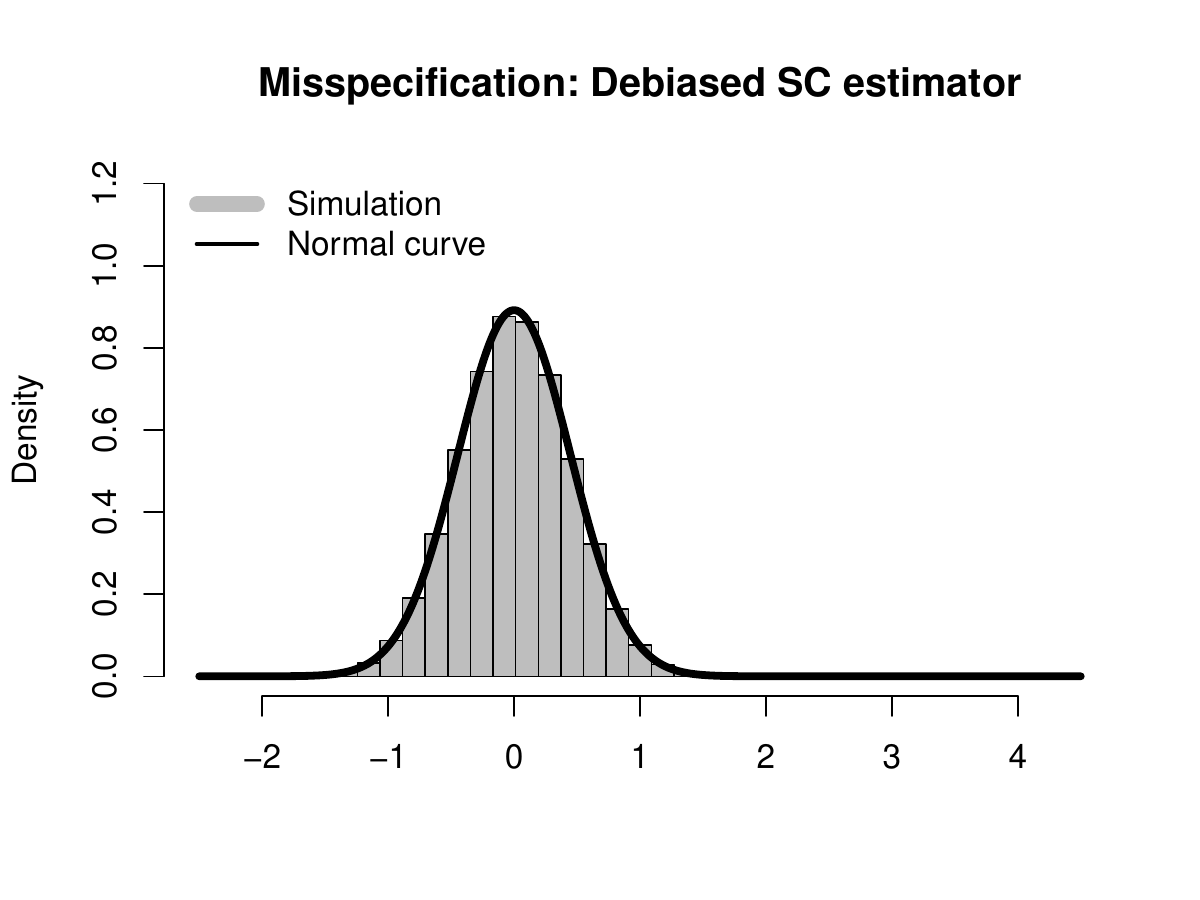}
	\end{center}
\scriptsize{\textit{Notes:} Simulations with 50,000 repetitions. $Y_{0t}(0)=\mu+\sum_{i=1}^Nw_iY_{it}(0)+u_t$, $Y_{it}(0)=2\cdot 1\{i\le 3\} + v_{it},v_{it}\overset{iid}\sim N(0,1)$, $\{u_t\}$ is a Gaussian AR(1) process with coefficient $0.31$, and $(T_0,T_1,N)=(30,16,14)$ as in the empirical application. Correct specification: $(\mu,w)=\left(0,(1/3,1/3,1/3,0\dots,0)'\right)$. Misspecification: $(\mu,w)=\left(2,(1/3,1/3,1/3,0\dots,0)'\right)$. $Y_{0t}(1)-Y_{0t}(0)=0$ for $t>T_0$. Debiasing based on $K=3$.}	
\end{figure}

\begin{figure}
\caption{Undercoverage with Newey-West standard errors}
	\label{fig:undercoverage}
	\begin{center}
  		\includegraphics[width=0.4\textwidth,trim = {0 1cm 0 2cm}]{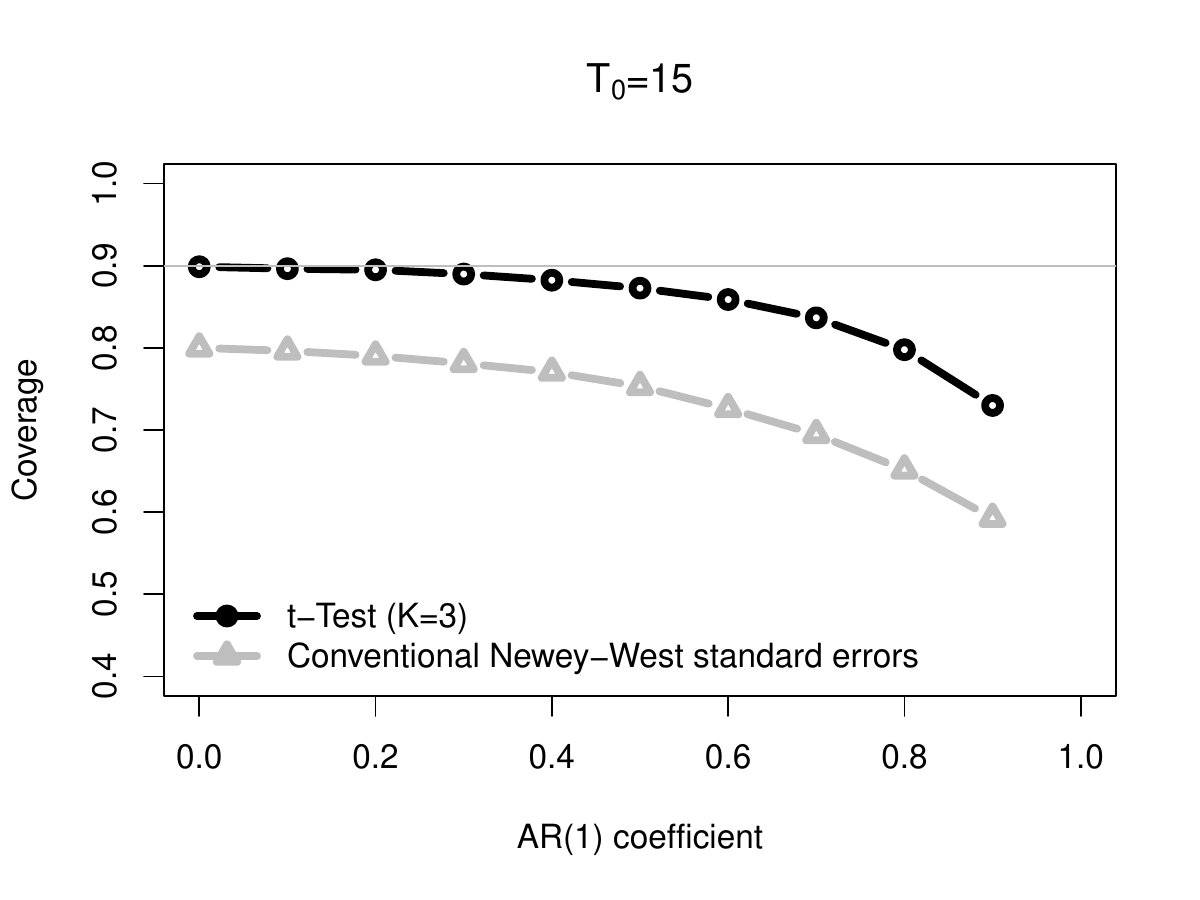}
      	\includegraphics[width=0.4\textwidth,trim = {0 1cm 0 2cm}]{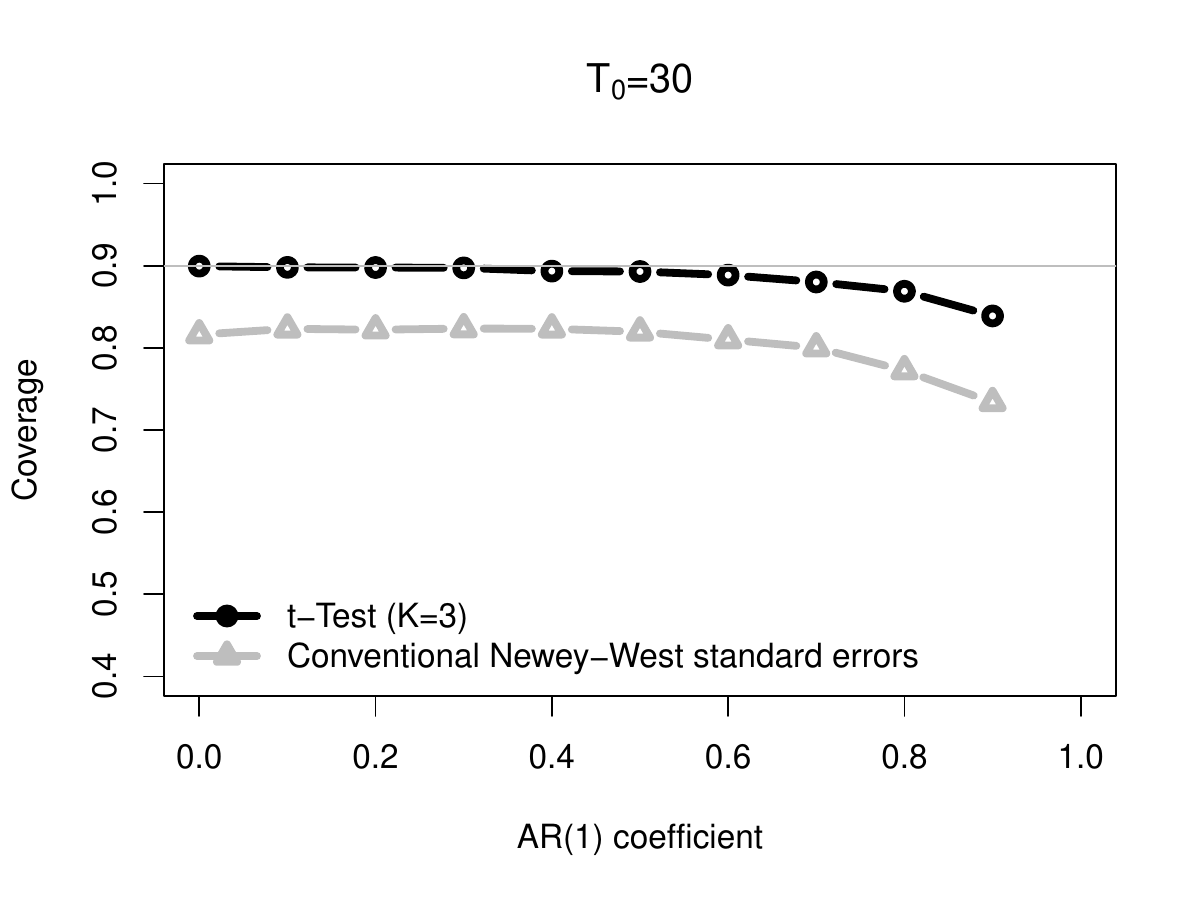}
	\end{center}
\scriptsize{\textit{Notes:} Simulations with 50,000 repetitions. Nominal coverage: 90\%. $Y_{0t}(0)=\sum_{i=1}^Nw_iY_{it}(0)+u_t$, $\{u_t\}$ is a Gaussian AR(1) process, $Y_{0t}(1)-Y_{0t}(0)=0$ for $t>T_0$, and $(T_1,N)=(16,14)$ as in the empirical application. We focus on the Oracle case where $w$ is known to abstract from the estimation error in $\hat{w}$. The LRV is estimated based on the pre-treatment data using the \texttt{NeweyWest} command (\texttt{R}-package \texttt{sandwich} \citep{sandwich2004,sandwich2020}) with pre-whitening.}	
\end{figure}

The $t$-test is valid with stationary and non-stationary data. As a result, researchers do not need to pre-test for nonstationarity and choose the inference method according to the pre-test. This is important since nonstationarity is a common feature of SC applications. Existing methods typically rely on assumptions on the type of nonstationarity (e.g., unit root or deterministic trend) \citep[e.g.,][]{li2020statistical}. However, it is well-known that even the slightest misspecification in modeling nonstationarity (e.g., unit root vs.\ near unit root) yields invalid inferences. Therefore, we take a different approach and do not make assumptions about the exact nature of the nonstationarity. Instead, we restrict the heterogeneity in the nonstationarity across units. That is, we require the units to be sufficiently similar, which is a crucial contextual requirement in SC studies \citep{abadie2021using}. Avoiding assumptions about the exact type of nonstationarity is important because one of the main reasons for using SC is that the precise nature of the counterfactual trend is unknown. 

With stationary data, the $t$-test is valid under arbitrary misspecification. With non-stationary data, we establish the validity of the $t$-test under two different settings. First, the $t$-test is valid under arbitrary misspecification when all units share a common nonstationarity. Second, the $t$-test remains valid when units deviate from the common nonstationarity under restrictions on the magnitude and heterogeneity of the deviations. The second result covers many relevant types of nonstationarity, such as heterogeneous deterministic time trends and certain forms of cointegration, but requires SC to be correctly specified. 

A by-product of our method is improved asymptotic efficiency over difference-in-differences (DID). We show that the asymptotic variance of the debiased SC estimator is no larger than that of DID, irrespective of whether the SC model is correctly specified or not. Moreover, the $t$-test is valid when the common trends assumption underlying DID is violated and thus also more robust than DID.

The proposed debiasing strategy and $t$-test are generic and can be applied in conjunction with many other (than SC) estimators of the ATT. Our focus on SC estimators is motivated by their popularity and tractability. However, the required conditions, including $\ell_2$-consistency of the weights $(\hat{w}_1,\dots,\hat{w}_{N})$, can be verified for a broad class of estimators  available in the literature. See Remark \ref{rem:generic_Result} for further discussions.

We illustrate the usefulness of our method by revisiting \citet{andersson2019carbon}'s analysis of the effect of a carbon tax on emissions between 1990 and 2005 in Sweden. In this application, the ATT captures the average effect of the carbon tax on emissions  --- a one-number summary of the overall effect that \citet[][p.14]{andersson2019carbon} explicitly mentions. The $t$-test provides robust confidence intervals for this parameter. The estimated ATT is negative and significant. Our findings complement and corroborate the inference results in \citet{andersson2019carbon}, which are based on the permutation test of \citet{abadie10sc}.

The $t$-test demonstrates an excellent small sample performance in  simulations calibrated based on the empirical application and performs well relative to existing alternatives such as DID, subsampling \citep{li2020statistical}, and synthetic DID (SDID) \citep{arkhangelsky2021synthetic}. 

Finally, we provide recommendations for practice. We discuss when to use the $t$-test and what to consider when implementing it in applications.

\subsection{Related literature}
\label{sec:literature}
We contribute to the literature by proposing methods for making inferences on the ATT and expected treatment effects based on SC in settings with one or few treated units. The closest papers to ours are \citet{li2020statistical} and \citet{arkhangelsky2021synthetic}. Here we compare and contrast the proposed $t$-test with their contributions. See Section \ref{sec:simulations} for a simulation comparison and Section \ref{sec:guide} for guidance on which method to choose in applications.

\citet{li2020statistical} proposes a subsampling approach to inference in SC settings with a small number of controls ($N$ is assumed to be fixed).\footnote{Here we describe \citet{li2020statistical}'s method in the context of the SC estimator $\htau^{\text{SC}}$. We note that their theory also covers more general SC estimators, such as SC estimators without the adding-up constraint.} They derive the asymptotic distribution of the SC estimator of $E(\alpha_t)$, where $\alpha_t=Y_{i0}(1)-Y_{i0}(0)$, using a projection framework with stationary, trend-stationary, and unit root data. For inference, they propose a subsampling procedure. They note that 
$\sqrt{T_1}(\htau^{\text{SC}} -E(\alpha_t))$ can be decomposed into two parts, $A_1$ and $A_2$. Since only $A_1$ depends on the estimated SC weights, subsampling is only required for this part, and $A_2$ can be approximated using the bootstrap. \citet{li2020statistical} establishes the validity of this procedure for the case where the SC prediction errors and $\{\alpha_t-E(\alpha_t)\}$ are serially uncorrelated. When the data are trend-stationary, they recommend detrending before applying the inference method.

Our $t$-test differs from \citet{li2020statistical} in several important aspects. We allow for a large number of control units ($N$ can grow with $(T_0,T_1)$) and  explicitly correct for the bias of SC using a cross-fitting approach. This leads to a self-normalized $t$-statistic and an inference procedure with higher-order improvements that does not require pre-processing the data to make them stationary. The $t$-test also avoids subsampling, which may not perform well in small samples.

\citet{arkhangelsky2021synthetic} propose synthetic DID (SDID), which combines ideas from DID and SC. In our notation, their estimator of the ATT, $\htau^{\text{SDID}}$, is given by 
\begin{equation}
(\htau^{\text{SDID}},\hat\mu,\hat\alpha,\hat\beta)=\underset{\tau,\mu,\alpha,\beta}{\arg\min}\sum_{i=0}^N\sum_{t=1}^{T_0+T_1}(Y_{it}-\mu-\alpha_i-\beta_t-D_{it}\tau)^2\hat{w}^{\text{SDID}}_i\hat\lambda^{\text{SDID}}_t,\label{eq:sdid}
\end{equation}
where $D_{it}$ is the treatment indicator, $\hat{w}^{\text{SDID}}$ are unit weights that balance the pre-treatment trends of treated and control units, and $\hat{\lambda}^{\text{SDID}}$ are time weights that balance the pre- and post-treatment periods. SDID thus adds time weights and unit fixed effect to the SC method. Both of these additions help reduce the bias of standard SC and improve robustness. \citet{arkhangelsky2021synthetic} establish consistency and asymptotic normality of SDID under a latent factor model for the potential outcomes. To make inferences with one treated unit, they build on ideas in \citet{conley2011inference} to propose a placebo variance estimator that relies on homoskedasticity across units.\footnote{An appealing feature of SDID is that it naturally accommodates multiple treated units. For such settings, \citet{arkhangelsky2021synthetic} propose a bootstrap and a jackknife method for inference.}

Our $t$-test differs from SDID with respect to the underlying model, the debiasing strategy, and the inference method. The $t$-test is developed based on a linear prediction model, which can be motivated by but does not rely on factor models. We use a cross-fitting approach that directly removes the bias if the bias is stable over time. To make inferences, we exploit stationarity over time and build on the cross-fitting structure to construct a self-normalized $t$-statistic that avoids estimating the LRV, is theoretically valid across a variety of settings, and enjoys higher-order improvements. 

Throughout the paper, we focus on SC estimators of the counterfactual. However, our method only requires $\ell_2$-consistency of the estimator of the weights, which can be established for many penalized regression estimators. Therefore, our paper further contributes to the literature on inference methods based on standard and penalized regression estimators motivated by asymptotics where $T_0,T_1\rightarrow \infty$. Building on the framework of \citet{hsiao2012panel}, \citet{li2017estimation} propose a least squares method for making inferences on the expected treatment effect based on estimators of the LRV when the data are stationary. \citet{carvalho2018arco} propose a Lasso-based method for making inferences on the ATT based on LRV estimators under sparsity when the data are stationary.\footnote{In Section 4.1, they consider an extension to trending regressors but do not provide inference methods for high-dimensional settings.} \citet{masini2020counterfactual} study the asymptotic distribution of counterfactual estimators based on least squares when the data are non-stationary. They show that the limiting distribution is generally nonstandard and depends on $T_0/T_1$. They propose a subsampling method for inference when $T_0\approx T_1$. 
Compared to this strand of the literature, our $t$-test is generic in that it accommodates a wide range of penalized and unpenalized regression estimators, is valid under general forms of nonstationarity, does not rely on estimating the LRV, and enjoys higher-order improvements. Moreover, the debiased ATT estimator is asymptotically normal across all settings we consider, which, together with the self-normalized $t$-statistic, makes the proposed method easy to implement.

Focusing on average effects over time and expected effects, we complement the existing procedures for testing sharp null hypotheses about the treatment effect trajectory $\{Y_{0t}(1)-Y_{0t}(0)\}_{t=T_0+1}^{T_0+T_1}$ and making inferences on per-period effects \citep[e.g.,][]{abadie10sc,firpo18synthetic,cattaneo2021prediction,benmichael2021augmented,chernozhukov2021exact,masini2021jasa,shaikh2021randomization}.\footnote{The time series permutation approach of \citet{chernozhukov2021exact} can be extended to test hypotheses about the ATT by collapsing the data across time, provided that $T_0\gg T_1$. More recently, \citet{cattaneo2023uncertainty} extended the method in \citet{cattaneo2021prediction} to accommodate more general treatment effects, including ATT over time and across units in staggered adoption designs.} Within that strand of literature, our paper is most closely related to  \citet{benmichael2021augmented}, given our focus on bias correction. \citet{benmichael2021augmented} propose a debiasing approach based on an outcome model and suggest using a conformal prediction approach for making inferences on per-period effects, building on \citet{chernozhukov2021exact}. Finally, it is worthwhile noting that the $t$-test relies on a sampling-based framework in which the potential outcomes are viewed as random, whereas some of the approaches for testing sharp null hypotheses \citep[e.g.,][]{abadie10sc,firpo18synthetic} are motivated from a design-based perspective.

\subsection{Notation} Let $\mathbf{1}_p$ denote a $p\times 1$ dimensional vector of ones. For $q \geq  1$, we denote the $\ell_{q}$-norm of a vector as $\|\cdot\|_{q}$.  For a matrix $A$, we denote $\|A\|_{\infty}=\|{\rm vec}(A)\|_{\infty}$, where ${\rm vec}(A)$ is the column-wise vectorization of $A$.  We write $a\lesssim b$ to
denote $a \leq  cb$ for some constant $c > 0$ that does not depend on the sample size. We
write $a \asymp b$ to denote $a\lesssim b$ and $b\lesssim a$. For a set $A$, $|A|$ is the cardinality of $A$. 

\section{Debiasing and $t$-tests for SC inference on the ATT}
\label{sec:methodology}

\subsection{Setup and object of interest}

We consider a synthetic control setup with one treated unit, $N$ control units, and $T$ periods \citep[e.g.,][]{abadie10sc,DI16,kellogg2021combining}. The treated unit is untreated for the first $T_0$ periods and treated for the remaining $T-T_0=T_1$ periods. The control units remain untreated throughout. The potential outcomes with and without the treatment are $Y_{it}(1)$ and $Y_{it}(0)$, respectively. In the main text, we assume that the treatment status is fixed and label the treated unit as $i=0$ and the control units as $i=1,\dots,N$. In Appendix \ref{app:random_treatment}, we provide sufficient conditions under which the $t$-test remains valid when the treatment status is random and selection is based on past outcomes or latent variables in factor models.  Observed outcomes are given by $Y_{it}=Y_{it}(0)+\alpha_{it}\mathbf{1}\{i=0,t>T_0\}$, where  $\alpha_{it}:=Y_{it}(1)-Y_{it}(0)$ is the treatment effect for unit $i$ in period $t$. 

The existing inference approaches for SC and related methods with one (or few) treated units differ regarding the identification assumptions they rely on and what is assumed to be random and fixed, respectively. We consider a setting where the potential outcomes are random and the treatment effect sequence can be either random or fixed.\footnote{Similar settings are considered by \citet{li2020statistical}, \citet{arkhangelsky2021synthetic}, \citet{benmichael2021augmented}, \citet{cattaneo2021prediction}, \citet{chernozhukov2021exact}, \citet{ferman2021properties}, \citet{ferman2021synthetic}, and \citet{cattaneo2023uncertainty}. An alternative and complementary strand of the literature focuses on design-based inference, treating the potential outcomes as fixed and leveraging assumptions on the assignment process \citep[e.g.,][]{abadie10sc,firpo18synthetic}.}

Our main goals are to estimate and construct confidence intervals for the ATT,  
\begin{equation}
\tau=\frac{1}{T_1}\sum_{t=T_0+1}^T\alpha_{0t}\label{eq:att}.
\end{equation}
When $\{\alpha_{0t}\}$ and thus $\tau$ are random, these confidence intervals are not confidence intervals in the traditional sense, but should be interpreted as prediction intervals \citep{cattaneo2021prediction,chernozhukov2021exact,cattaneo2023uncertainty}. See Sections \ref{sec:implementation} and \ref{sec: asymptotic properties} for further discussions. To simplify the exposition, we sometimes suppress the subscript ``$0$'' and write $\alpha_{t}=\alpha_{0t}$ whenever there is no ambiguity.

The ATT is an interpretable and easy-to-communicate summary measure of the impact of the treatment when treatment effects are heterogeneous over time. For example, in our reanalysis of \citet{andersson2019carbon}, the ATT captures the average per-period effect of a carbon tax on emissions in Sweden between 1990, when the tax was introduced, and 2005, when the EU started its emissions trading system. \citet[][p.14]{andersson2019carbon} explicitly mentions the ATT when discussing the empirical results. Estimates of the ATT are commonly reported in empirical applications \citep[e.g.,][]{abadie10sc,bohn2014did,abadie15sc,pinotti2015economic,andersson2019carbon,jones2022labor}.  
In settings with only one treated unit, the ATT in \eqref{eq:att} is equivalent, for example, to the treatment effects considered by \citet{carvalho2018arco}, \citet{arkhangelsky2021synthetic}, and \citet{cattaneo2023uncertainty}. 

We consider inference on the ATT in an asymptotic framework where $T_1\rightarrow \infty$. This framework and the assumptions required for our theoretical results suggest that the $t$-test is suitable for applications with enough post-treatment periods and without structural breaks. When there are not enough post-treatment periods (e.g., due to data limitations or control units also getting treated in staggered adoption designs) or when there are structural breaks shortly after $T_0$, the $t$-test should not be used. See Section \ref{sec:guide} for a detailed discussion on when to use the $t$-test and what to consider when using it.

\begin{rem}
The proposed method can be applied to make inferences on the ATT over subperiods of the post-treatment period, $\mathcal{T}^s=\{T_0+1+r,\dots,T-s\}$, where $r,s\ge 0$,
$
|\mathcal{T}^s|^{-1}\sum_{t\in \mathcal{T}^s}\alpha_{t}.
$
Our asymptotic theory requires that $|\mathcal{T}^s|\rightarrow \infty$, so that $|\mathcal{T}^s|$ needs to be large enough for the asymptotic approximations to be accurate. \qed
\end{rem}

\begin{rem}
In settings where researchers are willing to restrict treatment effect heterogeneity over time and assume that the treatment effect sequence $\{\alpha_{t}\}$ is stationary and weakly dependent, the expected treatment effect, $E(\alpha_{t})$, is a natural alternative to the ATT \citep[e.g.,][]{li2017estimation,li2020statistical}. We propose a modification of our method for making inferences on $E(\alpha_{t})$ in Appendix \ref{app: inference on expected effect}. Remark \ref{rem:att_vs_expected_effect} discusses the choice between the ATT and the expected treatment effect. \qed
\end{rem}

\subsection{Debiasing SC and implementing the $t$-test}
\label{sec:implementation}
To remove the bias of SC, we employ a $K$-fold cross-fitting procedure, where $K$ is fixed. We discuss the choice of $K$ in Section \ref{sec: choosing K}. We choose $K$ consecutive blocks from the pre-treatment period: $H_{1}\bigcup H_{2}\bigcup\cdots\bigcup H_{K}\subseteq\{1,\dots,T_{0}\}$. Define $r=\min\{\left\lfloor T_{0}/K\right\rfloor,T_1 \} $ and let $H_{k}=\{(k-1)r+1,\dots,kr\}$ for $1\leq k\le K$.\footnote{Here, we choose to use the first $K$ blocks. Other choices such as the last $K$ blocks are also valid.} For simplicity, we assume that $T_{0}/K$ is an integer. For $k=1,\dots,K$, compute
\begin{equation}
\htau_{k}=\frac{1}{T_{1}}\sum_{t=T_{0}+1}^{T}\left(Y_{0t}-\sum_{i=1}^N\hw_{i,(k)}Y_{it}\right)-\frac{1}{|H_{k}|}\sum_{t\in H_{k}}\left(Y_{0t}-\sum_{i=1}^N\hw_{i,(k)}Y_{it}\right), \label{eq:tau_hat}
\end{equation}
where $\hw_{(k)}=(\hw_{1,(k)},\dots, \hw_{N,(k)})'$ is obtained by applying SC to the data in $H_{(-k)}:=\{1,\dots,T_0\}\setminus H_k$. This construction ensures that under weak dependence, $\hat{w}_{(k)}$ is approximately independent of the data in $H_{k}\bigcup\{T_{0}+1,\dots,T\} $, which allows us to establish the theoretical validity of our procedure under weak conditions. 

The basic idea behind the construction of estimator \eqref{eq:tau_hat} is as follows. The first component, $T_1^{-1}\sum_{t=T_{0}+1}^{T}\left(Y_{0t}-\sum_{i=1}^N\hw_{i,(k)}Y_{it}\right)$, corresponds to the natural SC estimator $\htau^{\text{SC}}$ in \eqref{eq:naive_estimator} with $\hw$ replaced by $\hw_{(k)}$. This estimator is biased (see Figure \ref{fig:bias}). The second component, $|H_{k}|^{-1}\sum_{t\in H_{k}}\left(Y_{0t}-\sum_{i=1}^N\hw_{i,(k)}Y_{it}\right)$, is an estimator of the bias of this estimator in the pre-treatment period. Under the assumptions specified below, the bias is the same in the pre-treatment and the post-treatment period. As a result, subtracting the second component removes the bias.

For concreteness, we consider the following canonical SC estimator \citep[e.g.,][]{DI16} in the main text:
\begin{equation}
\hw_{(k)}\in \underset{w\in \mathcal{W}^{SC}}{\arg\min}\sum_{t\in H_{(-k)}}\left(Y_{0t}-\sum_{i=1}^Nw_iY_{it}\right)^{2},\label{eq:sc_estimator}
\end{equation}
where $\mathcal{W}^{SC}:=\left\{w: ~w_{i}\ge 0,~ \sum_{i=1}^N w_{i}=1 \right\}$.\footnote{The argmin may not be unique if $T_0 - r < N$. In this case, we can take $\hw_{(k)}$ to be any element in the argmin set.} We study the theoretical properties of the classical SC estimator of \citet{abadie10sc} in Appendix \ref{sec:constistency_classical_SC}.
The $t$-test only requires an $\ell_2$-consistent estimator of the weights. Therefore, it also works in conjunction with many other SC estimators and penalized regression estimators; see Remark \ref{rem:generic_Result}.

The final estimator of the ATT is simply the average of $\htau_{1},\dots,\htau_K$,
\begin{equation}\label{eq: estimator ATE}
\hat{\tau}=\frac{1}{K}\sum_{k=1}^{K}\htau_{k}.
\end{equation}
To avoid the difficult estimation of the LRV, we construct a scale-free test statistic. The idea is to form a ratio in which the numerator and the denominator are both scaled by the long-run standard deviation. Specifically, we construct a quantity based on $\htau_{1},\dots,\htau_K$:
\begin{eqnarray}\label{eq: tstat}
\mathbb{T}_K = \frac{\sqrt{K}\left(\htau-\tau \right)}{\hat\sigma_{\htau}},
\end{eqnarray}
where 
\[
\hat\sigma_{\htau}=\sqrt{1+\frac{Kr}{T_1}}\sqrt{\frac{1}{K-1}\sum_{k=1}^K \left(\htau_k -\htau\right)^2}.
\]
In Sections \ref{sec: theory stationary} and \ref{sec:nonstationarity} we show that $\mathbb{T}_K$ has an asymptotic $t$-distribution with $K-1$ degrees of freedom as $T_0,T_1\rightarrow \infty$. Note that $\hat\sigma_{\htau}/\sqrt{K}$ is different from standard errors based on asymptotic normality and consistent estimators of the variance. In our framework $\hat\sigma_{\htau}$ is a random variable since $K$ is fixed, as in the classical statistical analysis of $t$-tests based on a fixed number of Gaussian variables.

The $t$-statistic \eqref{eq: tstat} differs from the standard $t$-statistic by the factor $\sqrt{1+Kr/T_1}$ in the denominator. This rescaling is necessary in our context because the $\htau_{1},\dots,\htau_K$ are not asymptotically independent since they share a common component coming from the average over the post-treatment period. The common component cancels out in the denominator but not in the numerator of $\mathbb{T}_K$. Therefore, to account for this common component and show that the test statistic $\mathbb{T}_K$ has an asymptotic $t$-distribution, we need to divide the numerator by $\sqrt{1+Kr/T_1}$. We refer to the proof of Theorem \ref{thm: t statistic} for a formal discussion.

The asymptotic $t$-distribution of $\mathbb{T}_K$ suggests the following $(1-\alpha)$ confidence interval for $\tau$:
\begin{equation}
\mathcal{I}_{K}(1-\alpha)=\left [ \htau -t_{K-1}(1-\alpha/2)\frac{\hat\sigma_{\htau}}{\sqrt{K}},~ \htau +t_{K-1}(1-\alpha/2) \frac{\hat\sigma_{\htau}}{\sqrt{K}}\right], \label{eq:ci}
\end{equation}
where $t_{K-1}(1-\alpha/2)$ is the $(1-\alpha/2)$-quantile of a student $t$-distribution with $K-1$ degrees of freedom. The interpretation of $\mathcal{I}_{K}(1-\alpha)$ depends on whether $\tau$ is fixed or random. If $\tau$ is fixed, $\mathcal{I}_{K}(1-\alpha)$ can be interpreted as a conventional confidence interval; if $\tau$ is random, $\mathcal{I}_{K}(1-\alpha)$ can be interpreted as a prediction interval \citep[as in][]{cattaneo2021prediction,chernozhukov2021exact,cattaneo2023uncertainty}. 
For simplicity, we will refer to $\mathcal{I}_K(1-\alpha)$ as a confidence interval throughout.

\begin{rem}\label{rem:IM}
The construction of $\mathbb{T}_K$ is related to \citet[][]{ibragimovmueller2010} with two important differences. First, $\{\hat\tau_k\}$ naturally arise from our cross-fitting procedure for bias correction. Second, the $\hat\tau_k$s share a common component and, consequently, are not asymptotically independent (Theorem \ref{thm: asy distr}). Nevertheless, we are able to show that, after scaling the denominator by $\sqrt{1+K r/T_1}$, $\mathbb{T}_K$ has an asymptotic $t$-distribution under the null hypothesis (Theorem \ref{thm: t statistic}). Our construction of $\mathbb{T}_K$ is also related to the \citet{fama1973risk} variance estimator \citep[e.g.,][]{cattaneo2020characteristic}. \qed
\end{rem}

\section{Theoretical properties with stationary data}
\label{sec: theory stationary}
Here, we show that our $t$-test is valid with stationary data, robust to misspecification, and more efficient than DID. We consider inference in a repeated sampling framework under asymptotics where $T_0,T_1\rightarrow \infty$, and our results also allow for $N\rightarrow \infty$.
\subsection{Asymptotic theory}
\label{sec: asymptotic properties}
Following the SC literature, we predict $Y_{0t}(0)$ for $t>T_0$ using a linear combination of $\left(Y_{1t}(0),\dots,Y_{Nt}(0)\right)$. Defining $Y_{t}(0):=Y_{0t}(0)$, $Y_{t}:=Y_{0t}$, and $X_t:=\left(Y_{1t}(0),\dots,Y_{Nt}(0)\right)'$, the linear prediction model can be written as
\begin{equation}
Y_{t}(0) = X_t'w_\ast + u_{t}, \quad 1\le t\le T,  \label{eq:prediction_model}
\end{equation}
where the pseudo-true SC weights are defined as $w_{*}:=\arg\min_{w\in \mathcal{W}^{SC}}\ E(Y_{t}(0)-X_{t}'w)^{2}$ and the pseudo-true residuals or prediction errors are $u_{t}:=Y_t(0)-X_t'w_{*}$.\footnote{If the best linear predictor is well-defined and satisfies the SC constraints so that $w_\ast=[E(X_t X_t')]^{-1}E(X_t Y_t(0))\in \mathcal{W}^{SC}$, then $E(X_t u_t)= 0$. However, in practice, $\mathcal{W}^{SC}$ could be ``too small'' such that $[E(X_t X_t')]^{-1}E(X_t Y_t(0))\notin \mathcal{W}^{SC}$ and $E(X_t u_{t})\ne 0$.} We interpret model \eqref{eq:prediction_model} as a statistical or predictive model and not as a structural model.\footnote{See, for example, \citet{cattaneo2021prediction} and \citet{chernozhukov2021exact} for similar interpretations.} A major advantage of this interpretation is that  it allows SC to be misspecified and the pseudo-true weights $w_\ast$ to be different from the true (infeasible) SC weights.\footnote{Starting with the seminal paper by \citet{abadie10sc}, the true SC weights have often been defined via a factor model for the potential outcomes.}  Allowing for misspecification is important in practice. For example, under a linear factor model for the potential outcomes, SC does not recover the true weights and is biased in general \citep[e.g.,][]{ferman2021synthetic}. More generally, the true model might be nonlinear. The proposed inference method is valid in both of these cases and also accommodates many other forms of misspecification. For our main analysis, we assume that the predictive relationship is stable over time. In Section \ref{sec: time-varying weights}, we show that the $t$-test can accommodate certain forms of time-varying weights.

Throughout this section, we maintain the following standard stationarity assumption. We establish the properties of our method with non-stationary data in Section \ref{sec:nonstationarity} and discuss its robustness to non-stationary prediction errors in Appendix \ref{app:robustness}.
\begin{assumption}
\label{assu: stationary Y0 X} $\{(Y_{t}(0),X_t)\}_{t=1}^{T} $ is covariance-stationary.
\end{assumption}

By simple algebra, Assumption \ref{assu: stationary Y0 X} gives us the following observation. 

\begin{lem}
\label{lem: algebra} Let Assumption \ref{assu: stationary Y0 X} hold.
Then, for $1\leq k\le K$,
\[
\hat{\tau}_{k}-\tau=\left(\frac{1}{T_{1}}\sum_{t=T_{0}+1}^{T}u_{t}-\frac{1}{|H_{k}|}\sum_{t\in H_{k}}u_{t}\right)+\left(\frac{1}{|H_{k}|}\sum_{t\in H_{k}}\tX_{t}-\frac{1}{T_{1}}\sum_{t=T_{0}+1}^{T}\tX_{t}\right)'\Delta_{(k)},
\]
where $\tilde{X}_{t}=X_{t}-E(X_{t})$ 
and $\Delta_{(k)}=\hw_{(k)}-w_\ast$. 
\end{lem}
To establish the asymptotic properties of our method, we need to show that the second term in Lemma \ref{lem: algebra} is negligible. For that, we impose additional assumptions.

We start by imposing $\ell_{2}$-consistency of the SC estimator $\hat{w}_{(k)}$.
\begin{assumption}\label{ass:ell_2_consistency} $\max_{1\leq k\leq K}\|\hw_{(k)}-w_\ast\|_{2}=o_{P}(1)$. 
\end{assumption}
We state Assumption \ref{ass:ell_2_consistency} in terms of the pseudo-true SC weights $w_\ast:=\arg\min_{w\in \mathcal{W}^{SC}}\ E(Y_{t}(0)-X_{t}'w)^{2}$. The results in Theorems \ref{thm: asy distr} and \ref{thm: t statistic} below continue to hold as long as the estimators $\hw_{(k)}$ converge to any time-invariant vector of weights. 

The following lemma verifies Assumption \ref{ass:ell_2_consistency} for the SC estimator $\hw_{(k)}$. It allows $N$ to be large relative to $T_0$. 
\begin{lem}
	\label{lem: consistency misspecification}  Suppose Assumption \ref{assu: stationary Y0 X} and the following conditions hold.
\begin{enumerate} \setlength\itemsep{0pt}
\item $\max_{1\le k \le K}\|\hmu_{(-k)}-\mu\|_{\infty}=o_P(1)$, where $\mu=EX_{t}Y_{t}(0)$ and $\hmu_{(-k)}=|H_{(-k)}|^{-1}\sum_{t\in H_{(-k)}}X_{t}Y_{t}(0)$. 
\item $\max_{1\le k \le K}\|\hSigma_{(-k)}-\Sigma\|_{\infty}=o_{P}(1)$ and $\lambda_{\min}(\Sigma)\geq c$, where $\Sigma=EX_{t}X_{t}'$ and  $\hSigma_{(-k)}=|H_{(-k)}|^{-1}\sum_{t\in H_{(-k)}}X_{t}X_{t}'$.
\end{enumerate}	
Then, we have that $\max_{1\le k \le K}\|\hw_{(k)}-w_{*}\|_{2}=o_P(1)$. In particular,
	\[
	\max_{1\le k \le K}\|\hw_{(k)}-w_{*}\|_{2}^{2}\leq\frac{6\|\hSigma_{(-k)}-\Sigma\|_{\infty}+2\|\hmu_{(-k)}-\mu\|_{\infty}}{c}.
	\]
\end{lem}
The first condition in Lemma \ref{lem: consistency misspecification} holds under weak serial dependence, mild conditions on the tail of the distribution of the variables, and conditions on $N$. For example, when the entries of $X_t$ and $Y_t(0)$ are sub-Gaussian, we can allow for $\log N=o(\sqrt{T_{0}})$; when the entries of $X_t$ and $Y_t(0)$ have bounded $q $th moment for $q>2$, then we can typically allow for $N=o(T_{0}^{q/4}) $. This feature is essential because $N$ and $T_0$ have a similar order of magnitude in many SC applications. The second condition requires the eigenvalues of $\Sigma_{(-k)} $ to be bounded away from zero to achieve identification of the pseudo-true SC weights. We emphasize that Lemma \ref{lem: consistency misspecification} does not impose any sparsity assumptions on the weights. 

We also impose weak dependence assumptions on the data. Define $\tilde{u}_t:=u_t-E(u_t)$ and $\sigma^{2}=\lim_{T\rightarrow\infty}E\left(T^{-1/2}\sum_{t=1}^{T}\tilde{u}_{t}\right)^{2}$.
\begin{assumption}
\label{assu: weak dependence}Suppose the following conditions hold.
\begin{enumerate}  \setlength\itemsep{0pt}
\item There exists a constant $\kappa_{1}>0$ such that for any $A\subseteq\{1,\dots,T\}$,
the largest eigenvalue of $E\left[|A|^{-1}\left(\sum_{t\in A}\tilde{X}_{t}\right)\left(\sum_{t\in A}\tilde{X}_{t}\right)'\right]$
is bounded above by $\kappa_{1}$. 
\item There exists a sequence $\rho_{T}>0$ such that $P(\max_{1\leq t\leq T}\|\tilde{X}_{t}\|_{\infty}\leq\rho_{T})\rightarrow1$. 
\item The data $\{(X_{t},\tilde{u}_{t})\}_{t=1}^{T}$ are $\beta$-mixing with coefficient
satisfying $\betamix(\gamma_{T})\rightarrow0$ for some sequence $\gamma_{T}$
satisfying $0<\gamma_{T}<r/2$ and $\rho_{T}\gamma_{T}=o_P(\min\{\sqrt{T_{0}},\sqrt{T_{1}}\})$, where $r=\min\{\left\lfloor T_{0}/K\right\rfloor,T_1 \}$.
\item $\{\tilde{u}_{t}\}_{t=1}^{T}$ satisfies $\max_{1\leq t\leq T}E|\tilde{u}_{t}|^q=O(1) $ and $\betamix(i)\lesssim i^{-\eta} $ for some constants $q> 2$ and $\eta>q/(q-2)$ and $\sigma^2>0$.
\end{enumerate}
\end{assumption}

The weak dependence is stated in terms of $\beta$-mixing, which holds for a large class of stochastic processes. 
Stationarity and $\beta$-mixing conditions are commonly imposed for time series data. These conditions are satisfied in various Markov chains and hidden Markov models, including ARMA, GARCH, and many stochastic volatility models \citep[e.g.,][]{carrasco2002mixing,meyn2012markov}. Note that the $\beta$-mixing and moment conditions in Assumption \ref{assu: weak dependence} rule out unit-root or near-unit-root processes.

The following theorem establishes the asymptotic distribution of the component estimators. Let $g_{c_0,K}=K \oneb\{ c_0<1\}+(K/c_0) \oneb\{1\leq c_0\leq K \} +  \oneb\{ c_0>K\}  $.

\begin{thm}
\label{thm: asy distr}  Let Assumptions  \ref{assu: stationary Y0 X}, \ref{ass:ell_2_consistency}, and \ref{assu: weak dependence} hold. Suppose that $T_0,T_1\rightarrow \infty$ and that $T_{0}/T_{1}\rightarrow c_{0}$ for some $c_{0}\in[0,\infty]$. Then, 
\[
\sqrt{\min\{T_0,T_1\} }\begin{pmatrix}\htau_{1}-\tau\\
\vdots\\
\htau_{K}-\tau
\end{pmatrix}\overset{d}{\rightarrow}\begin{pmatrix}\sqrt{\min\{c_{0},1\} }\xi_{0}-\sqrt{g_{c_0,K} }\xi_{1}\\
\vdots\\
\sqrt{\min\{c_{0},1\} }\xi_{0}-\sqrt{g_{c_0,K}}\xi_{K}
\end{pmatrix}\sigma,
\]
where   $\xi_{0},\dots,\xi_{K}$ are independent $N(0,1)$ random
variables.
\end{thm}

The proof of Theorem \ref{thm: asy distr} proceeds in two steps. First, observe that $\hw_{(k)}-w_\ast$ is approximately independent of the data in $H_{k}\bigcup\{T_{0}+1,\dots,T\} $ under weak dependence of the data (Assumption \ref{assu: weak dependence}). Consequently, $\ell_{2}$-consistency of $\hat{w}_{(k)}$ can be used to bound the second term in Lemma \ref{lem: algebra}, so that
	\begin{equation}
\hat{\tau}_{k}-\tau-\left(\frac{1}{T_{1}}\sum_{t=T_{0}+1}^{T}u_{t}-\frac{1}{|H_{k}|}\sum_{t\in H_{k}}u_{t}\right)=o_{P}\left(\frac{1}{\sqrt{\min\{T_0,T_1\}}} \right).\label{eq:proof_explanation_step1}
\end{equation}
Second, under stationarity, we can replace $u_t$ in  \eqref{eq:proof_explanation_step1} by $\tilde{u}_t$, which is mean-zero  by construction. As a consequence, the desired result follows from a CLT. The common component $\xi_{0}$ corresponds to the post-treatment average, and $\xi_{1},\dots,\xi_{K}$ correspond the averages over the  blocks in the pre-treatment period.

Theorem \ref{thm: asy distr} imposes no restrictions on the relative magnitude of $T_0$ and $T_1$, accommodating a wide range of applications. This feature adds useful robustness since researchers do not have to choose between different asymptotic approximations and inference procedures depending on the relative magnitude of $T_0$ and $T_1$.

The next theorem establishes the asymptotic distribution of the ATT estimator and $\mathbb{T}_K$.
\begin{thm} \label{thm: t statistic}
Let the assumptions in Theorem \ref{thm: asy distr} hold. Then (i)
\[
\sqrt{\min\{T_0,T_1\}}(\htau-\tau) \overset{d}\rightarrow N(0,(\min\{c_{0},1\} + g_{c_0,K}K^{-1} )\sigma^2)
\]
and (ii) $\mathbb{T}_K\overset{d}\rightarrow t_{K-1}$, where the random variable $t_{K-1}$ has a standard $t$-distribution with $K-1$ degrees of freedom.
\end{thm}
Part (i) of Theorem \ref{thm: t statistic} is a direct consequence of Theorem \ref{thm: asy distr} and establishes the asymptotic normality of the ATT estimator. Making inferences directly based on this result requires estimating the LRV $\sigma^2$, which is difficult in small sample settings. We therefore use the self-normalized test statistic $\mathbb{T}_K$, which allows us to avoid estimating the LRV. Part (ii) demonstrates that $\mathbb{T}_K$ has an asymptotically pivotal student $t$-distribution with $K-1$ degrees of freedom. This result is useful from a practical perspective as one does not have to simulate non-standard critical values, nor rely on subsampling or permutation distributions. 

The result in Part (ii) is different from classical results on $t$-statistics because the component estimators $\htau_{1},\dots,\htau_K$ are not asymptotically independent due to presence of the common component $\xi_0$ in the asymptotic distribution in Theorem \ref{thm: asy distr}. However, the common component cancels out in the denominator of the $t$-statistic so that numerator and denominator are asymptotically independent and $\mathbb{T}_K$ has an asymptotic $t$-distribution after rescaling the denominator to account for the presence of $\xi_0$ in the numerator.

The following corollary of Theorem \ref{thm: t statistic} formally establishes the $(1-\alpha)$ coverage guarantee of the confidence interval $\mathcal{I}_K(1-\alpha)$ defined in \eqref{eq:ci}.
\begin{cor}\label{cor:CI_PI}Let the assumptions in Theorem \ref{thm: asy distr} hold. Then, $P(\tau \in \mathcal{I}_K(1-\alpha))\rightarrow 1-\alpha.$
\end{cor}
In light of Corollary \ref{cor:CI_PI}, we can interpret $\mathcal{I}_K(1-\alpha)$ as a confidence interval if $\tau$ is fixed and as a prediction interval if $\tau$ is random. Thus, our framework provides a unified framework, encompassing both leading cases discussed in the literature.

\begin{rem}\label{rem:generic_Result}
Theorem \ref{thm: asy distr} only requires $\ell_2$-consistency of $\hw_{(k)}$ (Assumption \ref{ass:ell_2_consistency}). As a result, the $t$-test is generic and can be used in conjunction with the many other $\ell_2$-consistent estimators available in the literature (e.g., Lasso). This is similar to the double machine learning (DML) literature \citep[][]{chernozhukov2018double}, where the first-order asymptotics of the final estimator depend on the influence function in the debiasing process, but not on the properties of the estimators of the nuisance parameters, beyond a rate requirement. We emphasize that unlike the DML literature, the $t$-test does not have any rate requirements. \qed
\end{rem}

\begin{rem}\label{rem:higher_order}
The alternative fixed-$b$ approach \citep[e.g.,][]{kiefer2000simple,kiefer2002heteroskedasticity,kiefer2002et,jansson2004error,kiefer2005new,sun2008optimal} does not naturally arise from the cross-fitting procedure and encounters extra technical difficulty due to the error from estimating the high-dimensional weights. The usual justification for the fixed-$b$ approach is its higher-order improvement in Gaussian location models. Our procedure achieves the same property in this setting; see Appendix \ref{sec: higher order}.\qed
\end{rem}

\begin{rem} The sign-based randomization test of \citet{canay2017randomization} is an alternative to the conventional $t$-test \citep[e.g.,][]{ibragimovmueller2010}, provided that $K$ is large enough. In our setting, the randomization test cannot be directly applied due to the dependence between the component estimators $\htau_1,\dots,\htau_K$. However, the structure of our problem makes it possible to decorrelate the component estimators without estimating the LRV, so that the randomization test can be applied after decorrelation. Specifically, Theorem \ref{thm: asy distr} shows that $\sqrt{\min\{T_0,T_1\} }(\htau_{1}-\tau,\dots,\htau_{K}-\tau)'\overset{d}\rightarrow N(0,\Sigma)$, where $\Sigma=\sigma^2\tilde{\Sigma}$ with $\tilde{\Sigma}_{kk}=\min\{c_0,1\}+g_{c_0,K}$ and $\tilde{\Sigma}_{kl}=\min\{c_0,1\}$ for $k\ne l$. Because $\tilde{\Sigma}$ does not depend on $\sigma$, and $\sigma$ enters as a multiplicative factor, one can decorrelate $\htau_1,\dots,\htau_K$ without estimating the LRV.\footnote{\citet{canay2017randomization} show that the sign-based randomization test outperforms the $t$-test when there is heterogeneity in the variances of the component estimators. However, in the presence of such heterogeneity, decorrelating $\htau_1,\dots,\htau_K$ would require estimating the LRV for each component estimator, which is difficult in small samples.}
\qed
\end{rem}

\subsection{Choosing $K$}
\label{sec: choosing K}
The test statistic $\TT_K$ and its limiting distribution depend on $K$. Choosing $K$ seems unavoidable; it is inherent to the cross-fitting procedure required for bias correction. The choice of $K$ is subject to a trade-off between the expected length of the confidence intervals and their finite sample coverage properties. 

Figure \ref{fig:cov_leng} illustrates this trade-off. It shows that choosing a larger $K$ will lead to shorter confidence intervals but may impact the coverage accuracy of our method. The reason for the loss in coverage accuracy is that we are constructing the bias correction term in \eqref{eq:tau_hat} by averaging over a shorter time period, which lowers the quality of the normal approximation in Theorem \ref{thm: asy distr} and leads to increased finite sample dependence between the blocks. See, for example, \citet[][Section 2.3]{ibragimovmueller2010} for a related discussion in the context of the standard $t$-test with dependent data.

\begin{figure}[h!]
\caption{Trade-off between coverage accuracy and length when choosing $K$}
\label{fig:cov_leng}
\begin{center}
    \begin{subfigure}[b]{0.495\textwidth}
         \centering
	\includegraphics[width=\textwidth,trim = {0 3cm 0 2.5cm}]{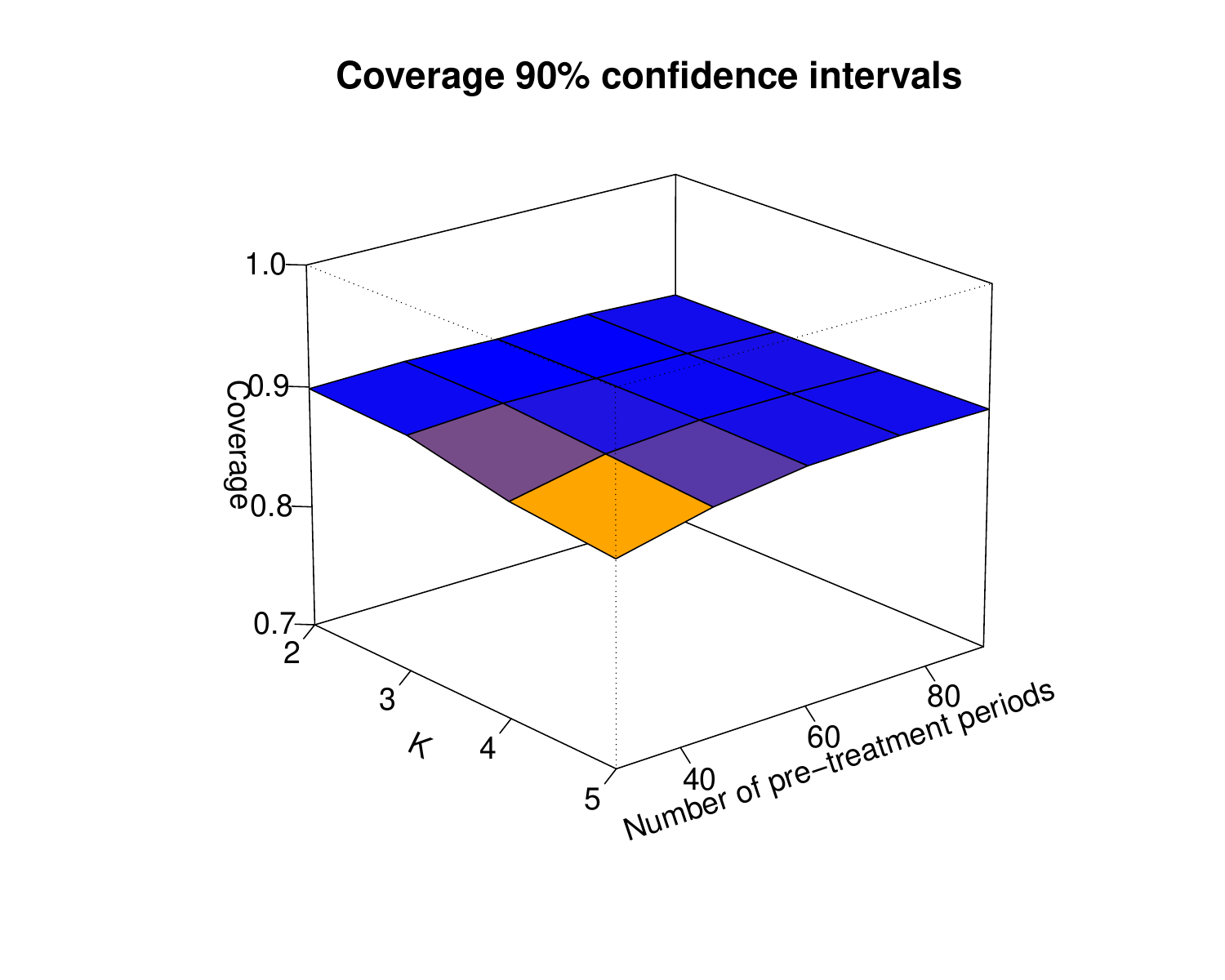}
    \end{subfigure}
    \begin{subfigure}[b]{0.495\textwidth}
         \centering
	\includegraphics[width=\textwidth,trim = {0 3cm 0 2.5cm}]{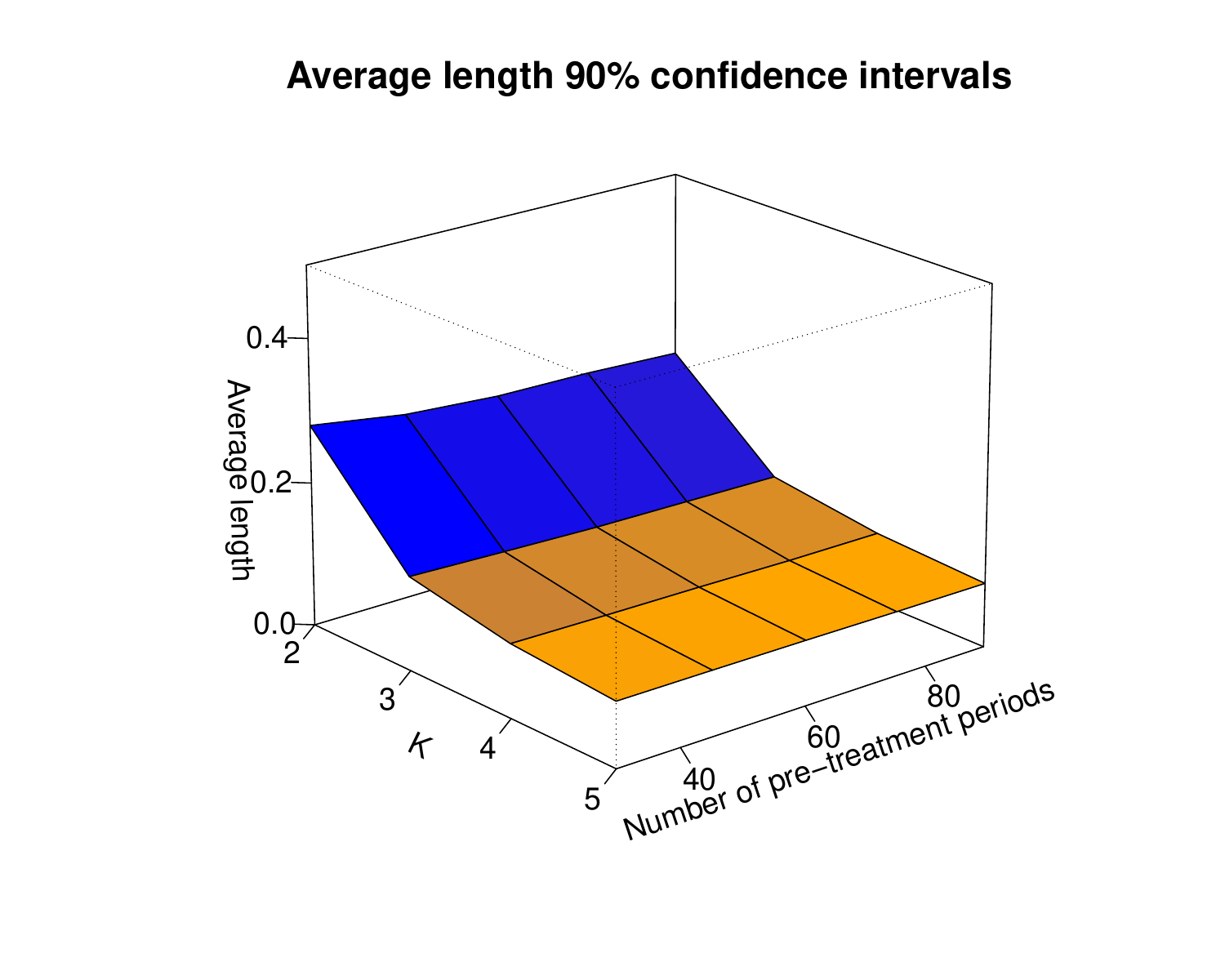}
    \end{subfigure}
\end{center}
    
\scriptsize{\textit{Notes:} Simulations with 20,000 repetitions based on DGP1 described in Section \ref{sec:simulations}. Because the trade-off is more relevant and pronounced when the persistence is higher, we show results for $\rho_u=0.6$ for the purpose of illustration ($\rho_u=0.31$ in Section \ref{sec:simulations}). Nominal coverage: $1-\alpha=0.9$. We set $(T_1,N)=(16,14)$ as in the empirical application and vary $T_0$ and $K$.}	
\end{figure}

To formalize the trade-off between coverage accuracy and length, it is helpful to analyze the asymptotic efficiency of the confidence intervals in \eqref{eq:ci} by comparing the expected asymptotic length for a fixed $K$ to the limiting case where $K\rightarrow \infty$. While our theory requires $K$ to be fixed, the case where $K\rightarrow \infty$ is a useful theoretical benchmark.

The length of the confidence interval \eqref{eq:ci} is $L=t_{K-1}(1-\alpha/2)\hat\sigma_{\htau}/\sqrt{K}$. In Appendix \ref{app: expected length}, we show that $\sqrt{\min\{T_0,T_1\}}L\overset{d}\rightarrow\mathcal{L}$, with 
\begin{equation}
E(\mathcal{L})=2 \sigma t_{K-1}(1-\alpha/2)\frac{1}{\sqrt{K}\sqrt{K-1}}\sqrt{1+\min \{c_0,K \}} \sqrt{g_{c_0,K}}\sqrt{2}\frac{\Gamma(K/2)}{\Gamma((K-1)/2)}, \label{eq: expected length}
\end{equation}
where $\Gamma(\cdot)$ denotes the Gamma function. Using Stirlings's approximation of the Gamma function, we obtain the following limiting length as $K\rightarrow \infty$,
\begin{equation}
E(\mathcal{L})\approx 2\sigma \Phi^{-1}(1-\alpha/2)\sqrt{\min\{c_{0}^{-1},1\}} \sqrt{1+c_0 },\label{eq: limiting length}
\end{equation}
where $\Phi^{-1}(\cdot)$ is the quantile function of the standard normal distribution.
See Appendix \ref{app: expected length} for detailed derivations. The relative asymptotic efficiency (RAE) can then be computed as the ratio of \eqref{eq: limiting length} and \eqref{eq: expected length},
\begin{equation}
\operatorname{RAE}=\frac{\Phi^{-1}(1-\alpha/2)\sqrt{\min\{c_{0}^{-1},1\}} \sqrt{1+c_0 }}{ t_{K-1}(1-\alpha/2)\frac{1}{\sqrt{K}\sqrt{K-1}}\sqrt{1+\min \{c_0,K \}} \sqrt{g_{c_0,K}}\sqrt{2}\frac{\Gamma(K/2)}{\Gamma((K-1)/2)}},\label{eq:rae}
\end{equation}
and is a function of $\alpha$, $c_0$, and $K$ only. Table \ref{tab: rae} shows the RAE for $c_0=T_0/T_1=30/16$, as in the empirical application in Section \ref{sec:application}. 
\begin{table}[h!] 
\footnotesize

\centering
\caption{Relative asymptotic efficiency}
\label{tab: rae} 
\begin{tabular}{lccccccccc}
\toprule
\midrule
$K$&2&3&4&5&6&7&8&9&10\\
\midrule
RAE&32.65 & 63.56 & 75.86 & 82.08 & 85.79 & 88.23 & 89.97 & 91.26 & 92.25 \\ 
\midrule
\bottomrule
\multicolumn{10}{p{11.7cm}}{\scriptsize{\it Notes:} RAE: relative asymptotic efficiency as defined in \eqref{eq:rae}. Results are shown for $\alpha=0.1$ and $c_0=30/16$, as in the empirical application.}
\end{tabular}
\end{table}
\normalsize

Table \ref{tab: rae} shows that choosing $K=3$ almost doubles the RAE relative to $K=2$. However, this choice can be conservative: increasing $K$ to $4$, $5$, or even $6$ improves RAE by 12, 19, and 22 percentage points, respectively. The RAE gains from increasing $K$ further are decreasing. Importantly from a practical perspective, one can achieve a very high RAE without choosing a large $K$: choosing $K= 9$ already yields confidence intervals with more than 90\% RAE.

The RAE formula captures one side of the trade-off. The other side of the trade-off, coverage accuracy, can be assessed using application-based simulations. The simulations reported in this paper (based on \citet{andersson2019carbon}) and in an earlier version (\citet{chernozhukov2022arxiv7}, based on \citet{abadie2003economic}) suggest that $K = 3$ works well when $T_0$ is small. When $T_0$ is moderate or large, the coverage accuracy remains excellent even for larger values of $K$, which justifies choosing larger values of $K$. 

An important determinant of the coverage accuracy is the degree of persistence in the prediction errors $\{u_t\}$: the more persistent the prediction errors, the lower the coverage accuracy. A simple approach for gauging the persistence in $\{u_t\}$ is to fit an AR(1) model to the SC residuals in the pre-treatment period, as we do in Section \ref{sec:simulations}. Assessing the persistence in this way is a helpful first step in evaluating coverage accuracy and an essential input for application-based simulations.

In summary, $K=3$ is a useful starting point in typical SC applications where $T_0$ is small. If the estimated persistence in $\{u_t\}$ is relatively low and simulations indicate a good coverage accuracy, choosing $K=4$ improves the efficiency of the $t$-test. If $T_0$ is moderate or large, $K$ can often be chosen to achieve 80\% RAE or even 90\% RAE without affecting the coverage properties too much.

\subsection{Debiased SC is more efficient than DID}
\label{app:efficiency}

Here we show that debiased SC is more efficient than DID.\footnote{The recent work by \citet{bottmer2021design} complements our results by studying the efficiency of SC methods under design-based uncertainty.} 
To illustrate, suppose that $T_0<T_1$. The DID estimator of the ATT  can be written as \citep[e.g.,][]{DI16},
\[
\hat\tau^{\text{DID}}=\frac{1}{T_{1}}\sum_{t=T_0+1}^{T}\left(Y_{t}-X_t'w_{\text{DID}} \right)-\frac{1}{T_{0}}\sum_{t=1}^{T_0}\left(Y_{t}-X_t'w_{\text{DID}}\right),
\]
where $w_{\text{DID}}:=\left(\frac{1}{N},\dots,\frac{1}{N}\right)'$. For simplicity, suppose that $E(Y_t(0))=0$ and $E(X_t)=0$ and assume that the data are iid. Then a CLT gives 
\[
\sqrt{T_0}(\hat\tau^{\text{DID}} -\tau)\overset{d}\rightarrow N(0,(c_0+1)\sigma_{\text{DID}}^2),\quad \sigma^2_{\text{DID}}=E\left(Y_{t}(0)-X_t'w_{\text{DID}}  \right)^2.
\]
For SC, the previous results imply
\[
\sqrt{T_0}(\hat\tau -\tau)\overset{d}\rightarrow N(0,(c_0+1)\sigma_{SC}^2),\quad \sigma^2_{SC}=E(Y_{t}(0)-X_t'w_{*})^2,
\]
where $w_{*}=\arg\min_{v\in \mathcal{W}^{SC}}\ E(Y_{t}(0)-X_{t}'w)^{2}$. To show that the SC estimator is more efficient than DID, it suffices to note that \[
\sigma_{*}^2=E(Y_{t}(0)-X_t'w_{*})^2=\min_{w\in \mathcal{W}^{SC}}E(Y_{t}(0)-X_t'w)^2\leq E\left(Y_{t}(0)-X_t'w_{\text{DID}}  \right)^2=\sigma_{\text{DID}}^2.
\]
The above inequality shows that the magnitude of the efficiency gain is determined by $\sigma_{\text{DID}}^2-\sigma_{*}^2$, which depends on the true data-generating process. We note that the results in this section are asymptotic. In small samples, the uncertainty from estimating the weights can mask the asymptotic efficiency improvements, especially when these efficiency improvements are relatively small, i.e., when $\sigma_{\text{DID}}^2-\sigma_{*}^2$ is small.

\section{Theoretical properties with non-stationary data}
\label{sec:nonstationarity}
In many applications the data are non-stationary. An important rationale for applying SC is that the precise structure of the nonstationarity is unknown such that one has to rely on controls to identify counterfactual trends. Therefore, unlike the existing literature \citep[e.g.,][]{li2020statistical}, we do not make assumptions on the exact form of nonstationarity and instead restrict the heterogeneity in the nonstationarity across units. In essence, we require the units to be sufficiently homogeneous, which is an important contextual requirement in SC studies \citep[][]{abadie2021using}.

\subsection{Unrestricted common nonstationarity}
\label{sec:common_trend}
We start by establishing the validity of our procedure under the following general class of nonstationary processes. As in Section \ref{sec: theory stationary}, we allow for arbitrary misspecification.

\begin{assumption} \label{ass:common_nonstationarity} For $1\le t\le T$, $Y_t(0)=V_t(0)+\theta_t$ and $X_{t}=Z_t+\mathbf{1}_N\theta_t$, where $\theta_t\in \mathbb{R}$ is an unrestricted stochastic process and $\{(V_t(0),Z_t)\}_{t=1}^T$ is covariance-stationary.
\end{assumption}
Assumption \ref{ass:common_nonstationarity} requires that the potential outcomes for the treated unit and all control units can be decomposed into a common non-stationary component $\theta_t$ and a stationary component. Importantly, $\{\theta_t\}$ can be an arbitrary stochastic process (e.g., a deterministic trend, random walk, or general ARIMA process) such that researchers are not required to impose further restrictions on the nonstationarity as long as it is shared among units. This feature adds very useful robustness. Even for low-dimensional problems it is well-known that the slightest misspecification in modeling nonstationarity (e.g., unit root vs. near unit root) can yield invalid inferences.\footnote{For example, there is a large literature on making inference on the scalar parameter $\beta$ in $Y_t=X_t \beta +U_t$, where $X_t=\rho X_{t-1}+Z_t $ and $\rho=1+c/T$. The key difficulty is that the asymptotic distribution of the usual $t$-statistic for $\beta$ depends on  $c$, but $c$ cannot be consistently estimated \citep[e.g.,][]{phillips2014confidence}. Hence, using a unit root process ($c=0$) for the dynamics of $X_t$ is not robust to barely detectable misspecification (e.g., local-to-unity $c<0$).} 

To derive the asymptotic properties of our method under Assumption \ref{ass:common_nonstationarity}, we exploit the specific structure of the SC estimator and the constraints on the weights.
First, the estimated and the pseudo-true SC weights satisfy $\mathbf{1}_N'\hw_{(k)}=\mathbf{1}_N'w_\ast=1$. Therefore, under Assumption \ref{ass:common_nonstationarity}, we have a version of Lemma \ref{lem: algebra} with $\tilde{X}_t$ replaced by $\tilde{Z}_t$, which is stationary.

Second, $\hw_{(k)}$ is $\ell_2$-consistent under Assumption \ref{ass:common_nonstationarity}. Since $\mathcal{W}^{SC}\subset \{w:\mathbf{1}_N'w=1\}$, the estimated and the pseudo-true SC weights can be expressed in terms of the stationary parts of the potential outcomes, $(V_t(0),Z_t)$:
\begin{eqnarray}
&&\hat{w}_{(k)}=\underset{w\in \mathcal{W}^{SC} }{\arg\min}\sum_{t\in H_{(-k)}}(Y_{t}(0)-X_{t}'w)^{2} =\underset{w\in \mathcal{W}^{SC} }{\arg\min}\sum_{t\in H_{(-k)}}(V_t(0)-Z_{t}'w)^{2}\label{eq:sc_estimator_common_nonstationarity}\\
&&w_{*}:=\underset{w\in \mathcal{W}^{SC} }{\arg\min}\ E(Y_{t}(0)-X_{t}'w)^{2}=\underset{w\in \mathcal{W}^{SC} }{\arg\min}\ E(V_{t}(0)-Z_{t}'w)^{2}\label{eq:w_ast_common_nonstationarity}
\end{eqnarray}
Therefore, if the conditions in Lemma \ref{lem: consistency misspecification} hold with $(Y_t(0),X_t)$ replaced by $(V_t(0),Z_t)$, consistency of $\hat{w}_{(k)}$ follows. The following lemma states the formal result.
\begin{lem} \label{lem:consistency_common_nonstationarity}Suppose that the assumptions of Lemma \ref{lem: consistency misspecification} hold with $(Y_t(0),X_t)$ replaced by $(V_t(0),Z_t)$. Then Assumption \ref{ass:ell_2_consistency} holds.
\end{lem}

Finally, by \eqref{eq:sc_estimator_common_nonstationarity}, $\hw_{(k)}$ is only a function of $\{Z_t\}_{t\in H_{(-k)}}$ under Assumption \ref{ass:common_nonstationarity}. Thus, under weak dependence of $\{Z_t\}$, $\hw_{(k)}-w$ is approximately independent of $\{Z_t\}_{t\in H_{k}\bigcup\{T_{0}+1,\dots,T\} }$, so that the asymptotic properties of our method can be established using similar arguments as in the proofs of Theorems \ref{thm: asy distr} and \ref{thm: t statistic}.

\begin{thm}\label{thm: common nonstationarity}
Let Assumptions \ref{ass:ell_2_consistency}, \ref{assu: weak dependence} with $X_t$ replaced by $Z_t$, and \ref{ass:common_nonstationarity} hold. Suppose that $T_{0}/T_{1}\rightarrow c_{0}$ for some $c_{0}\in[0,\infty]$. Then (i)
%\[
%\sqrt{\min\{T_0,T_1\}}(\htau-\tau) \overset{d}\rightarrow N(0,(\min\{c_{0},1\} + g_{c_0,K}K^{-1} )\sigma^2)
%\]
$\mathbb{T}_K\overset{d}\rightarrow t_{K-1}$ and (ii) $P(\tau \in \mathcal{I}_K(1-\alpha))\rightarrow 1-\alpha.$
\end{thm}

Theorem \ref{thm: common nonstationarity} formally guarantees the validity of the confidence interval $\mathcal{I}_K(1-\alpha)$ when the data exhibit a common unrestricted nonstationarity.

\subsection{Deviations from common nonstationarity}
\label{sec:sparse_deviations}

Here we study a more general setting that allows for deviations from the common nonstationarity. Unlike in the previous sections, we do not allow SC to be misspecified. Consider the following assumption.

\begin{assumption}
	\label{assu: non-stationary DGP }Assume that 	$X_{t}=\oneb_{N}\theta_{t}+Z_{t}+\beta\xi_{t}$, $1\leq t\leq T$,
	where $\theta_{t}$ is an unrestricted stochastic process, $E(Z_{t})=0$, $\xi_{t}$ is a non-stationary process, $\beta\in\RR^{N}$ is an arbitrary
	vector. Assume that $\{\theta_{t}\}_{t=1}^{T}$,
	$\{Z_{t}\}_{t=1}^{T}$, and $\{\xi_{t}\}_{t=1}^{T}$ are mutually independent.
	Moreover, $Y_{t}(0)=X_{t}'w+u_{t}$, $1\leq t\leq T$,
	where $E(u_{t})=0$, $w\in \mathcal{W}^{SC}$.
\end{assumption}
Assumption \ref{assu: non-stationary DGP } considers deviations
from Assumption \ref{ass:common_nonstationarity}. The deviations are driven by the non-stationary process $\xi_t$. Assumption \ref{assu: non-stationary DGP } does not impose any restrictions on the vector $\beta$. This allows us to accommodate many leading examples of non-stationary data. For example, the control units can have different deterministic trends. Moreover, Assumption \ref{assu: non-stationary DGP } allows for any cointegration system whose nonstationarity is driven by $\theta_t$ and $\xi_t$. 
The case with multiple $\xi_t$'s would be much more complicated since the analysis would need to take into account the correlation among them and whether they diverge at the same rate; we leave this extension for future research.

The following assumption restricts the relative magnitude of $T_0$, $T_1$, and $N$, as well as the magnitude of the deviation process $\xi_{t}$. 
\begin{assumption}
	\label{assu: non-stationary deviation}Assume $T_{1}\ll T_{0}/\log N$.
Suppose that 
\[
\frac{\sum_{t=T_{0}+1}^{T}\xi_{t}}{\sqrt{\sum_{t\in H_{(-k)}}\xi_{t}^{2}}}=O_{P}(T_{1}T^{-1/2})\qquad\text{and}\qquad\frac{\sum_{t\in H_{k}}\xi_{t}}{\sqrt{\sum_{t\in H_{(-k)}}\xi_{t}^{2}}}=O_{P}(T_{1}T^{-1/2}).
\]
\end{assumption}
With a common nonstationarity, no restrictions on the relative magnitude of $T_0$ and $T_1$ are required (Section \ref{sec:common_trend}). However, to accommodate deviations form a common nonstationarity, we require $T_0$ to be much larger than $T_1$.
The restriction on the magnitude of  $\xi_{t}$ can be verified for many non-stationary processes including unit roots and polynomial trends; see Appendix \ref{app:verification} for details.

Assumptions \ref{assu: non-stationary DGP } and \ref{assu: non-stationary deviation} allow for general non-stationary trends and deviations without imposing specific structures on the dynamics. This is important in practice because procedures that rely on specific assumptions on the nonstationarity are typically not robust against deviations from these assumptions, and even the slightest misspecification can yield invalid inferences, as discussed above.

Finally, we impose conditions on the stationary component $Z_{t}$.
\begin{assumption}
	\label{assu: nonstat reg}Let Assumption   \ref{assu: weak dependence}  hold with $X_{t}$ replaced
	by $Z_{t}$. In addition, suppose that the following conditions hold.
	\begin{enumerate} \setlength\itemsep{0pt}
	\item $\|\sum_{t\in H_{(-k)}}Z_{t}u_{t}\|_{\infty}=O_{P}(\sqrt{T_0\log N})$
	\item $\|\sum_{t\in H_{(-k)}}Z_{t}\xi_{t}\|_{\infty}/\sqrt{\sum_{t\in H_{(-k)}}\xi_{t}^{2}}=O_{P}(\sqrt{\log N})$
	\item $(\sum_{t\in H_{(-k)}}u_{t}\xi_{t})/\sqrt{\sum_{t\in H_{(-k)}}\xi_{t}^{2}}=O_{P}(1)$
	\item $\|\sum_{t\in H_{(-k)}}(Z_{t}Z_{t}'-EZ_{t}Z_{t}')\|_{\infty}=O_{P}(\sqrt{T_0\log N})$
	\end{enumerate}
\end{assumption}

The next theorem establishes the asymptotic validity of the proposed inference method when there are deviations from a common non-stationarity.
\begin{thm}
	\label{thm: nonstat deviation}Let  Assumptions \ref{assu: non-stationary DGP },
	 \ref{assu: non-stationary deviation}, and \ref{assu: nonstat reg}  hold and suppose that we use the SC estimator \eqref{eq:sc_estimator}. Then, for $1\leq k\le K$,
	\[
	\sqrt{T_{1}}(\htau_{k}-\tau)=\frac{1}{\sqrt{T_1}}\sum_{t=T_{0}+1}^{T}u_{t}-\frac{1}{\sqrt{|H_{k}|}}\sum_{t\in H_{k}}u_{t}+o_{P}(1).
	\]
	Moreover, $\mathbb{T}_K\overset{d}\rightarrow t_{K-1}$ and  $P(\tau \in \mathcal{I}_K(1-\alpha))\rightarrow 1-\alpha.$
\end{thm}

To our knowledge, Theorem \ref{thm: nonstat deviation} guarantees the inference validity of SC under the most general available conditions without imposing a specific structure on the dynamics of the nonstationarity.

\begin{rem} DID is generally inconsistent under  deviations from a common nonstationarity. For example, suppose that there is a linear trend ($\xi_t=t$). The DID estimator would have a bias of $T_1(T_0+(T_1+1)/2)(w'\beta-N^{-1}\oneb_N' \beta)$. Thus, DID is inconsistent unless $w'\beta-N^{-1}\oneb_N' \beta=o(T^{-1}T_{1}^{-1}) $. \qed
\end{rem}

\section{Main extensions}
In this section, we present two main extensions. See Appendix \ref{app: additional extensions} for additional extensions.
\subsection{Time-varying predictive relationship} \label{sec: time-varying weights}

In the previous sections, we assumed that the predictive relationship between the treated and the control units is time-invariant, such that the SC weights do not change over time. Here, we discuss settings with time-varying weights.

\smallskip 

\noindent
\textbf{The $t$-test is valid with stationary weights.} The $t$-test accommodates certain forms of time-varying weights. Consider the following variant of model \eqref{eq:prediction_model} with time-varying weights,
\begin{equation}
Y_t(0)=X_t'w_t+u_t,\quad 1\le t\le T.\label{eq:prediction_model_time_varying_weights}
\end{equation}
Suppose that the time-varying weights satisfy the SC constraints, $w_t\in \mathcal{W}^{SC}$, $1\le t\le T$. 
In Section \ref{sec:common_trend}, we assume $Y_t(0)=V_t(0)+\theta_t$ and $X_t=\oneb_N \theta_t +Z_t$, where $\{(V_t(0),Z_t)\}_{t=1}^T$ is covariance-stationary. The model \eqref{eq:prediction_model_time_varying_weights} implies that $V_t(0)=u_t+X_t'w_t-\theta_t=u_t+Z_t'w_t$ since $w_t\in \mathcal{W}^{SC}$. Thus, as long as $\{(V_t(0),Z_t)\}_{t=1}^T$ is stationary, the results in Section \ref{sec:common_trend} apply and the $t$-test remains valid with time-varying weights. 

Stationarity of $\{(V_t(0),Z_t)\}_{t=1}^T$ is a high-level condition for the validity of the $t$-test that can be difficult to interpret in practice because of the composite nature of $V_t(0)$. A more primitive and interpretable sufficient condition for the validity of the $t$-test is stationarity of $\{(Z_t,w_t,u_t)\}_{t=1}^T$. When the weights are time-invariant ($w_t=w_\ast$ for all $t$), stationarity of $\{(Z_t,u_t)\}_{t=1}^T$ directly implies stationarity of $\{(V_t(0),Z_t)\}_{t=1}^T$, since $\{(V_t(0),Z_t)\}_{t=1}^T=\{((u_t+Z_t'w_\ast),Z_t)\}_{t=1}^T$. Therefore, the main additional requirement for the validity of the $t$-test in settings with time-varying predictive relationships is the stationarity of the weights $\{w_t\}_{t=1}^T$. 

To interpret the stationarity of the weights, it is helpful to relate our setting to time-varying coefficient models. Time-varying coefficient models typically have additional assumptions on the nature of the variation. For example, in latent large factor models, the factors or the factor loadings can be assumed to be nonparametric functions of observed variables \citep[e.g.,][]{connor2007semiparametric,connor2012efficient,fan2016projected,fan2021augmented}. We can consider a similar situation for the weights. In particular, it might be reasonable to assume that $w_t=f(Z_t,e_t)$ for some fixed unknown function $f(\cdot,\cdot)$, where $\{(Z_t,e_t,u_t)\}_{t=1}^{T}$ is stationary. Under this assumption, $\{w_t\}_{t=1}^T$ and $\{(V_t(0),Z_t)\}_{t=1}^T$ are stationary, and the theoretical results in Section \ref{sec:common_trend} apply. 

\smallskip 

\noindent
\textbf{Time-variability in the weights and the length of the confidence intervals.}
The variability in the weights over time affects the length of the confidence intervals obtained from the $t$-test. To see this, rewrite model \eqref{eq:prediction_model_time_varying_weights} in terms of the pseudo-true weights $w_\ast$ as 
\begin{equation*}
Y_t(0)=X_t'w_\ast+v_t,\quad v_t:=X_t'(w_t-w_\ast)+u_t.
\end{equation*}
This shows that the variability in the weights around $w_\ast$ and their persistence affect the LRV of the prediction errors and thus the length of the confidence intervals. 

In Appendix \ref{app: time-varying weights}, we provide simulation evidence illustrating the impact  of time-variability in $w_t$ on the length of the confidence intervals. We consider a scenario where $w_t\in \mathcal{W}^{SC}$ follows a Markov switching process. Consistent with our theory, more persistence (a larger probability of staying in the same regime) leads to a larger LRV and  wider confidence intervals. We also consider a scenario with misspecification ($w_t\notin \mathcal{W}^{SC}$) and show that the increase in length is even more pronounced in this case due to the variability in $w_t$ not being restricted by the SC constraints.

\smallskip 

\noindent
\textbf{Non-stationary weights.} So far, we have discussed the implications when $w_t$ is stationary. When $w_t$ does not vary in a stationary manner, one might question the suitability of SC methods in the application of interest. In this case, we might expect the fit of the SC model to change across time in a systemic pattern. One formal procedure to check this is the following placebo test. 
\begin{enumerate}\setlength\itemsep{0pt}
\item Split the pre-treatment period into a placebo pre-treatment period, $\{1,\dots,\widetilde{T}_0\}$, and a placebo post-treatment period $\{\widetilde{T}_0+1,\dots,T_0\}$.
\item Apply the $t$-test, treating $\{1,\dots,\widetilde{T}_0\}$ as the pre-treatment period and $\{\widetilde{T}_0+1,\dots,T_0\}$ as the post-treatment period. 
\item Reject if zero is not in the confidence interval.
\end{enumerate}
Maintaining the other assumptions underlying the $t$-test, a rejection  indicates that the weights are time-varying in ways that render the $t$-test invalid. More generally, this placebo test is an omnibus specification test for the joint validity of all the assumptions underlying the $t$-test. We caution that while non-rejections provide evidence in favor of the applicability of the $t$-test, such non-rejections do not imply that the $t$-test is valid.

In general, extrapolation is unavoidable in SC settings. Therefore, when the placebo test rejects the null of no effect, researchers need to make assumptions on how the weights change over time that allow for such extrapolation. For example, SDID accommodates a different type of time-varying weights than the $t$-test (via $\hat{w}^{\text{SDID}}_i\hat\lambda^{\text{SDID}}_t$ in \eqref{eq:sdid}). In applications where the weights are time-varying in ways that are neither covered above nor by SDID, one could impose explicit models on how the weights change over time. We leave the development of such methods for future research.

\subsection{Random treatment status}

In the previous sections, we considered a setting where the treatment status is fixed. In Appendix \ref{app:random_treatment}, which we summarize here, we establish the validity of the $t$-test when the treatment status is random. We consider a setting with one treated unit in which the identity of the treated unit is random. Treatment assignment is typically not unconditionally randomized in SC applications. Therefore, we allow the treatment assignment to depend on observed and latent variables.

First, we consider a setting where the treatment assignment depends on the outcomes in the pre-treatment period and potentially other observables. We show that the $t$-test is valid in such settings, provided that selection into treatment only depends on ``recent'' pre-treatment outcomes but not on ``ancient history.'' This assumption is reasonable in settings where treatment adoption is the result of recent developments and changes in outcomes.

Second, we consider a setting in which the potential outcomes are determined by a factor model and selection is based on the unit-specific time-invariant latent components in the factor structure and additional selection-specific idiosyncratic shocks. The $t$-test is valid in this setting because (we show) the assumptions in the previous sections hold conditional on the treatment assignment.

\section{Application-based simulations}
\label{sec:simulations}

It is important to understand the advantages and limitations of the $t$-test in practice. Here we present simulation evidence demonstrating when it works well and when it does not. We calibrate the simulations to the empirical application in Section \ref{sec:application}, where we revisit the analysis of the effect of carbon taxes on emissions in \citet{andersson2019carbon}. We show results for $(T_0,T_1,N)=(30,16,14)$ as in \citet{andersson2019carbon}. In Appendix \ref{app:additional_simulations}, we further investigate the performance of the $t$-test with $T_0=150$, because the theoretical results in Section \ref{sec:sparse_deviations} require $T_0$ to be much larger than $T_1$, and the performance with $T_1\in \{10,12,14,16\}$. All simulations were carried out in \texttt{R} \citep{R2023}.

We generate the treated outcome as
\[
Y_{t}(0)=\mu+X_t'w+u_t, ~u_{t}=\rho_uu_{t-1}+v_t, ~v_t\overset{iid}\sim N(0,\sigma^2_{v}),~ 1\le t\le T,
\]
and set $\alpha_t=\tau=0$. The parameters $(\mu,w)$ vary across DGPs, and $\sigma^2_{v}$ and $\rho_u$ are obtained by fitting an AR(1) model to the empirical SC residuals. The degree of persistence in $\{u_t\}$ is relatively low: $\rho_u=0.31$. To generate the control outcomes, we fit a factor model with four factors to the detrended data and let
$Y_{it}(0)=\theta_{it}+L_i'F_t+\eta_{it}$, 
where $\theta_{it}$ is a unit-specific non-stationary component that varies across DGPs, $F_{t}=(F_{1t},\dots, F_{4t})'$, $F_{st}\overset{iid}\sim N(0,\sigma^2_{F_s})$, $\eta_{it}=\rho_i\eta_{it-1}+\epsilon_{it}$, and $\epsilon_{it}\sim N(0,\sigma_{\epsilon_i}^2)$. 
$(L_1,\dots,L_N)$, $(\sigma^2_{F_1},\dots,\sigma^2_{F_4})$, $(\rho_1,\dots,\rho_N)$,  and $(\sigma^2_{\epsilon_1},\dots,\sigma^2_{\epsilon_N})$ are obtained from and estimated based on the factor model fitted to the data. 

We consider three stationary and six non-stationary DGPs; see Table \ref{tab:dgps}.
DGP1--DGP3 satisfy the assumptions in Section \ref{sec: theory stationary}. DGP4--DGP5 satisfy the conditions in Section \ref{sec:common_trend} since $\mathbf{1}_N'w^{\text{MIS}}=1$, and DGP6--DGP7 satisfy the assumptions in Section \ref{sec:sparse_deviations}. DGP8 and DGP9 allow for studying the performance of the $t$-test in settings not covered by our theory. While the nonstationarity in DGP8 satisfies Assumption \ref{assu: non-stationary DGP }, which allows for unit-specific trends, our theory does not allow for misspecification when there are deviations from a common nonstationarity. DGP9 captures a setting where the deviations are driven by multiple non-stationary processes, whereas Assumption \ref{assu: non-stationary DGP } requires $\{\xi_t\}$ to be a scalar-valued process.
\begin{table}[H]
\centering
\caption{DGPs}
\footnotesize
\begin{tabular}{l l l}
\toprule
\midrule
\multicolumn{3}{c}{Stationary DGPs}\\
\midrule
DGP1&$\mu=0$, $w=w^{\text{SC}}$, $\theta_{it}=0$& SC weights\\
DGP2&$\mu=0$, $w=w^{\text{DID}}$, $\theta_{it}=0$& DID weights\\
DGP3&$\mu=2$, $w=w^{\text{MIS}}$, $\theta_{it}=0$& Misspecified\\
\midrule
\multicolumn{3}{c}{Non-stationary DGPs}\\
\midrule
DGP4 & $\mu=2$, $w=w^{\text{MIS}}$, $\theta_{it}=t$ &Common linear trend \& misspecification\\ 
DGP5 & $\mu=2$, $w=w^{\text{MIS}}$, $\theta_{it}=\theta_{t}$, $\theta_t=\theta_{t-1}+\xi_t, \xi_t\overset{iid}\sim N(0,1)$&Common random walk \& misspecification\\ 
DGP6 & $\mu=0$, $w=w^{SC}$, $\theta_{it}=t+ \mathbf{1}\left\{i=1\right\} \cdot t$& Common trend \& deviation\\
DGP7 & $\mu=0$, $w=w^{SC}$, $\theta_{it}=\theta_{1t}+\mathbf{1}\left\{i=1\right\}\cdot \theta_{2t}$,& Common random walk \& deviation\\ 
&$\theta_{1t}=\theta_{1t-1}+\xi_{1t}, \xi_{1t}\overset{iid}\sim N(0,1)$,\\ 
&$\theta_{2t}=\theta_{2t-1}+\xi_{2t}, \xi_{2t}\overset{iid}\sim N(0,1)$&\\
DGP8&$\mu=0$, $w=w^{\text{MIS}}$, $\theta_{it}=i+i\cdot t$&Heterogeneous trends \& misspecification\\
DGP9&$\mu=0$, $w=w^{\text{SC}}$, $\theta_{it}=i+\theta_{it-1}+\xi_{it}, \xi_{it}\overset{iid}\sim N(0,1)$& Random walks with heterogeneous drifts\\
\midrule
\bottomrule
\multicolumn{3}{p{16cm}}{\scriptsize{\it Notes:} $w^{\text{SC}}$: SC estimate of $w$. $w^{\text{DID}}=(1/N,\dots,1/N)'$, $w^{\text{MIS}}=(-3,3,1,0,\dots,0)'$ (such that $\mathbf{1}_N'w^{\text{MIS}}=1$, which ensures that DGP4--DGP5 satisfy the assumptions in Section \ref{sec:common_trend}).}
\end{tabular}
\label{tab:dgps}
\end{table}

We compare the proposed $t$-test to four alternative methods: (i) DID with $K$-fold cross-fitting, (ii) subsampling inference based on $\htau^{\text{SC}}$ \citep{li2020statistical},\footnote{Our \texttt{R}-implementation is based on \texttt{Matlab} code obtained from Kathleen Li. We take the subsample size to be $(2/3)\cdot T_0$, corresponding to the middle choice of $m=60$ in Section 5 of \citet{li2020statistical}.} (iii) SDID \citep{arkhangelsky2021synthetic} implemented using the \texttt{R}-package \texttt{synthdid} \citep{hirshberg2021synthdid}, and (iv) $K$-fold cross-fitting based on the true $w$ (Oracle). In Appendix \ref{app:comparison_permutation}, we also compare the $t$-test to the widely-used permutation approach by \citet{abadie10sc}.

Table \ref{tab:empirical_mc} shows the bias, coverage, and average length for all methods and DGPs. The nominal coverage is $1-\alpha=0.9$. Under stationarity, our $t$-test exhibits an excellent performance. Under correct specification, bias, coverage, and average length of the confidence intervals are similar to the Oracle. Consistent with our theory, cross-fitting effectively removes the bias under misspecification. Misspecification does not affect coverage accuracy but leads to wider confidence intervals because it increases the variance of the prediction errors.

When there is a common nonstationarity, the $t$-test performs well and is fully robust to misspecification, consistent with our theory. When there are deviations from the common nonstationarity and $T_0=30$, it may exhibit some undercoverage, especially if there is also misspecification (for example, under DGP8, which is not covered by our theory). This is expected since our theoretical results under deviations from nonstationarity require correct specification and $T_0$ to be much larger than $T_1$. As shown in the Appendix \ref{app:large_T0}, when $T_0=150$, the $t$-test works well under all DGPs covered by our theory. The results for DGP8 and DGP9 even suggest that if $T_0$ is large enough, the $t$-test remains quite robust in settings not covered by our theory.

The $t$-test performs well compared to the alternative methods. It is more robust than DID, which can exhibit large biases under deviations from a common nonstationarity, and can yield substantially shorter confidence intervals than DID (DGP6--DGP9).\footnote{We note that our application-based DGPs are quite favorable to DID: even when the data are stationary and SC is correctly specified, the DID confidence intervals are not much wider than those from the $t$-test. This finding suggests that the asymptotic efficiency gains from using the $t$-test are limited under these specific DGPs. By contrast, in simulations based on \citet{abadie2003economic}, we found that the DID confidence intervals can be much wider than those from the $t$-test, even when DID is theoretically valid \citep[][Section 5]{chernozhukov2022arxiv7}.} Subsampling undercovers even under correct specification when the bias is negligible. When SC is misspecified, $\htau^{\text{SC}}$ is typically biased (see also Figure \ref{fig:bias}), which can result in zero coverage. The $t$-test demonstrates a better coverage accuracy overall, and cross-fitting effectively removes the bias due to misspecification. Finally, while SDID is nearly as effective at removing the bias as cross-fitting, it exhibits undercoverage due to the confidence intervals being too short or overcoverage due to the confidence intervals being too long under most DGPs. When interpreting the results of this simulation comparison, it is important to note that SDID with one treated unit relies on cross-sectional variation and homoskedasticity for inference, whereas the $t$-test relies on time series variation. These two sets of assumptions are non-nested. See Sections \ref{sec:literature} and \ref{sec: when} for further discussions.

\begin{table}[H]
\setlength{\tabcolsep}{4pt}
\linespread{1.05}
\scriptsize
\centering
\caption{Simulation results with $(T_0,T_1,N)=(30,16,14)$}
\begin{tabular}{lccccccccccccccc}

\toprule
\midrule

&\multicolumn{5}{c}{Bias$\times$10} & \multicolumn{5}{c}{Coverage} & \multicolumn{5}{c}{Average length CI} \\
\cmidrule(l{5pt}r{5pt}){2-6}\cmidrule(l{5pt}r{5pt}){7-11} \cmidrule(l{5pt}r{5pt}){12-16}
$K$&$t$-DISCo&DID&SuSa&SDID&Oracle&$t$-DISCo&DID&SuSa&SDID&Oracle&$t$-DISCo&DID&SuSa&SDID&Oracle\\
\midrule
&\multicolumn{15}{c}{DGP1 (stationary, SC weights)}\\
\cmidrule(l{5pt}r{5pt}){2-16}
3 & 0.00 & 0.01 & -0.00 & -0.01 & 0.00 & 0.91 & 0.90 & 0.77 & 0.92 & 0.89 & 0.09 & 0.11 & 0.03 & 0.10 & 0.07 \\ 
  4 & -0.00 & -0.01 & -0.00 & -0.01 & -0.00 & 0.90 & 0.89 & 0.77 & 0.92 & 0.89 & 0.07 & 0.09 & 0.03 & 0.10 & 0.06 \\
\midrule
&\multicolumn{15}{c}{DGP2 (stationary, DID weights)}\\
\cmidrule(l{5pt}r{5pt}){2-16}
3 & -0.00 & -0.00 & -0.00 & -0.01 & -0.00 & 0.90 & 0.89 & 0.76 & 0.93 & 0.89 & 0.09 & 0.07 & 0.04 & 0.10 & 0.07 \\ 
  4 & -0.00 & -0.00 & -0.00 & -0.01 & -0.00 & 0.91 & 0.89 & 0.76 & 0.93 & 0.89 & 0.07 & 0.06 & 0.04 & 0.10 & 0.06 \\ 
\midrule
&\multicolumn{15}{c}{DGP3 (stationary, misspecified)}\\
\cmidrule(l{5pt}r{5pt}){2-16}
 3 & -0.00 & -0.02 & 19.99 & -0.05 & 0.00 & 0.90 & 0.90 & 0.00 & 0.69 & 0.89 & 0.68 & 0.63 & 0.32 & 0.10 & 0.07 \\ 
  4 & 0.02 & 0.02 & 19.99 & -0.05 & 0.00 & 0.90 & 0.90 & 0.00 & 0.69 & 0.89 & 0.56 & 0.53 & 0.32 & 0.10 & 0.06 \\ 
\midrule
&\multicolumn{15}{c}{DGP4 (common linear trend, misspecified)}\\
\cmidrule(l{5pt}r{5pt}){2-16}
3 & -0.00 & -0.01 & 20.01 & -0.01 & -0.00 & 0.89 & 0.90 & 0.00 & 0.71 & 0.90 & 0.67 & 0.63 & 0.32 & 0.10 & 0.08 \\ 
  4 & 0.02 & -0.00 & 20.01 & -0.01 & -0.00 & 0.89 & 0.89 & 0.00 & 0.71 & 0.90 & 0.56 & 0.53 & 0.32 & 0.10 & 0.06 \\ 
\midrule  
&\multicolumn{15}{c}{DGP5 (common random walk, misspecified)}\\
\cmidrule(l{5pt}r{5pt}){2-16}
3 & 0.02 & 0.01 & 19.98 & 0.02 & 0.00 & 0.90 & 0.90 & 0.00 & 0.72 & 0.89 & 0.68 & 0.63 & 0.32 & 0.10 & 0.07 \\ 
  4 & 0.04 & 0.02 & 19.98 & 0.02 & -0.00 & 0.89 & 0.89 & 0.00 & 0.72 & 0.89 & 0.56 & 0.53 & 0.32 & 0.10 & 0.06 \\ 
\midrule
&\multicolumn{15}{c}{DGP6 (common linear trend \& deviation, SC weights)}\\
\cmidrule(l{5pt}r{5pt}){2-16}
3 & -0.07 & -16.43 & -0.09 & -0.03 & -0.00 & 0.81 & 1.00 & 0.74 & 1.00 & 0.89 & 0.08 & 4.08 & 0.04 & 18.83 & 0.07 \\ 
  4 & -0.05 & -15.71 & -0.09 & -0.03 & 0.00 & 0.78 & 0.00 & 0.74 & 1.00 & 0.88 & 0.06 & 2.52 & 0.04 & 18.83 & 0.06 \\ 
\midrule
&\multicolumn{15}{c}{DGP7 (common random walk \& deviation, SC weights)}\\
\cmidrule(l{5pt}r{5pt}){2-16}
3 & -0.00 & 0.05 & 0.00 & 0.04 & -0.00 & 0.88 & 0.77 & 0.75 & 0.98 & 0.90 & 0.10 & 0.79 & 0.04 & 2.61 & 0.07 \\ 
  4 & -0.00 & -0.01 & 0.00 & 0.04 & 0.00 & 0.86 & 0.64 & 0.75 & 0.98 & 0.89 & 0.07 & 0.54 & 0.04 & 2.61 & 0.06 \\
\midrule
&\multicolumn{15}{c}{DGP8 (heterogeneous trends, misspecified)}\\
\cmidrule(l{5pt}r{5pt}){2-16}
 3 & -0.02 & -345.03 & -0.00 & -0.74 & 0.00 & 0.59 & 1.00 & 0.80 & 1.00 & 0.89 & 0.33 & 85.77 & 0.42 & 10.59 & 0.07 \\ 
  4 & 0.04 & -329.97 & -0.00 & -0.74 & -0.00 & 0.64 & 0.00 & 0.80 & 1.00 & 0.89 & 0.35 & 52.90 & 0.42 & 10.59 & 0.06 \\ 
\midrule
&\multicolumn{15}{c}{DGP9 (random walks with heterogeneous drifts, SC weights)}\\
\cmidrule(l{5pt}r{5pt}){2-16}
 3 & -0.02 & -369.70 & 0.00 & 0.01 & -0.00 & 0.92 & 1.00 & 0.81 & 1.00 & 0.90 & 0.53 & 91.81 & 0.21 & 15.05 & 0.08 \\ 
  4 & 0.01 & -353.62 & 0.00 & 0.01 & 0.00 & 0.82 & 0.00 & 0.81 & 1.00 & 0.90 & 0.26 & 56.76 & 0.21 & 15.05 & 0.06 \\ 
\midrule
\bottomrule

\multicolumn{16}{p{16.2cm}}{\scriptsize{\it Notes:} Simulation design based on the empirical application as described in the main text. Simulations are based on 5,000 repetitions, except for SDID for which we use 200 repetitions due to the placebo procedure being computationally expensive. $t$-DISCo: Proposed debiased SC estimator and $t$-test.\tablefootnote{``$t$-DISCo'' is short for $t$-Statistic Debiased Inference with Synthetic Controls. We use ``$t$'' and capitalization to distinguish the abbreviation from ``disco,'' which has been used to refer to Distributional Synthetic Controls \citep[e.g.,][]{gunsilius2025disco}.} SuSa: Subsampling inference based on $\htau^{\text{SC}}$ \citep{li2020statistical}. CI: Confidence interval. Nominal coverage: $1-\alpha= 0.9$. }
\end{tabular}
\addtolength{\tabcolsep}{4pt}    
\label{tab:empirical_mc}
\end{table}
\normalsize
\linespread{1.25}

\section{Estimating the impact of carbon taxes on emissions}
\label{sec:application}
In this section we revisit the SC analysis of the causal effect of carbon taxes on CO2 emissions in \citet{andersson2019carbon}. \citeauthor{andersson2019carbon} exploits the introduction of a carbon tax on transport fuels during the early 1990s in Sweden, using $J=14$ OECD countries as control units.\footnote{The countries are: ``Australia, Belgium, Canada, Denmark, France, Greece, Iceland, Japan, New Zealand, Poland, Portugal, Spain, Switzerland, and the United States'' \citep[][p.9]{andersson2019carbon}.} The outcome variable of interest measures CO2 emissions from transport (in metric tons per capita). The data are annual panel data from 1960--2005.\footnote{The data are available in the replication package \citep{andersson2019replication}.} \citeauthor{andersson2019carbon} uses data up to 2005 because the EU emissions trading system started in that year. The pre-treatment period is 1960--1989, and the post-treatment period is 1990--2005, so that $(T_0,T_1)=(30,16)$.

Figure \ref{fig:data_raw} displays the raw data. It shows a drop right before 1990 and reduced growth in emissions afterward. We apply the $t$-test to investigate whether this drop and reduced growth are due to the carbon tax and compare the results to other methods. All computations were performed in \texttt{R} \citep{R2023}. 

\begin{figure}[H]
\caption{Raw data on emissions per capita \citep{andersson2019carbon,andersson2019replication}}
\begin{center}
\includegraphics[width=0.6\textwidth,trim = {0 2cm 0 3cm}]{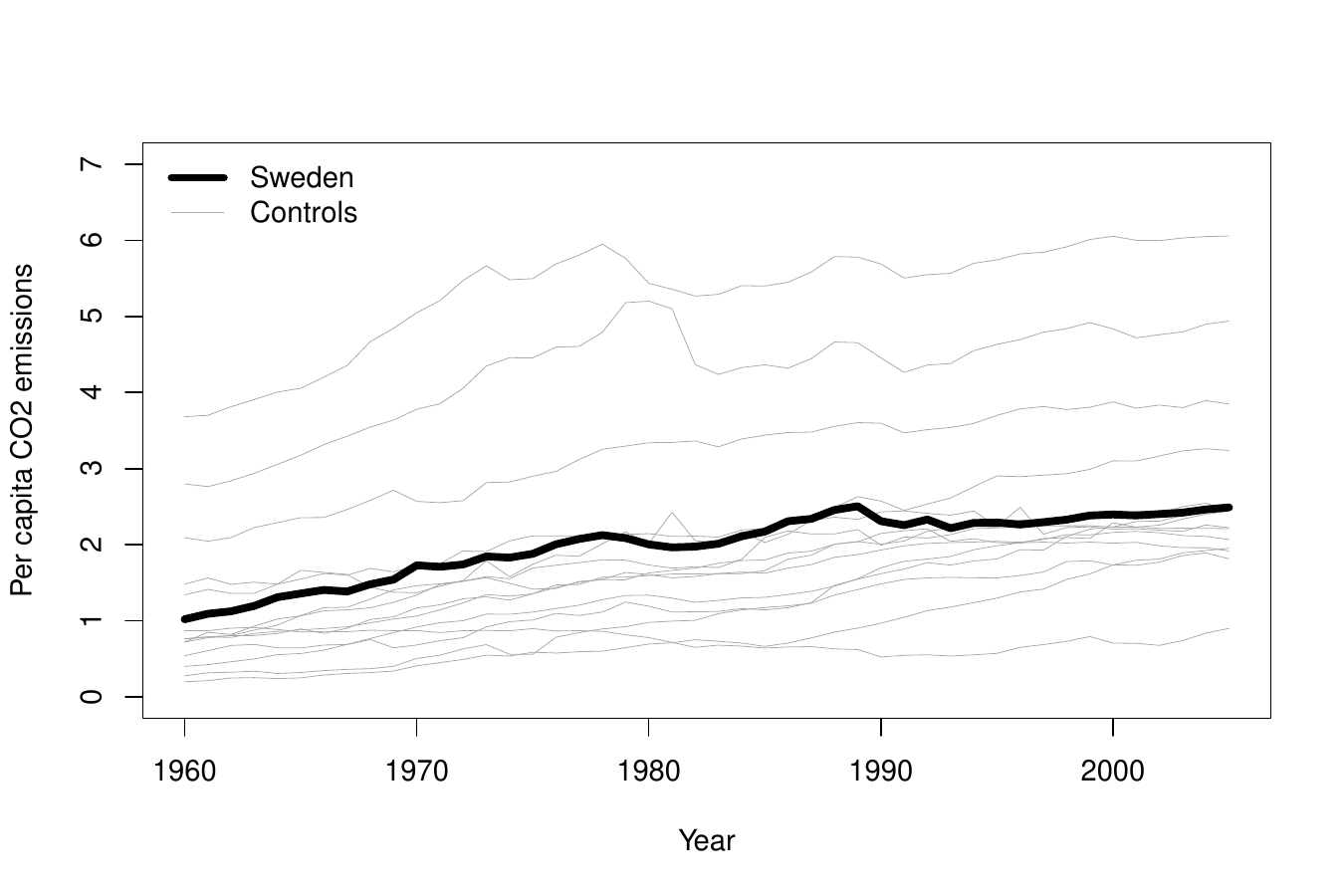}
\end{center}
\label{fig:data_raw}
\end{figure}

We start by assessing the validity of the $t$-test using the placebo test proposed in Section \ref{sec: time-varying weights}. 
Table \ref{tab:placebo} shows the placebo ATT estimates and the confidence intervals for $K=3$ and $\widetilde{T}_0\in \{18,21\}$ (to ensure a long-enough time period for estimating the weights and divisibility by $K=3$). The placebo estimates are smaller in magnitude than the actual effect estimates in Table \ref{tab:comparison} below, and the confidence intervals include zero. Thus, we do not reject the validity of the assumptions underlying the $t$-test.
\begin{table}[H]
\footnotesize
\centering
\caption{Results placebo test}
\label{tab:placebo}

\begin{tabular}{lccc}
\toprule
\midrule
$\widetilde{T}_0$ & ATT & \multicolumn{2}{c}{90\%-CI}\\
\midrule
18 & 0.01 & -0.18 & 0.19 \\ 
  21 & 0.10 & -0.21 & 0.41 \\ 
\midrule
\bottomrule
\multicolumn{4}{p{4.2cm}}{\scriptsize{\it Notes:} Table shows point estimates and 90\% confidence intervals for the placebo test described in Section \ref{sec: time-varying weights}.}
\end{tabular}
\end{table}
\normalsize

Panel (a) of Table \ref{tab:comparison} shows the estimated ATT and 90\% confidence intervals based on the $t$-test with $K=3$. $K=3$ is a useful benchmark in settings where $T_0$ is small or moderate (Section \ref{sec: choosing K}). The estimated ATT is negative and significant. Thus, our findings corroborate the results in \citet{andersson2019carbon} based on the permutation test of \citet{abadie10sc}.\footnote{Note that our SC specification differs from \citet{andersson2019carbon}. We use all past outcomes as predictors, whereas \citet{andersson2019carbon} uses a subset of past outcomes and additional predictors to compute the weights.
}

Panels (b)--(d) of Table \ref{tab:comparison} show the corresponding results based on DID with $K=3$, subsampling, and SDID. The point estimates are  similar across all methods, ranging from $-0.35$ for SDID to $-0.21$ for DID. Consistent with the $t$-test, the results based on DID and subsampling suggest that the carbon tax significantly decreased emissions. SDID yields wider confidence intervals that include zero. This finding is in line with the simulations in Section \ref{sec:simulations}, which show that SDID may yield wider confidence intervals than the $t$-test when there is heterogeneity in the nonstationarities across control units.

\begin{table}[H]
\footnotesize
\centering
\caption{Results $t$-test and comparison to other methods}
\label{tab:comparison}

\begin{tabular}{ccc ccc ccc ccc}
\toprule
\midrule
 \multicolumn{3}{c}{(a) $t$-DISCo ($K=3$)}&  \multicolumn{3}{c}{(b) DID ($K=3$)} & \multicolumn{3}{c}{(c) Subsampling} & \multicolumn{3}{c}{(d) Synthetic DID}\\
 \cmidrule(l{5pt}r{5pt}){1-3} \cmidrule(l{5pt}r{5pt}){4-6}  \cmidrule(l{5pt}r{5pt}){7-9} \cmidrule(l{5pt}r{5pt}){10-12} 
ATT & \multicolumn{2}{c}{90\%-CI} &ATT & \multicolumn{2}{c}{90\%-CI} &ATT & \multicolumn{2}{c}{90\%-CI}&ATT & \multicolumn{2}{c}{90\%-CI}  \\
\midrule
-0.27 & -0.41 & -0.14 & -0.21 & -0.36 & -0.07 & -0.28 & -0.37 & -0.23 & -0.35 & -0.84 & 0.14 \\ 
\midrule
\bottomrule
\multicolumn{12}{p{13.5cm}}{\scriptsize{\it Notes:} Panel (a): Debiased SC estimate and 90\% confidence intervals based on the $t$-test ($t$-DISCo). Panels (b)--(d): Results for DID, subsampling, and SDID, implemented as described in Section \ref{sec:simulations}. Note that the subsampling CI is not symmetric due to being based on the percentile method.}
\end{tabular}
\end{table}
\normalsize

\section{Recommendations for practice}
\label{sec:guide}

\subsection{When should you use the $t$-test?}
\label{sec: when}
Before deciding whether to use the proposed $t$-test, researchers need to decide whether to use SC in the first place. \citet{abadie2021using} provides a detailed discussion of the relevant practical considerations. The most popular alternative to SC is DID. While conventional SC and DID are non-nested \citep{DI16}, we show that debiased SC is more robust than DID. Therefore, we recommend it whenever the DID assumptions are questionable.

There are three important and interrelated considerations when choosing a suitable  method for making inferences on the ATT using the SC method. See Section \ref{sec:literature} for references to papers proposing inference methods for per-period effects and sharp null hypotheses.

\smallskip 

\noindent
\textbf{Number of treated units.} The $t$-test is designed for applications with one treated unit. When there are multiple treated units, the $t$-test can be applied separately for each unit. However, when the number of treated units is large, other approaches that directly target average effects across units, such as \citet{abadie2021penalized} and \citet{arkhangelsky2021synthetic}, might be preferable. For applications with multiple treated units and staggered treatment adoption, we recommend using methods that are specifically designed to accommodate staggered adoption, such as \citet{shaikh2021randomization}, \citet{benmichael2022synthetic}, or \citet{cattaneo2023uncertainty}.

\smallskip 

\noindent
\textbf{Number of periods.} Asymptotic normality results for the ATT with a single treated unit require the number of periods to tend to infinity. Existing inference methods based on asymptotic normality that rely on classical estimators of the LRV or subsampling can exhibit substantial size distortions in small samples. By contrast, the $t$-test, which relies on a self-normalized test statistic, exhibits higher-order improvements and performs well in simulations even when $T_0$ and $T_1$ are small. Therefore, we recommend using the $t$-test instead of methods that rely on estimating the LRV or subsampling when $T_0$ and $T_1$ are small or moderate, as is the case in many SC applications. 
When $T_1$ is too small to rely on asymptotics where $T_1\rightarrow \infty$, researchers can use inference procedures that are valid when $T_1$ is fixed, such as \citet{chernozhukov2021exact} or \citet{cattaneo2023uncertainty}.\footnote{See, for example, \citet{masini2020counterfactual} and \citet{masini2021jasa} for inference methods based on (penalized) regression approaches for estimating counterfactuals.} This is relevant, for example, when structural breaks in the post-treatment period invalidate the stationarity assumptions required for the $t$-test to be valid or when control units also get treated shortly after the treated unit (e.g., in staggered adoption designs).

\smallskip 

\noindent
\textbf{Which type of variation to exploit for inference.} In panel data settings, inference methods can exploit the time series and/or the cross-sectional dimension. The $t$-test exploits the time series dimension for inference. It relies on stationarity and weak dependence of the SC prediction errors for the treated unit over time and is therefore particularly well-suited for settings where the units are heterogeneous. This is often the case in SC applications based on aggregate units, such as states or countries. For example, in the empirical application in Section \ref{sec:application}, the treated unit is Sweden, and the control units are other OECD countries, including Iceland, Spain, and the United States.

The stationarity and weak dependence assumptions underlying the $t$-test are not innocuous and might be questionable when there are structural breaks in the data. If the units are homogeneous enough to justify homoskedasticity and weak dependence across units, we recommend methods that exploit cross-sectional variation, such as \citet{arkhangelsky2021synthetic}.

\subsection{Recommendations for implementing the $t$-test}

In Section \ref{sec: when}, we discuss when to use the $t$-test. Here we provide some recommendations for implementing the $t$-test in empirical applications.

First, with non-stationary data, the robustness of the $t$-test improves substantially as $T_0$ increases. Therefore, we recommend collecting enough pre-treatment data in such settings. When collecting additional pre-treatment data, researchers need to be careful about structural breaks in the data. The placebo test described in Section \ref{sec: time-varying weights} can be used to test for structural breaks.

Second, the $t$-test relies on the predictive relationship between the treated and the control units being sufficiently stable over time. In Section \ref{sec: time-varying weights}, we demonstrate that the $t$-test remains valid when the SC weights vary  over time in a stationary manner. For applications where researchers are concerned about the weights changing in a non-stationary manner, we recommend using the placebo test described in Section \ref{sec: time-varying weights} to assess the validity of the assumptions underlying the $t$-test. This test requires enough pre-treatment periods.

Finally, the choice of $K$ is subject to an inherent trade-off between the expected length of the confidence intervals and their coverage accuracy. For SC applications with small and moderate sample sizes, choosing $K=3$ provides a useful and robust benchmark: it ensures good coverage properties while providing a reasonable RAE. More generally, researchers can use the RAE formula \eqref{eq:rae} in conjunction with estimates of the degree of persistence in the prediction errors and application-based simulations to guide the choice of $K$. Section \ref{sec: choosing K} provides a detailed discussion on how to choose $K$.

%\spacingset{0.9}
\setlength{\bibsep}{1pt}
\bibliographystyle{apalike}
\bibliography{SC_biblio}

\newpage

\appendix

\setcounter{page}{1} 

\begin{center}
\LARGE{Appendix (for online publication)}
\end{center}

\startcontents[sections]
\printcontents[sections]{l}{1}{\setcounter{tocdepth}{2}}

\section{Random treatment status}
\label{app:random_treatment}

In the main text, we considered a setup where the treatment status is fixed. Here we discuss the properties of our method when the treatment status is random. In SC applications, the treatment assignment is usually not completely exogenous, and unconfoundedness conditions have been considered in the literature. The idea is that conditional on certain variables, the treatment assignment is independent (or mean-independent) of the potential outcomes. In \citet{kellogg2021combining}, the conditioning variables are allowed to include all the observed variables before the treatment; in works such as  \citet{athey2021matrix}, \citet{ferman2021properties}, and \citet{imbens2023identification} among  others, the conditioning variables are latent variables.\footnote{See \citet{arkhangelsky2023largesample} for some interesting recent work on combining both types of assumptions in SC settings with many treated units.} Here we discuss our results under these two approaches. We treat the treatment status as random and the treatment date as fixed. Let $D\in \{0,\dots,N\}$ denote a random variable indicating the identity of the treated unit. For $1\leq t\leq T$ and $i\in \{0,\dots,N\}$, we observe $Y_{it}=Y_{it}(1)\cdot \oneb\{D=i\} \cdot  \oneb\{t\geq T_0+1\} + Y_{it}(0)\cdot (1-\oneb\{D=i\} \cdot  \oneb\{t\geq T_0+1\}) $, where $Y_{it}(1)=Y_{it}(0)+\alpha_{it}$.

\subsection{Selection on observed variables}

We follow \citet{kellogg2021combining} and allow the treatment assignment $D$ to depend on observed variables, including the outcome variables in the pre-treatment period. Their paper establishes a decomposition of the bias of SC methods. We deviate from their paper in two ways. First, we allow the distribution of each unit to be different. Second, we impose  additional assumptions to derive inference results for the ATT.

\begin{assumption}
\label{assu: kellogg part 1}For $i\in\{0,\dots,N\}$, $Y_{it}(0)=\theta_{t}+V_{it}$,
where $\{V_{t}\}_{t=1}^{T}$ and $\{\theta_{t}\}_{t=1}^{T}$ are two
independent processes with $V_{t}=(V_{0t},\dots,V_{Nt})'\in\RR^{N+1}$.
Moreover, $\{V_{t}\}_{t=1}^{T}$ is covariance-stationary and satisfies
the following conditions:
\begin{enumerate}\setlength\itemsep{-1.5pt}
\item There exists a constant $\kappa_{1}>0$ such that all the eigenvalues
of $EV_{t}V_{t}'$ are in $[1/\kappa_{1},\kappa_{1}]$
\item There exists a sequence $\rho_{T}>0$ such that $P(\max_{1\leq t\leq T}\|V_{t}\|_{\infty}\leq\rho_{T})\rightarrow1$
\item $\{V_{t}\}_{t=1}^{T}$ is $\beta$-mixing with coefficient satisfying
$\beta_{{\rm mix}}(\gamma_{T})\rightarrow0$ for some $\gamma_{T}$
with $0<\gamma_{T}<r/2$ and $\rho_{T}\gamma_{T}=o(\min\{\sqrt{T_{0}},\sqrt{T_{1}}\})$
\item $\max_{1\leq k\leq K}\||H_{(-k)}|^{-1}\sum_{t\in H_{(-k)}}V_{t}V_{t}'-EV_{t}V_{t}'\|_{\infty}=o_{P}(1)$
\end{enumerate}
\end{assumption}
Assumption \ref{assu: kellogg part 1} allows for a common non-stationary trend $\theta_t$ and is similar to the conditions in Section \ref{sec:common_trend}.  
Our object of interest is 
\begin{equation}
\tau=T_{1}^{-1}\sum_{i=0}^{N}\sum_{t=T_{0}+1}^{T}\alpha_{it}\cdot\oneb\{D=i\}.\label{eq:random_att}
\end{equation}
Notice that only one unit is treated, and  $D \in \{0,\dots,N\}$ is allowed to be a random variable. Hence, this target is also random and our goal is to construct a prediction interval for $\tau$. We impose the following unconfoundedness condition.

\begin{assumption}
\label{assu: kellogg part 2}Let $Y_{t}(0)=(Y_{0t}(0),\dots,Y_{Nt}(0))'\in\RR^{N+1}$.
Conditional on $\{Y_{t}(0)\}_{T_{0}-B_{T}\leq t\leq T_{0}}$, $D$
is independent of $\{Y_{t}(0)\}_{1\leq t\leq T}$, where $B_{T}\rightarrow\infty$
satisfies $B_{T}\rho_{T}=o(\min\{\sqrt{T_{0}},\sqrt{T_{1}}\})$. 
\end{assumption}

Under Assumption \ref{assu: kellogg part 2}, the selection into treatment can only depend on the outcome variables $B_T$ periods before the treatment date but does not depend on ``ancient history''. This might be reasonable if the treatment is a result of contemporary political, social and/or economic developments, such as public sentiment on tobacco or changes in geopolitical status. One can include additional pre-determined variables to the conditioning variables, and the same analysis applies once we write the assumptions conditional on these extra variables. 

We define our procedure as follows. For $i\in\{0,\dots,N\}$ and $k\in\{1,\dots,K\}$, we define 
\[
\hat{w}_{(k)}^{(i)}\in\underset{w\in\mathcal{W}^{SC}}{\arg\min}\,\sum_{t\in H_{(-k)}}(Y_{it}-Y_{-i,t}'w)^{2},
\]
where $Y_{-i,t}\in\RR^{N}$ is the vector $Y_{t}\in\RR^{N+1}$ with
the $i$-th component removed. For $k\in\{1,\dots,K\}$, we define the
estimator 
\[
\hat{\tau}_{k}=\frac{1}{T_1}\sum_{i=0}^{N}\sum_{t=T_{0}+1}^{T}(Y_{it}-Y_{-i,t}'\hat{w}_{(k)}^{(i)})\mathbf{1}\{D=i\}-\frac{1}{|H_{k}|}\sum_{i=0}^{N}\sum_{t\in H_{k}}(Y_{it}-Y_{-i,t}'\hat{w}_{(k)}^{(i)})\mathbf{1}\{D=i\}
\]

In practice, we only need to compute the SC weights for the treated unit. Hence,  the same construction as in Section \ref{sec:implementation} applies:
\[
\mathbb{T}_{K}=\frac{\sqrt{K}(\hat{\tau}-\tau)}{\hat{\sigma}_{\hat{\tau}}},
\]
where $\hat{\tau}=K^{-1}\sum_{k=1}^{K}\hat{\tau}_{k}$ and $\hat{\sigma}_{\hat{\tau}}=\sqrt{1+Kr/T_{1}}\sqrt{(K-1)^{-1}\sum_{k=1}^{K}(\hat{\tau}_{k}-\hat{\tau})^{2}}$.
The $(1-\alpha)$ confidence interval for $\tau$ also takes the same
form: 
\[
\mathcal{I}_{K}(1-\alpha)=\left[\hat{\tau}-t_{K-1}(1-\alpha/2)\frac{\hat{\sigma}_{\hat{\tau}}}{\sqrt{K}},\ \hat{\tau}+t_{K-1}(1-\alpha/2)\frac{\hat{\sigma}_{\hat{\tau}}}{\sqrt{K}}\right].
\]

The following result establishes the validity of the procedure. 
\begin{thm}
\label{thm: kellog}Let Assumptions \ref{assu: kellogg part 1} and
\ref{assu: kellogg part 2} hold. Then $\mathbb{T}_{K}\overset{d}{\rightarrow}t_{K-1}$
and $P(\tau\in\mathcal{I}_{K}(1-\alpha))\rightarrow1-\alpha$. 
\end{thm}

\subsection{Selection on latent variables}

Following the literature, we consider a factor model for the potential outcomes \citep[e.g.,][]{abadie10sc,arkhangelsky2021synthetic,benmichael2021augmented,ferman2021properties,ferman2021synthetic},
\begin{equation}
Y_{it}(0)=\theta_t + L_i'F_t+\varepsilon_{it}, \quad E(\varepsilon_{it})=0, \quad 0\le i\le N,~ 1\le t\le T,\label{eq:factor_model}
\end{equation}
where $\theta_t$ is a common unrestricted nonstationarity, $L_i$ is a vector of loadings, $F_t$ is a vector of factors, and $\varepsilon_{it}$ is an idiosyncratic shock. Define $L:=(L_0,\dots,L_N)$, $F:=(F_1,\dots,F_T)$, and $\varepsilon:=(\varepsilon_{1},\dots,\varepsilon_{T})$, where $\varepsilon_t=(\varepsilon_{0t},\dots,\varepsilon_{Nt})'$ for $t=1,\dots,T$. We emphasize that while we focus on linear factor models for concreteness, the $t$-test is valid more broadly and accommodates various types of outcome models and general forms of misspecification.

The validity of the $t$-test follows if we can show that the prediction model implied by the factor model satisfies the conditions required in the main text, conditional on treatment assignment $D$. 
 SC methods are biased under factor models absent additional restrictions \citep[e.g.,][]{abadie10sc,ferman2021synthetic}, which  motivates the proposed bias correction. We consider the following assumption. 
\begin{assumption}
\label{assu: unconfoundedness factor}
For some random vector $v$, $D=f(L,v)$ and $(L,v)\indep (F,\varepsilon)$.
\end{assumption}
The selection mechanism $D=f(L,v)$ can be viewed as a generalization of the notion of ``selection on fixed effects,'' which has traditionally motivated the use of DID methods, to factor models.\footnote{See, for example, \citet{ghanem2023selection} for a discussion of selection mechanisms in the context of DID.} It captures the idea that the only systematic confounders are individual characteristics, similar to \citet{imbens2023identification}. 
 We can also state Assumption \ref{assu: unconfoundedness factor} as $(L,D)\indep (F,\varepsilon)$.

Define $Y_{t}(0):=Y_{Dt}(0)$, and denote by $X_t$ the vector collecting $Y_{jt}(0)$ for $j\ne D$.
The $t$-test is valid if, conditional on $D$, $\{(V_{t}(0),Z_t)\}_{t=1}^T$, satisfies the conditions in the main text, where $V_{t}(0)=Y_t(0)-\theta_t$ and $Z_{t}=X_t-\mathbf{1}_N\theta_t$. We establish the following result.
\begin{lem}\label{lem:sufficient_stationarity} Suppose that the data are generated by model \eqref{eq:factor_model}. Let Assumption \ref{assu: unconfoundedness factor} hold. If  $\{F_t,\varepsilon_t\}_{t=1}^T$ is stationary, then $\{(V_t(0),Z_t)\}_{t=1}^T$ is stationary conditional on $D$.
\end{lem}
By Lemma \ref{lem:sufficient_stationarity}, the validity of the $t$-test follows by invoking the same theory as in Section \ref{sec:common_trend}.

\section{Additional extensions}
\label{app: additional extensions}

\subsection{Consistency of the classical SC estimator}
\label{sec:constistency_classical_SC}

In the main text, we focus on the canonical SC estimator \eqref{eq:sc_estimator}. Here we provide conditions under which the classical SC estimator of \citet{abadie10sc} satisfies the consistency requirement in Assumption \ref{ass:ell_2_consistency}, which implies that the $t$-test can also be used in conjunction with the classical SC estimator. 

We start by discussing the theoretical properties of the classical SC estimator based on all pre-treatment data and discuss how to implement the cross-fitting approach underlying the $t$-test at the end. Consider the same setup as in the main text, but suppose that in addition to outcome data, researchers have access to a vector of covariates or predictors $Z_i\in \mathbb{R}^p$ for $i=0,\dots,N$. 

Our results accommodate settings where the outcomes are generated by a factor model with covariates as in \citet{abadie10sc}, 
\begin{equation}
Y_{it}(0)=\theta_t+\delta_tZ_i+L_iF_t+\varepsilon_{it},\label{eq:factor_model_abadie10sc}
\end{equation}
where $\theta_t$ is a common time effect, $\delta_t$ is a time-varying coefficient, $L_i$ is a vector of unit-specific loadings, and $F_t$ is a vector of time-varying factors. However, importantly, our results do not require the true data generating process to be a factor model. This feature provides useful robustness in applications.

To define the estimator and state our theoretical results, we introduce some additional notation. Define $Y_{i}=(Y_{i1},\dots,Y_{iT_{0}})'\in\RR^{T_{0}}$, $Y_{-0}=(Y_{1},\dots,Y_{N})\in\RR^{T_{0}\times N}$ and $Z_{-0}=(Z_{1},\dots,Z_{N})\in\RR^{p\times N}$.

We consider the classical SC estimator
\begin{equation}
\tw\in\arg\min_{w\in\mathcal{W}^{SC}}\left\Vert \begin{pmatrix}Y_{0}\\
Z_{0}
\end{pmatrix}-\begin{pmatrix}Y_{-0}\\
Z_{-0}
\end{pmatrix}w\right\Vert _{V}^{2},\label{eq:classical_sc}
\end{equation}
where $V$ is a positive semi-definite  matrix and $\|c\|_{V}^{2}=c'Vc$ for any $c\in\RR^{T_{0}+p}$.
It will be useful to partition $V$ as
\[
V=\begin{pmatrix}V_{1} & 0\\
0 & V_{2}
\end{pmatrix},
\]
where $V_{1}\in\RR^{T_{0}\times T_{0}}$ and $V_{2}\in\RR^{p\times p}$. To simplify the exposition, we include all pre-treatment outcomes in estimator \eqref{eq:classical_sc}, but our results continue to hold as long as researchers use a subset of $O(T_0)$ pre-treatment outcomes. Different choices of $V$ have been proposed \citep[e.g.,][]{abadie10sc,abadie15sc}. Our results below allow for general $V$, and we do not restrict the relative magnitude of $V_2$ and $V_1$. See Remark \ref{rem:importance_covariates} for a further discussion of the role of covariates.

The following lemma establishes the $\ell_2$-consistency of $\tilde{w}$ for the pseudo-true SC weights,
\[
\bw=(\bw_{1},...,\bw_{N})'=\arg\min_{w\in\mathcal{W}^{SC}}E\left\Vert \begin{pmatrix}Y_{0}\\
Z_{0}
\end{pmatrix}-\begin{pmatrix}Y_{-0}\\
Z_{-0}
\end{pmatrix}w\right\Vert _{V}^{2},
\]
where we use $\bw$ instead of $w_\ast$ to denote the pseudo-true weights to distinguish them from the pseudo-true weights for the SC estimator \eqref{eq:sc_estimator} in the main text.

As in Section \ref{sec:common_trend}, we allow for a common nonstationarity. Specifically, we assume that $Y_{it}(0)=\tY_{it}(0)+\theta_t$, where $\theta_t$ is a non-stationary process. Define $\tY_{-0}=(\tY_{1},\dots,\tY_{N})\in\RR^{T_{0}\times N}$ with $\tY_{i}=(\tY_{i1},\dots,\tY_{iT_{0}})'\in\RR^{T_{0}}$. 

\begin{lem}
\label{lem: classical SC}
Suppose that $\|T_{0}^{-1}u'V_{1}\tY_{-1}-T_{0}^{-1}Eu'V_{1}\tY_{-1}\|_{\infty}=o_{P}(1)$
and $\|T_{0}^{-1}\tY_{-1}'V_{1}\tY_{-1}-\Sigma_{V}\|_{\infty}=o_{P}(1)$,
where $u=Y_{0}=Y_{-0}\bw$ and $\Sigma_{V}=ET_{0}^{-1}\tY_{-0}'V_{1}\tY_{-0}$. Suppose further that $Z_i$, for $i\in 0,\dots,N$, is non-stochastic.\footnote{If $Z_i$, $i\in 0,\dots,N$, is stochastic, we interpret the analysis as conditional on $(Z_0,\dots,Z_N)$.} If $\lambda_{\min}(\Sigma_{V})$
is bounded below by a positive constant, then $\|\tw-\bw\|_{2}=o_{P}(1)$. 
\end{lem}

Lemma \ref{lem: classical SC} allows for arbitrary misspecification and a common nonstationarity $\theta_t$. It accommodates factor models, such as \eqref{eq:factor_model_abadie10sc}, but does not require the outcomes to be generated by a factor model. Moreover, we do not require the existence of ``true'' SC weights that balance the individual loadings and covariates, as, for example, in \citet{abadie10sc}. 

Lemma \ref{lem: classical SC} establishes the consistency of the classical SC estimator with covariates. Implementing the cross-fitting required for the $t$-test is straightforward. For each $k=1,\dots,K$, compute the component estimators as 
\begin{equation*}
\htau_{k}=\frac{1}{T_{1}}\sum_{t=T_{0}+1}^{T}\left(Y_{0t}-\sum_{i=1}^N\tw_{i,(k)}Y_{it}\right)-\frac{1}{|H_{k}|}\sum_{t\in H_{k}}\left(Y_{0t}-\sum_{i=1}^N\tw_{i,(k)}Y_{it}\right), 
\end{equation*}
where $\tw_{(k)}=(\tw_{1,(k)},\dots, \tw_{N,(k)})'$ are the weights obtained from applying the classical SC estimator to the outcome data in the subperiod $H_{(-k)}:=\{1,\dots,T_0\}\setminus H_k$ and all covariates. The estimator of the ATT and the $t$-statistic can be constructed as described in Section \ref{sec:implementation}. The results in Lemma \ref{lem: classical SC} imply that $\max_{1\le k \le K}\|\tw_{(k)}-\bar{w}\|_{2}=o_P(1)$ so that Assumption \ref{ass:ell_2_consistency} holds with $w_\ast$ replaced by $\bar{w}$, and the theoretical results in the main text imply that the $t$-test is valid.

\begin{rem}\label{rem:importance_covariates}
There are some recent discussions on the role of covariates. Although one might naturally think that including the covariates helps with the estimation accuracy of the weights, papers such as \cite{kaul2022standard} show that certain ways of choosing the weights eventually lead to $V_2=0$, which makes covariates irrelevant. Some papers derive consistency towards a target that has no covariates even though the estimator takes into account covariates, e.g.,
\citet{zhang2021asymptotic}. Here,  we do not take a stand on this by just taking $V$ as given. We allow for zero $V_2$ or large $V_2$. Therefore, the consistency result below applies no matter how important the covariates are. \qed
\end{rem}

\subsection{Inference on the expected effect}
\label{app: inference on expected effect}

In the main text, we propose a $t$-test for making inferences on the ATT, $\tau=T^{-1}\sum_{t=T_0+1}^T\alpha_{t}$. If researchers are willing restrict treatment effect heterogeneity over time and assume that $\{\alpha_{t}\}$ is stationary and weakly dependent, the expected treatment effect, $\tau_e=E(\alpha_{t})$, is a natural alternative to the ATT \citep[e.g.,][]{li2017estimation,li2020statistical}.
Here, we show how our method can be modified to make inferences on $\tau_e$.

To motivate this modification, note that by similar arguments as in Lemma \ref{lem: algebra} and Theorem \ref{thm: asy distr}, we have
\[
\hat{\tau}_{k}-\tau_{e}=\frac{1}{T_{1}}\sum_{t=T_{0}+1}^{T}u_{t}-\frac{1}{|H_{k}|}\sum_{t\in H_{k}}u_{t}+\frac{1}{T_{1}}\sum_{t=T_{0}+1}^{T}\tilde{\alpha}_{t}+o_{P}\left(\frac{1}{\sqrt{\min\{T_{0},T_{1}\}}}\right).
\]
where $\tilde{\alpha}_{t}=\alpha_{t}-\tau_e$. The main difference to the analysis in Section \ref{sec: theory stationary} is the term $T_{1}^{-1}\sum_{t=T_{0}+1}^{T}\tilde{\alpha}_{t}$. This term arises because of the different centering ($\tau_e$ instead of $\tau$) and captures the randomness of the treatment effects. Due to the presence of this additional term, the statistic $\sqrt{K}\left(\htau-\tau_e \right)/\hat\sigma_{\htau}$, does not have an asymptotic $t$-distribution. (Note, however, that $\hat{\tau}_{k}$ is asymptotically unbiased for $\tau_e$ since $T_{1}^{-1}\sum_{t=T_{0}+1}^{T}u_{t}-|H_{k}|^{-1}\sum_{t\in H_{k}}u_{t}$ has mean zero under our assumptions.)

To overcome this issue, we introduce a modified debiasing approach based on also splitting up the post-treatment period. To simplify the exposition, we assume that $T_1/K$ is an integer in the following. We split the post-treatment period into $K$ consecutive non-overlapping blocks with $T_1/K$ elements, $\tilde{H}_1,\dots,\tilde{H}_K$. The estimator of the expected ATT is given by 
$$
\tilde\tau = \frac{1}{K}\sum_{k=1}^K\tilde{\tau}_k,
$$
where 
\begin{equation}
\tilde{\tau}_{k}=\frac{1}{|\tilde{H}_k|}\sum_{t\in \tilde{H}_k}\left(Y_{t}-X_t'\hw_{(k)}\right)-\frac{1}{|H_{k}|}\sum_{t\in H_{k}}\left(Y_{t}-X_{t}'\hw_{(k)}\right).
\end{equation}
To make inferences, we again construct $t$-statistic based on $\{\tilde{\tau}_{1},\dots,\tilde{\tau}_{K}\}$. Since these estimators are based on non-overlapping blocks of data and thus asymptotically independent (see proof of Theorem \ref{thm: asy distribution expected effect}), rescaling the denominator by $\sqrt{1+Kr/T_1}$ is not necessary, and we can use a standard $t$-statistic,
$$
\tilde{\mathbb{T}}_K=\frac{\sqrt{K}(\tilde\tau - \tau_e)}{\tilde\sigma_{\tilde\tau}},
$$
where
$$
\tilde\sigma_{\tilde\tau}=\sqrt{\frac{1}{K-1}\sum_{k=1}^K \left(\tilde{\tau}_k -\tilde{\tau}\right)^2}.
$$
The corresponding $(1-\alpha)$ confidence interval for $\tau_e$ is 
\begin{equation*}
\tilde{\mathcal{I}}_{K}(1-\alpha)=\left [ \tilde{\tau} -t_{K-1}(1-\alpha/2)\frac{\tilde\sigma_{\tilde\tau}}{\sqrt{K}},~ \htau +t_{K-1}(1-\alpha/2) \frac{\tilde\sigma_{\tilde\tau}}{\sqrt{K}}\right].
\end{equation*}

To handle this case, we modify parts (3) and (4) of Assumption \ref{assu: weak dependence}. 

\begin{assumption}
\label{assu: new weak dependence}Suppose the following conditions hold.
\begin{enumerate}  \setlength\itemsep{0pt}
\item There exists a constant $\kappa_{1}>0$ such that for any $A\subseteq\{1,\dots,T\}$,
the largest eigenvalue of $E\left[|A|^{-1}\left(\sum_{t\in A}\tilde{X}_{t}\right)\left(\sum_{t\in A}\tilde{X}_{t}\right)'\right]$
is bounded above by $\kappa_{1}$. 
\item There exists a sequence $\rho_{T}>0$ such that $P(\max_{1\leq t\leq T}\|\tilde{X}_{t}\|_{\infty}\leq\rho_{T})\rightarrow1$. 
\item The data $\{(X_{t},\tilde{u}_{t},\tilde{\alpha}_{t})\}_{t=1}^{T}$ are $\beta$-mixing with coefficient
satisfying $\betamix(\gamma_{T})\rightarrow0$ for some sequence $\gamma_{T}$
satisfying $0<\gamma_{T}<r/2$ and $\rho_{T}\gamma_{T}=o_P(\min\{\sqrt{T_{0}},\sqrt{T_{1}}\})$, where $r=\min\{\left\lfloor T_{0}/K\right\rfloor,T_1 \}$.
\item $\{\tilde{u}_{t}\}_{t=1}^{T}$ and $\{\tilde{\alpha}_{t}\}_{t=1}^{T}$ satisfy $\max_{1\leq t\leq T}E|\tilde{u}_{t}|^q=O(1) $, $\max_{1\leq t\leq T}E|\tilde{\alpha}_{t}|^q=O(1) $ and $\betamix(i)\lesssim i^{-\eta} $ for some constants $q> 2$ and $\eta>q/(q-2)$ and $\sigma^2>0$.
\end{enumerate}
\end{assumption}

The following theorem provides a formal justification of the modified inference method.

\begin{thm}\label{thm: asy distribution expected effect}
Let Assumptions  \ref{assu: stationary Y0 X}, \ref{ass:ell_2_consistency}, and \ref{assu: new weak dependence}  hold. Suppose that $T_0,T_1\rightarrow \infty$ and that $T_{0}/T_{1}\rightarrow c_{0}$ for some $c_{0}\in[0,\infty]$. 
Then (i) $\tilde{\mathbb{T}}_K\overset{d}\rightarrow t_{K-1}$ and (ii) $P(\tau_e\in \tilde{\mathcal{I}}_{K}(1-\alpha))\rightarrow 1-\alpha.$
\end{thm}

The proof of Theorem \ref{thm: asy distribution expected effect} shows that the component estimators $\tilde{\tau}_{1},\dots,\tilde{\tau}_{K}$ are asymptotically independent due to the sample splitting in the post-treatment period, so that one can use a standard $t$-statistic. This result has two important practical implications. First, the $t$-test based on $\tilde{\mathbb{T}}_K$ is robust against heterogeneity in the variances of the component estimators \citep{bakirov2005student,ibragimovmueller2010} and thus is robust against certain types of nonstationarity in $\{u_t\}$. Second, the sign-based randomization test of \citet{canay2017randomization} constitutes a natural alternative to the $t$-test, provided that $K$ is large enough. The randomization test is particularly appealing when stationarity of $\{u_t\}$ is questionable: it is robust to heterogeneity in the variances across component estimators and more powerful than the $t$-test in such settings \citep{canay2017randomization}.

\begin{rem}[ATT vs.\ expected effect]\label{rem:att_vs_expected_effect}
The quantities $\tau$ and $\tau_e$ are sometimes referred to as the sample ATT and population ATT, respectively. See, for example,  \citet{imbens2004nonparametric} and \citet{rothe2017robust}. As pointed out by \citet{imbens2004nonparametric}, there would be no extra information in the data about $\tau_e$ beyond $\tau$, even if both potential outcomes were observed. In terms of efficiency, \citet{imbens2004nonparametric} also observed that $\tau$ can be estimated more accurately than $\tau_e$ when treatment effects are heterogeneous. In cross-sectional applications with iid data, researchers in economics often focus on the population ATTs. However, in our aggregate panel data setting with one treated unit, making inferences on $\tau_e$ requires restrictive stationarity assumptions on the treatment effect $\alpha_t$. Moreover, efficiency considerations are often important given the small or moderate sample sizes prevalent in SC applications. \qed
\end{rem}

\section{Robustness of the $t$-test}\label{app:robustness}
The results in \citet{bakirov2005student} and \citet{ibragimovmueller2010} show that the standard $t$-test is robust against heterogeneity in the variances of the component estimators for $\alpha\le 2 \Phi(-3) \approx 0.083.$ In our context, such heterogeneity arises when $\{u_t\}$ is not stationary, and the LRV varies across blocks within the pre-treatment period and/or between the pre-treatment and the post-treatment period. 

The existing robustness results do not directly apply to the $t$-test described in Section \ref{sec:implementation} because the component estimators are not asymptotically independent. Nevertheless, as we discuss below, the proposed $t$-test inherits certain robustness properties from the standard $t$-test. 

Suppose that $\{u_t\}$ is non-stationary with variance changing across blocks and between the pre-treatment and post-treatment period so that 
\[
\sqrt{\min\{T_0,T_1\} }\begin{pmatrix}\htau_{1}-\tau\\
\vdots\\
\htau_{K}-\tau
\end{pmatrix}\overset{d}{\rightarrow}\begin{pmatrix}\sqrt{\min\{c_{0},1\} }\sigma_0\xi_{0}-\sqrt{g_{c_0,K} }\sigma_1\xi_{1}\\
\vdots\\
\sqrt{\min\{c_{0},1\} }\sigma_0\xi_{0}-\sqrt{g_{c_0,K}}\sigma_K\xi_{K}
\end{pmatrix},
\]
where $\xi_{0},\dots,\xi_{K}$ are independent $N(0,1)$ random variables and $\sigma_0,\dots,\sigma_K$ are the block-specific LRVs. By the continuous mapping theorem, 
$
\mathbb{T}_{K}  \overset{d}{\rightarrow} \mathcal{T}_K,
$
where 
\begin{align}
\mathcal{T}_K & =\frac{\sqrt{\frac{K}{g_{c_0,K}(1+\min\{c_{0},K\})}}\left(\sqrt{\min\{c_0,1\} }\sigma_0\xi_{0}-\sqrt{g_{c_0,K}}K^{-1}\sum_{k=1}^{K}\sigma_k\xi_k\right)}{\sqrt{(K-1)^{-1}\sum_{k=1}^{K}(\sigma_k\xi_{k}-K^{-1}\sum_{k=1}^{K}\sigma_k\xi_k)^{2}}}\notag \\
& =\frac{\sqrt{\frac{K\min\{c_0,1\}}{g_{c_0,K}(1+\min\{c_{0},K\})}}\sigma_0\xi_{0}}{\sqrt{(K-1)^{-1}\sum_{k=1}^{K}(\sigma_k\xi_{k}-K^{-1}\sum_{k=1}^{K}\sigma_k\xi_k)^{2}}}-\frac{1}{\sqrt{1+\min\{c_0,K\}}}\mathbf{T}_K,\label{eq:decomposition_robustness}
\end{align}
and $\mathbf{T}_K$ is the standard $t$-statistic, which is robust to variance heterogeneity. 

The decomposition \eqref{eq:decomposition_robustness} shows that, while the proposed $t$-test inherits some robustness from the standard $t$-test due to the second term, it is not fully robust against heterogeneity in general. 

Importantly, however, the $t$-test is robust under different types of restrictions on the LRVs $\sigma_0,\dots,\sigma_K$. A leading example is when $\{u_t\}$ is non-stationary due to a structural break between the pre-treatment and post-treatment period, so that $\sigma_k=\sigma_{\text{pre}}>0$ for $k\ge 1$ and $\sigma_0=\sigma_{\text{post}}>0$ with $\sigma_{\text{pre}}\ne \sigma_{\text{post}}$. In this case, we have that
\begin{align*}
\mathcal{T}_K & =\frac{\sqrt{\frac{K}{g_{c_0,K}(1+\min\{c_{0},K\})}}\left(\sqrt{\min\{c_0,1\} }\frac{\sigma_{\text{post}}}{\sigma_{\text{pre}}}\xi_{0}-\sqrt{g_{c_0,K}}\bxi\right)}{\sqrt{(K-1)^{-1}\sum_{k=1}^{K}(\xi_{k}-\bxi_k)^{2}}}.
\end{align*}
Note that the numerator has a $N(0,\sigma_{\text{num}}^2)$ distribution, where 
\[
\sigma_{\text{num}}^2=\frac{K}{g_{c_0,K}(1+\min\{c_{0},K\})}\times\left(\min\{c_0,1\}\frac{\sigma_{\text{post}}^2}{\sigma_{\text{pre}}^2} +g_{c_0,K}K^{-1}\right)=\begin{cases}
\frac{1+\frac{\sigma_{\text{post}}^2}{\sigma_{\text{pre}}^2}c_0}{1+c_0} & {\rm if}\ 1\leq c_0 \leq K\\
\frac{1+\frac{\sigma_{\text{post}}^2}{\sigma_{\text{pre}}^2}c_0}{1+c_0} & {\rm if}\ c_{0}<1\\
\frac{1+\frac{\sigma_{\text{post}}^2}{\sigma_{\text{pre}}^2}K}{1+K} & {\rm if}\ c_{0}>K.
\end{cases}
\]
Note that $\sigma_{\text{num}}^2\le 1$ if $\sigma_{\text{post}}^2/\sigma_{\text{pre}}^2\le 1$. Therefore, by the independence of the numerator and denominator, we can condition on the denominator to show that the $t$-test remains valid in this case:
\begin{align*}
P(|\mathcal{T}_K|>t_{K-1}(1-\alpha/2))&=E\left(E\left(1\{|\mathcal{T}_K|>t_{K-1}(1-\alpha/2)\}\mid \sqrt{(K-1)^{-1}\sum_{k=1}^{K}(\xi_{k}-\bxi_k)^{2}}\right)\right)\\
&\le E\left(E\left(1\{|\mathbf{T}_K|>t_{K-1}(1-\alpha/2)\}\mid \sqrt{(K-1)^{-1}\sum_{k=1}^{K}(\xi_{k}-\bxi_k)^{2}}\right)\right)\\
&= P(|\mathbf{T}_K|>t_{K-1}(1-\alpha/2))\quad \text{for} \quad \frac{\sigma_{\text{post}}^2}{\sigma_{\text{pre}}^2}\le 1.
\end{align*}

The above derivation shows that the $t$-test remains valid when $\{u_t\}$ is non-stationary due to a decrease in the LRV between the pre- and post-treatment period ($\sigma_{\text{post}}/\sigma_{\text{pre}}\le 1$). A similar argument shows that the $t$-test will not be valid in general when the LRV increases so that $\sigma_{\text{post}}/\sigma_{\text{pre}}> 1$. 

The above derivation also suggests a simple modification of the $t$-statistic for applications where researchers are concerned about an increase in the LRV. Suppose that $\sigma_{\text{post}}^2/\sigma_{\text{pre}}^2\le B_{\sigma_{\text{post}}^2/\sigma_{\text{pre}}^2}$, which implies an upper bound on the variance of the numerator, 
$$
\sigma_{\text{num}}^2\le B_{\sigma_{\text{num}}^2}=\frac{1+B_{\sigma_{\text{post}}^2/\sigma_{\text{pre}}^2}c_0}{1+c_0}1\{c_0\le K\}+\frac{1+B_{\sigma_{\text{post}}^2/\sigma_{\text{pre}}^2}K}{1+K}1\{c_0> K\}.
$$
Then, the following modified $t$-statistic yields valid inferences:
\begin{eqnarray*}
\mathbb{T}_K = \frac{\sqrt{K}\left(\htau-\tau \right)}{\sqrt{1+\frac{Kr}{T_1}}\sqrt{B_{\sigma_{\text{num}}^2}}\sqrt{\frac{1}{K-1}\sum_{k=1}^K \left(\htau_k -\htau\right)^2}}.
\end{eqnarray*}

\section{Higher-order improvements}\label{sec: higher order}

Suppose that we observe $\{y_{t}\}_{t=1}^{T}$ from the following
Gaussian location model: 
\begin{equation}
y_{t}=\beta+u_{t},\label{eq: gauss location}
\end{equation}
where $u_{t}$ is a mean zero covariance-stationary Gaussian process
with $\sum_{h=-\infty}^{\infty}h^{2}\gamma(h)<\infty$ and $\gamma(h)=Eu_{t}u_{t-h}$.
This model has been considered for example by \citet{jansson2004error} and \citet[][Assumption 3]{sun2008optimal}. \citet[][]{sun2008optimal} show that the usual consistent estimators using Bartlett
kernel for the LRV (e.g., Newey-West estimator) would generate a test that has size
distortion at least $O(T^{-1/2})$. They also show that the fixed-$b$
approach would have size distortion of order $O(T^{-1})$. These
results are the typical explanation for why the fixed-$b$ approach
is more accurate than the classical LRV estimation approach. We now show that in this setting our self-normalization
approach also enjoys the higher-order improvement as it has size distortion
of order $O(T^{-1})$.

We first outline the inference procedure in the Gaussian location
model in (\ref{eq: gauss location}). Let $K\geq2$ be a fixed integer
and define $G=T/K$. For simplicity, we assume that $G$ is an integer.
Then we construct blocks $\{H_{k}\}_{k=1}^{K}$ with $H_{k}=\{(k-1)G+1,\dots,kG\}$.
Let $\hbeta_{k}=G^{-1}\sum_{t\in H_{k}}y_{t}$ for $1\le k\leq K$
and $\bbeta=K^{-1}\sum_{k=1}^{K}\hbeta_{k}$. For testing $H_{0}: \beta=0$,
we define the test statistic
\[
\mathbf{T}_K=\frac{\sqrt{K}\bbeta}{\sqrt{(K-1)^{-1}\sum_{k=1}^{K}(\hbeta_{k}-\bbeta)^{2}}}.
\]
The critical value is $t_{K-1}(1-\alpha/2)$. 
\begin{thm}
	\label{thm: higher-order t}Consider the model in (\ref{eq: gauss location}).
	Suppose that $H_{0}: \beta=0$ holds. Then 
	\[
	\left|P\left(|\mathbf{T}_K|>t_{K-1}(1-\alpha/2)\right)-\alpha\right|=O(T^{-1}).
	\]
\end{thm}

By Theorem \ref{thm: higher-order t}, the ``cross-fitted'' self-normalized $t$-test also has size distortion $O(T^{-1})$ in the Gaussian location model \eqref{eq: gauss location}. Although it is quite difficult to derive the higher-order asymptotics outside the Gaussian location model, Theorem \ref{thm: higher-order t} suggests that, for size considerations, the cross-fitted self-normalized test is expected to have similar properties as fixed-$b$ methods.
However, in high-dimensional settings, we are not aware of any existing results that would allow for establishing the validity of the fixed-$b$ approach in our context.

\section{Calculations expected length}
\label{app: expected length}
The length of the confidence interval is 
$$
L=2 t_{K-1}(1-\alpha/2)\frac{1}{\sqrt{K}}\sqrt{1+\frac{Kr}{T_1}}\sqrt{\frac{1}{K-1}\sum_{k=1}^K \left(\htau_k -\htau\right)^2}.
$$
Consider 
\begin{eqnarray*}
\sqrt{\min\{T_0,T_1\}}L&=&2 t_{K-1}(1-\alpha/2)\frac{1}{\sqrt{K}\sqrt{K-1}}\sqrt{1+\frac{Kr}{T_1}}\sqrt{\min\{T_0,T_1\}\sum_{k=1}^K \left(\htau_k -\htau\right)^2}.
\end{eqnarray*}
To derive the limiting distribution for fixed $K$, note that 
$$
\min\{T_0,T_1\}\sum_{k=1}^K \left(\htau_k -\htau\right)^2=\sum_{k=1}^K\left(\sqrt{\min\{T_0,T_1\}}(\htau_k-\tau)-\frac{1}{K}\sum_{k=1}^K\sqrt{\min\{T_0,T_1\}}(\htau_k-\tau)\right)^2
$$
and that $Kr/T_1 \rightarrow \min \{c_0,K \} $ as $T_0,T_1 \rightarrow \infty$.

Therefore, by Theorem \ref{thm: asy distr} and the continuous mapping theorem,
$
\sqrt{\min\{T_0,T_1\}}L\overset{d}\rightarrow \mathcal{L},
$
where
\begin{eqnarray*}
\mathcal{L}&=&2 \sigma t_{K-1}(1-\alpha/2)\frac{1}{\sqrt{K}\sqrt{K-1}}\sqrt{1+\min \{c_0,K \}}\\
&&\times \sqrt{\sum_{k=1}^K \left(\sqrt{\min\{c_{0},1\} }\xi_{0}-\sqrt{g_{c_0,K} }\xi_{k}-\frac{1}{K}\sum_{k=1}^K\sqrt{\min\{c_{0},1\} }\xi_{0}-\sqrt{g_{c_0,K} }\xi_{k}\right)^2}\\
&=&2\sigma  t_{K-1}(1-\alpha/2)\frac{1}{\sqrt{K}\sqrt{K-1}}\sqrt{1+\min \{c_0,K \}}\sqrt{g_{c_0,K} }\sqrt{\sum_{k=1}^K \left(\xi_{k}-\frac{1}{K}\sum_{k=1}^K\xi_{k}\right)^2}
\end{eqnarray*}
Thus, it remains to compute the expectation of
$$
\sqrt{\sum_{k=1}^K \left(\xi_{k}-\frac{1}{K}\sum_{k=1}^K\xi_{k}\right)^2}
$$
Because
$$
\sum_{k=1}^K \left(\xi_{k}-\frac{1}{K}\sum_{k=1}^K\xi_{k}\right)^2\sim \chi^2_{K-1},
$$
it follows that
$$
\sqrt{\sum_{k=1}^K \left(\xi_{k}-\frac{1}{K}\sum_{k=1}^K\xi_{k}\right)^2}\sim \chi_{K-1},
$$
where $\chi_{K-1}$ denotes a chi-distribution with $K-1$ degrees of freedom. Therefore,
$$
E\left(\sqrt{\sum_{k=1}^K \left(\xi_{k}-\frac{1}{K}\sum_{k=1}^K\xi_{k}\right)^2} \right)=\sqrt{2}\frac{\Gamma(K/2)}{\Gamma((K-1)/2)},
$$
where $\Gamma(\cdot)$ denotes the Gamma function, and
$$
E(\mathcal{L})=2\sigma  t_{K-1}(1-\alpha/2)\frac{1}{\sqrt{K}\sqrt{K-1}}\sqrt{1+\min \{c_0,K \}} \sqrt{g_{c_0,K}}\sqrt{2}\frac{\Gamma(K/2)}{\Gamma((K-1)/2)}.
$$

Next, we provide an approximation of the limit of $E(\mathcal{L})$ as $K\rightarrow \infty$. Since $K\rightarrow\infty$, assuming that $c_0<\infty$, we have that for large $K$, 
\[
g_{c_{0},K}=K\oneb\{c_{0}<1\}+(K/c_{0})\oneb\{1\leq c_{0}\leq K\}=\min\{c_{0}^{-1},1\}K.
\]
Note further that we take $T_1\rightarrow \infty$ before taking $K\rightarrow\infty$ so that $1+Kr/T_{1}=1$. To understand what happens when $K\rightarrow\infty$, we use Stirling's approximation of the Gamma function:  $\Gamma(n+1)=(1+o(1))\times\sqrt{2\pi n}(n/e)^{n}$ as $n\rightarrow \infty$.
This means that 
\begin{multline*}
\frac{\Gamma(n+3/2)}{\Gamma(n+1)}=(1+o(1))\cdot\frac{((n+1/2)/e)^{n+1/2}}{(n/e)^{n}}=(1+o(1))\cdot\left(\frac{n+1/2}{n}\right)^{n}\cdot\left(\frac{n+1/2}{e}\right)^{1/2}\\
=(1+o(1))\cdot e^{1/2}\left(\frac{n+1/2}{e}\right)^{1/2}=(1+o(1))\cdot\sqrt{n}.
\end{multline*}
Therefore, and using that $t_{K-1}(1-\alpha)\rightarrow \Phi^{-1}(1-\alpha/2)$ as $K\rightarrow \infty$, we have the following approximation as $K\rightarrow \infty$,
\begin{align*}
E(\mathcal{L}) & =2\sigma \Phi^{-1}(1-\alpha/2)\frac{1}{\sqrt{K}\sqrt{K-1}}\sqrt{g_{c_{0},K}}\sqrt{2}\frac{\Gamma(K/2)}{\Gamma((K-1)/2)} \sqrt{1+\min \{c_0,K \} }\\
 & =(1+o(1))\cdot2\sigma \Phi^{-1}(1-\alpha/2)\frac{1}{\sqrt{K}\sqrt{K-1}}\sqrt{\min\{c_{0}^{-1},1\}K}\sqrt{2}\cdot\sqrt{K/2} \sqrt{1+c_0 }\\
 & =(1+o(1))\cdot2\sigma \Phi^{-1}(1-\alpha/2)\sqrt{\min\{c_{0}^{-1},1\}} \sqrt{1+c_0 }.
\end{align*}

\section{Verification of Assumption \ref{assu: non-stationary deviation}}
\label{app:verification}

Here we illustrate the verification of Assumption \ref{assu: non-stationary deviation} based on two leading examples.

\begin{example}[Unit root]
	Assume that $\xi_{t}=\sum_{i=1}^{t}v_{t}$ with a mean-zero stationary
	process $v_{t}$. Then by the functional central limit theorem, we
	can show that $T^{-1/2}\xi_{[rT]}$ tends to $B(r)$, where $B(\cdot)$
	is a standard Brownian motion. Thus, the continuous mapping theorem
	and the condition of $T_{1}\ll T_{0}$ imply that $T_{1}^{-1}T^{-1/2}\sum_{t=T_{0}+1}^{T}\xi_{t}\overset{d}\rightarrow B(1)$
	and $T^{-2}\sum_{t\in H_{(-k)}}\xi_{t}^{2}=T^{-2}\sum_{t=1}^{T_{0}}\xi_{t}^{2}-T^{-2}\sum_{t\in H_{(k)}}\xi_{t}^{2}\overset{d}\rightarrow\int_{0}^{1}B^{2}(s)ds$ jointly.
By the continuous mapping theorem, this verifies the first part of Assumption \ref{assu: non-stationary deviation}. The second part follows by a similar calculation. 
 \qed
\end{example}

\begin{example}[Polynomial trend]
	Assume that $\xi_{t}=\sum_{k=1}^{M}c_{k}t^{k}$ with constants $c_{1},\dots,c_{M}$,
	where $M\geq1$ is fixed. ($M=1$ corresponds to a linear trend and
	$M=2$ represents a quadratic trend.) Recall the elementary relation
	of $(\sum_{i=1}^{n}i^{\gamma})/(n^{1+\gamma}/(1+\gamma))\rightarrow1$
	as $n\rightarrow\infty$ for any fixed $\gamma>0$. By $T_{1}\ll T_{0}$,
	 the first part of Assumption \ref{assu: non-stationary deviation} holds:
	\begin{eqnarray*}
	\frac{\sum_{t=T_{0}+1}^{T}\xi_{t}}{\sqrt{\sum_{t\in H_{(-k)}}\xi_{t}^{2}}}&=&(1+o(1))\cdot\frac{(1+M)^{-1}\left[(1+T_{1}/T_{0})^{M+1}-1\right]T_{0}^{M+1}}{\sqrt{(1+2M)^{-1}T_{0}^{2M+1}}}\\
	&=&O(T_{1}T^{-1/2}).
	\end{eqnarray*}
  The second part  of Assumption \ref{assu: non-stationary deviation}  follows by a similar calculation. 
	 \qed
\end{example}

\section{Proofs}
We will sometimes write $Y_t$ instead of $Y_t(0)$ to simplify the exposition.
\subsection{Proof of Lemma \ref{lem: algebra}}
Let $\mu=E(X_{t})$. Notice that for $T_{0}+1\leq t\leq T$, 
\[
Y_{t}-X_{t}'\hw_{(k)}=\alpha_{t}+u_{t}-X_{t}'\Delta_{(k)}=\alpha_{t}+u_{t}-\mu'\Delta_{(k)}-\tX_{t}'\Delta_{(k)}
\]
and for $t\in H_{k}$,
\[
Y_{t}(0)-X_{t}'\hw_{(k)}=u_{t}-X_{t}'\Delta_{(k)}=u_{t}-\mu'\Delta_{(k)}-\tX_{t}'\Delta_{(k)}.
\]
Therefore, 
\begin{align*}
\htau_{k}-\tau & =T_{1}^{-1}\sum_{t=T_{0}+1}^{T}(Y_{t}-X_{t}'\hw_{(k)})-\frac{1}{|H_{k}|}\sum_{t\in H_{k}}(Y_{t}-X_{t}'\hat{w}_{(k)})-\tau\\
 & =T_{1}^{-1}\sum_{t=T_{0}+1}^{T}(\alpha_{t}+u_{t}-\mu'\Delta_{(k)}-\tX_{t}'\Delta_{(k)})-\frac{1}{|H_{k}|}\sum_{t\in H_{k}}(u_{t}-\mu'\Delta_{(k)}-\tX_{t}'\Delta_{(k)})-\tau\\
 & =T_{1}^{-1}\sum_{t=T_{0}+1}^{T}u_{t}-\frac{1}{|H_{k}|}\sum_{t\in H_{k}}u_{t}+\frac{1}{|H_{k}|}\sum_{t\in H_{k}}\tX_{t}'\Delta_{(k)}-T_{1}^{-1}\sum_{t=T_{0}+1}^{T}\tX_{t}'\Delta_{(k)}.
\end{align*}
The proof is complete. \qed

\subsection{Proof of Lemma \ref{lem: consistency misspecification}}

We prove the results for the following more general estimator:
\begin{equation}
\hat{w}_{(k)}\in \underset{w\in \mathcal{W}}{\arg\min}\sum_{t\in H_{(-k)}}(Y_{t}-X_{t}'w)^{2}\label{eq:unified_estimator},
\end{equation}
where $\mathcal{W}$ is a subset of $\left\{w:\|w\|_1\le Q\right\}$ with $Q=O(1)$. For the canonical SC estimator, we see that $\mathcal{W}=\{w=(w_1,\dots,w_N)': w_i\geq 0,\ \sum_{i=1}^N w_i=1 \}$ and thus $Q=1$. 

\medskip

	For simplicity, we write $\hw=\hat{w}_{(k)} $, $\hmu=\hmu_{(-k)}$
	and $\hSigma=\hSigma_{(-k)}$. Let $\xi=(\hmu-\hSigma w_{*})-(\mu-\Sigma w_{*})$.

Since $\|w_{*}\|_\infty \leq Q$, we have that $\|\xi\|_{\infty}\leq \|\hmu - \mu\|_{\infty}+  \|(\hSigma-\Sigma) w_{*}\|_{\infty} \leq \|\hmu - \mu\|_{\infty}+ Q \|\hSigma-\Sigma \|_{\infty} $. Thus,  by assumption, $\|\xi\|_{\infty}=o_P(1)$. 
	
	We rewrite 
	\begin{align*}
	w_{*} =\arg\min_{v}\ v'\Sigma v-2\mu'v\qquad s.t.\qquad v\in \mathcal{W}
	\end{align*}
	and 
	\begin{align*}
	\hw  \in \arg\min_{v}\ v'\hSigma v-2\hmu'v\qquad s.t.\qquad v\in \mathcal{W}
	\end{align*}
	
	Let $\Delta=\hw-w_{*}$. 	For any $\lambda\in[0,1]$, define $w_{\lambda}=w_{*}+\lambda\Delta$.
	Then by the definition of $w_{*}$, we have that $w_{\lambda}'\Sigma w_{\lambda}-2\mu' w_{\lambda}\geq w_{*}'\Sigma w_{*}-2\mu' w_{*}$,
	which means 
	\[
	\lambda^{2}\Delta'\Sigma\Delta\geq2\lambda(\mu-\Sigma w_{*})'\Delta.
	\]
	Thus, for any $\lambda\in(0,1)$, we have that $\lambda\Delta'\Sigma\Delta\geq 2(\mu-\Sigma w_{*})'\Delta$.
	Since this holds for any $\lambda\in(0,1)$, we have that 
	\begin{equation}
	(\mu-\Sigma w_{*})'\Delta\leq0.\label{eq: key bnd misspecifiation}
	\end{equation}
	Now by definition, we have that 
	\[
	\hw'\hSigma\hw-2\hmu'\hw\le w_{*}'\hSigma w_{*}-2\hmu' w_{*}.
	\]
	It follows that 
	\[
	\Delta'\hSigma\Delta\leq2(\hmu-\hSigma w_{*})'\Delta.
	\]
	Notice that 
	$$
	(\hmu-\hSigma w_{*})'\Delta =(\mu-\Sigma w_{*})'\Delta+\xi'\Delta \overset{\text{(i)}}{\leq}\xi'\Delta
	\leq\|\xi\|_{\infty}\|\Delta\|_{1}\leq2\|\xi\|_{\infty}Q,
	$$
	where (i) holds by (\ref{eq: key bnd misspecifiation}). The above two displays imply
	that 
	\[
	\Delta'\hSigma\Delta\leq2\|\xi\|_{\infty}Q.
	\]
	Therefore, 
	\begin{eqnarray*}
	2\|\xi\|_{\infty}Q \geq\Delta'\Sigma\Delta+\Delta'(\hSigma-\Sigma)\Delta&\geq&\Delta'\Sigma\Delta-\|\hSigma-\Sigma\|_{\infty}\|\Delta\|_{1}^{2}\\
	&\geq& c\|\Delta\|_{2}^{2}-4\|\hSigma-\Sigma\|_{\infty}Q^2.
	\end{eqnarray*}
	The desired result follows by recalling $\|\xi\|_{\infty}\leq \|\hmu - \mu\|_{\infty}+ Q \|\hSigma-\Sigma \|_{\infty} $. \qed

\subsection{Proof of Lemma \ref{lem:consistency_common_nonstationarity}}

By Equations \eqref{eq:sc_estimator_common_nonstationarity} and \eqref{eq:w_ast_common_nonstationarity} the result follows from the same arguments as in the proof of Lemma \ref{lem: consistency misspecification} with $\mathcal{W}=\mathcal{W}^{SC}.$ \qed

\subsection{Proof of Lemma \ref{lem:sufficient_stationarity}}
Note that
\begin{align*}
Y_{t}(0) & =\sum_{i=0}^{D}\mathbf{1}\{D=i\}Y_{it}(0)\\
 & =\sum_{i=0}^{D}\mathbf{1}\{D=i\}(\theta_{t}+L_{i}'F_{t}+\varepsilon_{it})\\
 & =\theta_{t}+F_{t}'\left(\sum_{i=0}^{D}\mathbf{1}\{D=i\}L_{i}\right)+\sum_{i=0}^{D}\mathbf{1}\{D=i\}\varepsilon_{it}\\
 & =V_{t}(0)+\theta_{t},
\end{align*}
where $V_{t}(0)=F_{t}'\left(\sum_{i=0}^{D}\mathbf{1}\{D=i\}L_{i}\right)+\sum_{i=0}^{D}\mathbf{1}\{D=i\}\varepsilon_{it}$.
For $j\neq D$, we have that 
\[
Y_{jt}(0)=\theta_{t}+L_{j}'F_{t}+\varepsilon_{jt}=\theta_{t}+Z_{jt},
\]
where $Z_{jt}=L_{j}'F_{t}+\varepsilon_{jt}$. 

Next, note that Assumption \ref{assu: unconfoundedness factor} implies $(L,D)\indep (F,\varepsilon)$.
Let $\mathcal{F}$ be $\sigma$-algebra generated by $(L,D)$. Since
$(L,D)$ is independent of $(F,\varepsilon)$ and $(F_{t},\varepsilon_{t})$
is stationary, we have that conditional on $\mathcal{F}$, the following
is stationary: 
\[
\left[F_{t}'\left(\sum_{i=0}^{D}\mathbf{1}\{D=i\}L_{i}\right)+\sum_{i=0}^{D}\mathbf{1}\{D=i\}\varepsilon_{it},\{L_{j}'F_{t}+\varepsilon_{jt}\}_{j\neq D}\right].
\]
Stationarity conditional on $(L,D)$ then implies stationarity conditional on $D$ due to the law of iterated expectations. 

\subsection{Proof of Lemma \ref{lem: classical SC}}

Let $\Delta=\tw-\bw$ and $e=Z_{0}-Z_{-0}\bw$. 
Since $\bw\in\mathcal{W}^{SC}$, the definition of $\tw$ implies
that 
\begin{equation}
\left\Vert \begin{pmatrix}Y_{0}\\
Z_{0}
\end{pmatrix}-\begin{pmatrix}Y_{-0}\\
Z_{-0}
\end{pmatrix}\tw\right\Vert _{V}^{2}\leq\left\Vert \begin{pmatrix}Y_{0}\\
Z_{0}
\end{pmatrix}-\begin{pmatrix}Y_{-0}\\
Z_{-0}
\end{pmatrix}\bw\right\Vert _{V}^{2}.\label{eq: basic real sc inequality part 1}
\end{equation}

The right-hand side is $\left\Vert \begin{pmatrix}u\\
e
\end{pmatrix}\right\Vert _{V}^{2}=u'V_{1}u+e'V_{2}e$. For the left-hand side, we have 
\[
\begin{pmatrix}Y_{0}\\
Z_{0}
\end{pmatrix}-\begin{pmatrix}Y_{-0}\\
Z_{-0}
\end{pmatrix}\bw=\begin{pmatrix}Y_{-0}\bw+u\\
Z_{-0}\bw+e
\end{pmatrix}-\begin{pmatrix}Y_{-0}\\
Z_{-0}
\end{pmatrix}\hw=\begin{pmatrix}u-Y_{-0}\Delta\\
e-Z_{-0}\Delta
\end{pmatrix}\overset{\text{(i)}}{=}\begin{pmatrix}u-\tY_{-0}\Delta\\
e-Z_{-0}\Delta
\end{pmatrix}.
\]
where (i) follows by $Y_{-0}\Delta=(\delta\oneb_{N}'+\tY_{-0})\Delta=\tY_{-0}\Delta$
since $\oneb_{N}'\Delta=\oneb_{N}'(\tw-\bw)=1-1=0$, where $\delta=(\delta_{1},\dots,\delta_{T_{0}})'$.
Thus, \eqref{eq: basic real sc inequality part 1} implies that 
\[
\left(u-\tY_{-0}\Delta\right)'V_{1}\left(u-\tY_{-0}\Delta\right)+\left(e-Z_{-0}\Delta\right)'V_{2}\left(e-Z_{-0}\Delta\right)\leq u'V_{1}u+e'V_{2}e.
\]

Rearranging the terms, we obtain
\begin{align}
\Delta'\left(T_{0}^{-1}\tY_{-0}'V_{1}\tY_{-0}\right)\Delta+\Delta'\left(T_{0}^{-1}Z_{-0}'V_{2}Z_{-0}\right)\Delta & \leq2\left(T_{0}^{-1}u'V_{1}\tY_{-0}\right)\Delta+2\left(T_{0}^{-1}e'V_{2}Z_{-0}\right)\Delta.\label{eq: basic real sc inequality part 2}
\end{align}

We observe that for any $w$,
\begin{align*}
 & T_{0}^{-1}E\left\Vert \begin{pmatrix}Y_{0}\\
Z_{0}
\end{pmatrix}-\begin{pmatrix}Y_{-0}\\
Z_{-0}
\end{pmatrix}w\right\Vert _{V}^{2}\\
 & =T_{0}^{-1}E\left\Vert \begin{pmatrix}\tY_{0}\\
Z_{0}
\end{pmatrix}-\begin{pmatrix}\tY_{-0}\\
Z_{-0}
\end{pmatrix}w\right\Vert _{V}^{2}\\
 & =T_{0}^{-1}E\|\tY_{0}-\tY_{-0}w\|_{V_{1}}^{2}+T_{0}^{-1}\|Z_{0}-Z_{-0}w\|_{V_{2}}^{2}\\
 & =T_{0}^{-1}E\|\tY_{0}\|_{V_{1}}^{2}-2\left(T_{0}^{-1}E\tY_{0}'V_{1}\tY_{-0}\right)w+w'\left(T_{0}^{-1}E\tY_{-0}'V_{1}\tY_{-0}\right)w+T_{0}^{-1}\|Z_{0}-Z_{-0}w\|_{V_{2}}^{2}.
\end{align*}

Let $w_{\lambda}=\bw+\lambda\Delta$. Notice that $w_{\lambda}\in\mathcal{W}^{SC}$
for any $\lambda\in[0,1]$. Thus, by the definition of $\bw$, we
have that 
\[
\left.\frac{\partial}{\partial\lambda}\left(-2\left(T_{0}^{-1}E\tY_{0}'V_{1}\tY_{-1}\right)w_{\lambda}+w_{\lambda}'\left(T_{0}^{-1}E\tY_{-0}'V_{1}\tY_{-0}\right)w_{\lambda}+T_{0}^{-1}\|Z_{0}-Z_{-0}w_{\lambda}\|_{V_{2}}^{2}\right)\right|_{\lambda=0}\geq0.
\]

This means that 
\[
-2\left(T_{0}^{-1}E\tY_{0}'V_{1}\tY_{-0}\right)\Delta+2\bw'\left(T_{0}^{-1}E\tY_{-0}'V_{1}\tY_{-0}\right)\Delta-2T_{0}^{-1}\Delta Z_{-0}'V_{2}(Z_{0}-Z_{-0}\bw)\geq0.
\]

In other words, 
\begin{equation}
\left(T_{0}^{-1}Eu'V_{1}\tY_{-0}\right)\Delta+T_{0}^{-1}e'V_{2}Z_{-0}\Delta\leq0.\label{eq: basic real sc inequality part 3}
\end{equation}

Hence, by the assumption of $\|T_{0}^{-1}u'V_{1}\tY_{-0}-T_{0}^{-1}Eu'V_{1}\tY_{-0}\|_{\infty}=o_{P}(1)$,
we have that 
\begin{align*}
 & \left(T_{0}^{-1}u'V_{1}\tY_{-0}\right)\Delta+T_{0}^{-1}e'V_{2}Z_{-0}\Delta\\
 & =\left(T_{0}^{-1}Eu'V_{1}\tY_{-0}\right)\Delta+T_{0}^{-1}e'V_{2}Z_{-0}\Delta+\left(T_{0}^{-1}u'V_{1}\tY_{-0}-T_{0}^{-1}Eu'V_{1}\tY_{-0}\right)\Delta\\
 & \overset{\text{(i)}}{\leq}\left(T_{0}^{-1}u'V_{1}\tY_{-0}-T_{0}^{-1}Eu'V_{1}\tY_{-0}\right)\Delta\\
 & \leq\|T_{0}^{-1}u'V_{1}\tY_{-0}-T_{0}^{-1}Eu'V_{1}\tY_{-0}\|_{\infty}\|\Delta\|_{1}\overset{\text{(ii)}}{\leq}o_{P}(1),
\end{align*}
where (i) follows by (\ref{eq: basic real sc inequality part 3})
and (ii) follows by $\|\Delta\|_{1}\leq\|\tw\|_{1}+\|\bw\|_{1}=2$.
This and (\ref{eq: basic real sc inequality part 2}) imply 
\[
\Delta'\left(T_{0}^{-1}\tY_{-0}'V_{1}\tY_{-0}\right)\Delta+\Delta'\left(T_{0}^{-1}Z_{-0}'V_{2}Z_{-0}\right)\Delta\leq o_{P}(1).
\]

On the other hand, we have
\begin{align*}
 & \Delta'\left(T_{0}^{-1}\tY_{-0}'V_{1}\tY_{-0}\right)\Delta\\
 & =\Delta'\left(T_{0}^{-1}E\tY_{-0}'V_{1}\tY_{-0}\right)\Delta+\Delta'\left(T_{0}^{-1}\tY_{-0}'V_{1}\tY_{-0}-T_{0}^{-1}E\tY_{-0}'V_{1}\tY_{-0}\right)\Delta\\
 & \geq\Delta'\left(T_{0}^{-1}E\tY_{-0}'V_{1}\tY_{-0}\right)\Delta-\|T_{0}^{-1}\tY_{-0}'V_{1}\tY_{-0}-T_{0}^{-1}E\tY_{-0}'V_{1}\tY_{-0}\|_{\infty}\|\Delta\|_{1}^{2}\\
 & =\Delta'\left(T_{0}^{-1}E\tY_{-0}'V_{1}\tY_{-0}\right)\Delta-o_{P}(1).
\end{align*}

The above two displays imply 
\[
\Delta'\left(T_{0}^{-1}E\tY_{-0}'V_{1}\tY_{-0}\right)\Delta+\Delta'\left(T_{0}^{-1}Z_{-0}'V_{2}Z_{-0}\right)\Delta=o_{P}(1).
\]

Since $\Delta'\left(T_{0}^{-1}Z_{-0}'V_{2}Z_{-0}\right)\Delta\geq0$
and the smallest eigenvalue of $T_{0}^{-1}E\tY_{-0}'V_{1}\tY_{-0}$
is bounded away from zero, we have $\|\Delta\|_{2}=o_{P}(1)$.

\subsection{Proof of Theorem \ref{thm: asy distr}}

	By Lemma \ref{lem: algebra}, we have that 
	\begin{multline}\label{eq: lem main bnd 4}
\left|\hat{\tau}_{k}-\tau-\left(\frac{1}{T_{1}}\sum_{t=T_{0}+1}^{T}u_{t}-\frac{1}{|H_{k}|}\sum_{t\in H_{k}}u_{t}\right)\right|
	\leq  \left| \frac{1}{|H_{k}|}\sum_{t\in H_{k}}\tX_{t}'\Delta_{(k)}\right|+ \left|\frac{1}{T_{1}}\sum_{t=T_{0}+1}^{T}\tX_{t}'\Delta_{(k)}\right|,
	\end{multline}
	where $\tilde{X}_{t}=X_{t}-E(X_{t})$ 
	and $\Delta_{(k)}=\hw_{(k)}-w_\ast$.\footnote{We note that working with $\tilde{X}_t$ instead of $X_t$ is key for our theoretical argument since $E(\tilde{X}_{t})=0$.} We now bound the two terms on the right-hand side of (\ref{eq: lem main bnd 4}). 
	
Fix $k\in\{1,\dots,K\}$. Define $B_{k}$ to be the ``two-sided buffer'',
i.e., the set that contains the smallest $\gamma_{T}$ numbers and
the largest $\gamma_{T}$ numbers in  $H_{k}=\{(k-1)r+1,\dots,kr\}$ for $1\leq k\le K$. Also define $A_{k}=H_{k}\backslash B_{k}$,
i.e., $A_{k}=\{t:\ (k-1)r+1+\gamma_{T}+1\leq t\leq kr-\gamma_{T}\}$.
Thus, 
\[
\sum_{t\in H_{k}}\tilde{X}_{t}'\Delta_{(k)}=\sum_{t\in A_{k}}\tilde{X}_{t}'\Delta_{(k)}+\sum_{t\in B_{k}}\tilde{X}_{t}'\Delta_{(k)}.
\]

The second term can be bounded by 
\[
\left|\sum_{t\in B_{k}}\tilde{X}_{t}'\Delta_{(k)}\right|\leq\max_{1\leq t\leq T_{0}}\|\tilde{X}_{t}\|_{\infty}\|\Delta_{(k)}\|_{1}|B_{k}|=2\gamma_{T}\max_{1\leq t\leq T_{0}}\|\tilde{X}_{t}\|_{\infty}\|\Delta_{(k)}\|_{1}.
\]

Thus, 
\begin{equation}
P\left(\left|\sum_{t\in B_{k}}\tilde{X}_{t}'\Delta_{(k)}\right|\leq2\rho_{T}\gamma_{T}\|\Delta_{(k)}\|_{1}\right)\rightarrow1.\label{eq: thm main bnd eq 1}
\end{equation}

On the other hand, we use Berbee's coupling to bound $\sum_{t\in A_{k}}\tilde{X}_{t}'\Delta_{(k)}$.
By Theorem 16.2.1 of \citet{athreya2006measure}, on an enlarged probability
space, there exist random variables $\{\bar{X}_{t}\}_{t\in A_{k}}$
such that (1) $\{\bar{X}_{t}\}_{t\in A_{k}}$ and $\{\tilde{X}_{t}\}_{t\in A_{k}}$
have the same distribution, (2) $\{\bar{X}_{t}\}_{t\in A_{k}}$ is
independent of data in $\{1,...,T_{0}\}\backslash H_{k}$ and (3)
$P(\{\bar{X}_{t}\}_{t\in A_{k}}\neq\{\tilde{X}_{t}\}_{t\in A_{k}})\leq\betamix(\gamma_{T})$.
Notice that $\Delta_{(k)}$ is independent of $\{\bar{X}_{t}\}_{t\in A_{k}}$.
Hence, 
\begin{eqnarray*}
E\left[\left(\sum_{t\in A_{k}}\bar{X}_{t}'\Delta_{(k)}\right)^{2}\mid\Delta_{(k)}\right]&=&\Delta_{(k)}'E\left[\left(\sum_{t\in A_{k}}\bar{X}_{t}\right)\left(\sum_{t\in A_{k}}\bar{X}_{t}\right)'\right]\Delta_{(k)}\\
&\overset{\text{(i)}}{\leq}&|A_{k}|\kappa_{1}\|\Delta_{(k)}\|_{2}^{2},
\end{eqnarray*}
where (i) follows by Assumption \ref{assu: weak dependence} and the
fact that $\{\bar{X}_{t}\}_{t\in A_{k}}$ and $\{\tilde{X}_{t}\}_{t\in A_{k}}$
have the same distribution. Thus, $\sum_{t\in A_{k}}\bar{X}_{t}'\Delta_{(k)}=O_{P}(\sqrt{|A_{k}|}\|\Delta_{(k)}\|_{2})$.
Since $P(\{\bar{X}_{t}\}_{t\in A_{k}}\neq\{\tilde{X}_{t}\}_{t\in A_{k}})\leq\betamix(\gamma_{T})=o(1)$,
it follows that 
\begin{equation}
\sum_{t\in A_{k}}\tilde{X}_{t}'\Delta_{(k)}=O_{P}(\sqrt{|A_{k}|}\|\Delta_{(k)}\|_{2}).\label{eq: thm main bnd eq 2}
\end{equation}

Now by (\ref{eq: thm main bnd eq 1}) and (\ref{eq: thm main bnd eq 2}),
\begin{eqnarray*}
\sum_{t\in H_{k}}\tilde{X}_{t}'\Delta_{(k)}&=&O_{P}\left(\rho_{T}\gamma_{T}\|\Delta_{(k)}\|_{1}+\sqrt{|A_{k}|}\|\Delta_{(k)}\|_{2}\right)\\
&\overset{\text{(i)}}{=}&O_{P}\left(\rho_{T}\gamma_{T} \|\Delta_{(k)}\|_{1}+\sqrt{r} \|\Delta_{(k)}\|_{2}\right),
\end{eqnarray*}
where (i) follows by the assumption that $\gamma_{T}=o(T_{0})$ and $|A_k|<|H_{k}|=r$.
Similarly, we can show
\[
T_{1}^{-1}\sum_{t=T_{0}+1}^{T}\tilde{X}_{t}'\Delta=O_{P}\left(T_{1}^{-1}\rho_{T}\gamma_{T} \|\Delta_{(k)}\|_{1}+T_{1}^{-1/2} \|\Delta_{(k)}\|_{2}\right).
\]
Since $r\asymp \min\{T_0,T_1\} $, the above two displays and (\ref{eq: lem main bnd 4}) imply 
	\begin{equation*}
\left|\hat{\tau}_{k}-\tau-\left(\frac{1}{T_{1}}\sum_{t=T_{0}+1}^{T}u_{t}-\frac{1}{|H_{k}|}\sum_{t\in H_{k}}u_{t}\right)\right|=O_{P}\left(\frac{\rho_{T}\gamma_{T} \|\Delta_{(k)}\|_{1}}{\min\{T_0,T_1\}  }+ \frac{\|\Delta_{(k)}\|_{2}}{\sqrt{\min\{T_0,T_1\}  }} \right).
\end{equation*}

By   Assumptions \ref{ass:ell_2_consistency} and \ref{assu: weak dependence} as well as the covariance-stationarity of $u_t$, we have
	\begin{equation*}
\sqrt{\min\{T_0,T_1\} }\left|\hat{\tau}_{k}-\tau-\left(\frac{1}{T_{1}}\sum_{t=T_{0}+1}^{T}\tilde{u}_{t}-\frac{1}{|H_{k}|}\sum_{t\in H_{k}}\tilde{u}_{t}\right)\right|=o_{P}\left(1 \right),
\end{equation*}
where $\tilde{u}_t=u_t-E(u_t)$. 
The desired result in Theorem \ref{thm: asy distr} follows from Assumption \ref{ass:ell_2_consistency} and the usual CLT for dependent processes, e.g., Theorem 5.20 of \cite{white2001asymptotic}.

\subsection{Proof of Theorem \ref{thm: t statistic}}
Part (i) is a direct consequence of Theorem \ref{thm: asy distr}. For Part (ii), recall that 
\[
\mathbb{T}_{K}=\frac{\sqrt{K}(\htau-\tau_{0})}{\sqrt{1+\frac{Kr}{T_{1}}}\times\sqrt{(K-1)^{-1}\sum_{k=1}^{K}(\htau_{(k)}-\htau)^{2}}}.
\]

We notice that $r/T_1\rightarrow \min\{c_0/K,1\} $. Then Theorem \ref{thm: asy distr} and the continuous mapping theorem imply that   $\mathbb{T}_{K}\overset{d}{\rightarrow}\mathcal{T}_{K}$, where 
\begin{align*}
\mathcal{T}_{K} & =\frac{\sqrt{K}\left(\sqrt{\min\{c_0,1\} }\xi_{0}-\sqrt{g_{c_0,K}}\bxi\right)}{\sqrt{1+\min\{c_{0},K\}}\times\sqrt{(K-1)^{-1}\sum_{k=1}^{K}g_{c_0,K}(\xi_{k}-\bxi)^{2}}}\\
& =\frac{\sqrt{\frac{K}{g_{c_0,K}(1+\min\{c_{0},K\})}}\left(\sqrt{\min\{c_0,1\} }\xi_{0}-\sqrt{g_{c_0,K}}\bxi\right)}{\sqrt{(K-1)^{-1}\sum_{k=1}^{K}(\xi_{k}-\bxi)^{2}}},
\end{align*}
where  $\bxi=K^{-1}\sum_{k=1}^{K}\xi_k$.

Notice that $\bxi$ is independent of $\sum_{k=1}^{K}(\xi_{k}-\bxi)^{2}$.
Therefore, the numerator is independent of the denominator. The variance
of the numerator is 
\[
\frac{K}{g_{c_0,K}(1+\min\{c_{0},K\})}\times\left(\min\{c_0,1\} +g_{c_0,K}K^{-1}\right)=\begin{cases}
1 & {\rm if}\ 1\leq c_0 \leq K\\
1 & {\rm if}\ c_{0}<1\\
1 & {\rm if}\ c_{0}>K.
\end{cases}
\]

Hence, the numerator has a standard normal distribution. It follows
that $\Tcal_{K}\sim t_{K-1}$. The proof is complete. \qed

\subsection{Proof of Theorem \ref{thm: common nonstationarity}}
We start with an auxiliary lemma, which is an analog of Lemma \ref{lem: algebra}.
\begin{lem}
\label{lem: algebra nonstationary} Let the assumptions of Theorem \ref{thm: common nonstationarity} hold. 
 Then, for $1\leq k\le K$,
\[
\hat{\tau}_{k}-\tau=\left(\frac{1}{T_{1}}\sum_{t=T_{0}+1}^{T}u_{t}-\frac{1}{|H_{k}|}\sum_{t\in H_{k}}u_{t}\right)+\left(\frac{1}{|H_{k}|}\sum_{t\in H_{k}}\tZ_{t}-\frac{1}{T_{1}}\sum_{t=T_{0}+1}^{T}\tZ_{t}\right)'\Delta_{(k)},
\]
where $\tZ_{t}=Z_{t}-E(Z_{t})$ and $\Delta_{(k)}=\hw_{(k)}-w_\ast$. 
\end{lem}
\begin{proof}[\textbf{Proof of Lemma \ref{lem: algebra nonstationary}}]
Let $\mu=E(Z_t)$. Notice that because $\mathbf{1}_N'w_\ast=\mathbf{1}_N'\hw_{(k)}=1$, for $T_{0}+1\leq t\leq T$, 
\begin{eqnarray*}
Y_{t}-X_{t}'\hw_{(k)}=\alpha_{t}+u_{t}-Z_{t}'\Delta_{(k)}=\alpha_{t}+u_{t}-\mu'\Delta_{(k)}-\tZ_{t}'\Delta_{(k)}
\end{eqnarray*}
and for $t\in H_{k}$
\[
Y_{t}-X_{t}'\hw_{(k)}=u_{t}-Z_{t}'\Delta_{(k)}=u_{t}-\mu'\Delta_{(k)}-\tZ_{t}'\Delta_{(k)}.
\]
Therefore, 
\begin{align*}
\htau_{k}-\tau & =T_{1}^{-1}\sum_{t=T_{0}+1}^{T}(Y_{t}-X_{t}'\hw_{(k)})-\frac{1}{|H_{k}|}\sum_{t\in H_{k}}(Y_{t}-X_{t}'\hat{w}_{(k)})-\tau\\
 & =T_{1}^{-1}\sum_{t=T_{0}+1}^{T}u_{t}-\frac{1}{|H_{k}|}\sum_{t\in H_{k}}u_{t}+\frac{1}{|H_{k}|}\sum_{t\in H_{k}}\tZ_{t}'\Delta_{(k)}-T_{1}^{-1}\sum_{t=T_{0}+1}^{T}\tZ_{t}'\Delta_{(k)}.
\end{align*}
The proof is complete.
\end{proof}
\begin{proof}[\textbf{Proof of Theorem \ref{thm: common nonstationarity}}]
By Equation \eqref{eq:sc_estimator_common_nonstationarity}, $\hw_{(k)}$ only depends on $\{Z_t\}_{t\in H_{(-k)}}$ and, thus, is approximately independent of $\{Z_t\}_{t\in H_{(k)} \bigcup \{T_0+1,\dots,T\} }$. Therefore, using Lemma \ref{lem: algebra nonstationary} and the same arguments as in Theorem \ref{thm: asy distr}, one can show that
	\begin{equation*}
\left|\hat{\tau}_{k}-\tau-\left(\frac{1}{T_{1}}\sum_{t=T_{0}+1}^{T}u_{t}-\frac{1}{|H_{k}|}\sum_{t\in H_{k}}u_{t}\right)\right|=O_{P}\left(\frac{\rho_{T}\gamma_{T} \|\Delta_{(k)}\|_{1}}{\min\{T_0,T_1\}  }+ \frac{\|\Delta_{(k)}\|_{2}}{\sqrt{\min\{T_0,T_1\}  }} \right).
\end{equation*}
The result now follows by Assumption \ref{ass:ell_2_consistency} and the same arguments as in Theorems \ref{thm: asy distr} and \ref{thm: t statistic} and Corollary \ref{cor:CI_PI}.
\end{proof}

\subsection{Proof of Theorem \ref{thm: nonstat deviation}}

We first derive a result on the error bounds for the SC estimator. Then we prove Theorem \ref{thm: nonstat deviation}. 
\begin{lem}
\label{lem: non-stationary machinery lasso}Assume that $Y=Xw+u$
and $X=\theta\oneb_{N}'+Z+\xi\beta'$, where $u\in\RR^{T}$, $w\in\RR^{N}$
satisfies $\oneb_{N}'w=1$ and $w\geq0$. Consider the estimator 
\[
\hw=\arg\min_{v\in\RR^{N}}\ \|Y-Xv\|_{2}^{2}\qquad s.t.\qquad v\geq0,\ \oneb_{N}'v=1.
\]

Suppose that the eigenvalues of $EZ'Z/T$ are between $c_{0}^{-1}$
and $c_{0}$. Then on the event $\Mcal$, 
\[
\|\hw-w\|_{2}\leq\sqrt{\frac{c_{0}}{T}\left[8\lambda_{1}^{2}+2c_{1}^{2}+4\lambda_{2}+4\lambda_{3}\right]}
\]
and 
\[
|\beta'(\hw-w)|\leq\|\xi\|_{2}^{-1}\left[6\lambda_{1}+3c_{1}+4\sqrt{\lambda_{2}}\right],
\]
where 
\[
\Mcal=\left\{ \|Z'\xi\|_{\infty}\leq\lambda_{1}\|\xi\|_{2}\right\} \bigcap\left\{ |u'\xi|\leq c_{1}\|\xi\|_{2}\right\} \bigcap\left\{ \|Z'u\|_{\infty}\leq\lambda_{2}\right\} \bigcap\left\{ \|Z'Z-EZ'Z\|_{\infty}\leq\lambda_{3}\right\} .
\]
\end{lem}
\begin{proof}[\textbf{Proof of Lemma \ref{lem: non-stationary machinery lasso}}]
Let $\Delta=\hw-w$ and $\Sigma_{Z}=EZ'Z/T$. We proceed in two steps.

\textbf{Step 1:} bound $\|\Delta\|_{2}$.

Recall that $Y=Xw+u$. By definition, we have $\|Y-X\hw\|_{2}^{2}\leq\|Y-Xw\|_{2}^{2}$.
This means that 
\[
\|X\Delta\|_{2}^{2}\leq2u'X\Delta.
\]
We notice that $\oneb_{N}'\Delta=0$, $X=\theta\oneb_{N}'+Z+\xi\beta'$
and thus
\[
X\Delta=Z\Delta+\xi\beta'\Delta.
\]

Therefore, we have 
\[
\|Z\Delta\|_{2}^{2}+\|\xi\|_{2}^{2}(\beta'\Delta)^{2}+2\Delta'Z'\xi\beta'\Delta\leq2u'Z\Delta+2u'\xi\beta'\Delta.
\]

We then have 
\begin{equation}
\left[\|\xi\|_{2}(\beta'\Delta)+\|\xi\|_{2}^{-1}(\Delta'Z'\xi-u'\xi)\right]^{2}\leq\|\xi\|_{2}^{-2}(\Delta'Z'\xi-u'\xi)^{2}+2u'Z\Delta-\|Z\Delta\|_{2}^{2}.\label{eq: lem lasso nonstationary 3}
\end{equation}

It follows that
\[
\|\xi\|_{2}^{-2}(\Delta'Z'\xi-u'\xi)^{2}+2u'Z\Delta-\|Z\Delta\|_{2}^{2}\geq0.
\]

Notice that $\|\Delta\|_{1}\leq\|\hw\|_{1}+\|w\|_{1}=2$. Hence, 
\begin{equation}
(\Delta'Z'\xi-u'\xi)^{2}\leq2(\Delta'Z'\xi)^{2}+2(u'\xi)^{2}\leq2\|\Delta\|_{1}^{2}\|Z'\xi\|_{\infty}^{2}+2c_{1}^{2}\|\xi\|_{2}^{2}\leq(8\lambda_{1}^{2}+2c_{1}^{2})\|\xi\|_{2}^{2}.\label{eq: lem lasso nonstationary 3.1}
\end{equation}

By the above two displays, we have 
\[
8\lambda_{1}^{2}+2c_{1}^{2}+2u'Z\Delta\geq\|Z\Delta\|_{2}^{2}.
\]

Moreover, we notice that $|u'Z\Delta|\leq\lambda_{2}\|\Delta\|_{1}$
and 
\begin{align*}
\|Z\Delta\|_{2}^{2} & =\Delta'EZ'Z\Delta+\Delta'(Z'Z-EZ'Z)\Delta\\
 & \geq\Delta'EZ'Z\Delta-\|Z'Z-EZ'Z\|_{\infty}\|\Delta\|_{1}^{2}.
\end{align*}

By $\|\Delta\|_{1}\leq2$, we have that
\[
8\lambda_{1}^{2}+2c_{1}^{2}+4\lambda_{2}\geq T\Delta'\Sigma_{Z}\Delta-4\lambda_{3}.
\]

Therefore, 
\begin{equation}
\Delta'\Sigma_{Z}\Delta\leq\frac{1}{T}\left[8\lambda_{1}^{2}+2c_{1}^{2}+4\lambda_{2}+4\lambda_{3}\right].\label{eq: lem lasso nonstationary 5}
\end{equation}

We obtain the bound on $\|\Delta\|_{2}$ by $\lambda_{\min}(\Sigma_{Z})\geq c_{0}^{-1}$. 

\textbf{Step 2:} bound $\beta'\Delta$. 

By (\ref{eq: lem lasso nonstationary 3}) and (\ref{eq: lem lasso nonstationary 3.1}),
we have 
\begin{multline*}
\left[\|\xi\|_{2}(\beta'\Delta)+\|\xi\|_{2}^{-1}(\Delta'Z'\xi-u'\xi)\right]^{2}\leq\|\xi\|_{2}^{-2}(\Delta'Z'\xi-u'\xi)^{2}+2u'Z\Delta\\
\leq8\lambda_{1}^{2}+2c_{1}^{2}+2u'Z\Delta\leq8\lambda_{1}^{2}+2c_{1}^{2}+2\|Z'u\|_{\infty}\|\Delta\|_{1}\leq8\lambda_{1}^{2}+2c_{1}^{2}+4\lambda_{2}.
\end{multline*}

Then 
\begin{align*}
\|\xi\|_{2}\cdot|\beta'\Delta| & \leq\|\xi\|_{2}^{-1}\cdot|\Delta'Z'\xi-u'\xi|+\sqrt{8\lambda_{1}^{2}+2c_{1}^{2}+4\lambda_{2}}\\
 & \overset{\text{(i)}}{\leq}\sqrt{8\lambda_{1}^{2}+2c_{1}^{2}}+\sqrt{8\lambda_{1}^{2}+2c_{1}^{2}+4\lambda_{2}}\\
 & \leq2\sqrt{8\lambda_{1}^{2}+2c_{1}^{2}+4\lambda_{2}}\leq6\lambda_{1}+3c_{1}+4\sqrt{\lambda_{2}}.
\end{align*}
where (i) follows by (\ref{eq: lem lasso nonstationary 3.1}). The
proof is complete.
\end{proof}
\begin{proof}[\textbf{Proof of Theorem \ref{thm: nonstat deviation}}]
Let $\Delta_{(k)}=\hw_{(k)}-w$. Notice that
\begin{align}
\htau_{k}-\tau & =T_{1}^{-1}\sum_{t=T_{0}+1}^{T}(Y_{t}-X_{t}'\hw_{(k)})-\frac{1}{|H_{k}|}\sum_{t\in H_{k}}(Y_{t}-X_{t}'\hw_{(k)})-\tau\nonumber \\
 & =T_{1}^{-1}\sum_{t=T_{0}+1}^{T}(\alpha_{t}+u_{t}-X_{t}'\Delta_{(k)})-\frac{1}{|H_{k}|}\sum_{t\in H_{k}}(u_{t}-X_{t}'\Delta_{(k)})-\tau\nonumber \\
 & =T_{1}^{-1}\sum_{t=T_{0}+1}^{T}(u_{t}-X_{t}'\Delta_{(k)})-\frac{1}{|H_{k}|}\sum_{t\in H_{k}}(u_{t}-X_{t}'\Delta_{(k)})\nonumber \\
 & =T_{1}^{-1}\sum_{t=T_{0}+1}^{T}u_{t}-\frac{1}{|H_{k}|}\sum_{t\in H_{k}}u_{t}+\frac{1}{|H_{k}|}\sum_{t\in H_{k}}X_{t}'\Delta_{(k)}-T_{1}^{-1}\sum_{t=T_{0}+1}^{T}X_{t}'\Delta_{(k)}.\label{eq: thm nonstat devi 3}
\end{align}

We notice that 
\[
\sum_{t\in H_{k}}X_{t}'\Delta_{(k)}=\sum_{t\in H_{k}}Z_{t}'\Delta_{(k)}+\left(\sum_{t\in H_{k}}\xi_{t}\right)(\beta'\Delta_{(k)}).
\]

Since $\Delta_{(k)}$ depends on $\{(\theta_{t},Z_{t},\xi_{t})\}_{t\in H_{(-k)}}$
and $\{Z_{t}\}_{t\in H_{(-k)}}$ is independent of $\{(\theta_{t},\xi_{t})\}_{t=1}^{T}$,
it follows that conditional on $\{(\theta_{t},\xi_{t})\}_{t=1}^{T}$,
we can still use the coupling argument in the proof of Theorem \ref{thm: asy distr}
and obtain 
\[
\sum_{t\in H_{k}}Z_{t}'\Delta_{(k)}=O_{P}\left(\rho_{T}\gamma_{T}+\sqrt{|H_{k}|}\cdot\|\Delta_{(k)}\|_{2}\right).
\]

By Assumption \ref{assu: nonstat reg}, we apply Lemma \ref{lem: non-stationary machinery lasso}
and obtain 
\begin{equation}
\|\Delta_{(k)}\|_{2}=O_{P}\left((T^{-1}\log N)^{1/4}\right)=o_{P}(1)\label{eq: thm nonstat devi 3.3}
\end{equation}
and
\[
|\beta'\Delta_{(k)}|=O_{P}\left(\frac{(\log N)^{1/2}}{\sqrt{\sum_{t\in H_{(-k)}}\xi_{t}^{2}}}\right).
\]

Therefore, 
\begin{align}
 & \frac{\sqrt{\min\{T_{0},T_{1}\}}}{|H_{k}|}\sum_{t\in H_{k}}X_{t}'\Delta_{(k)}\nonumber \\
 & =\frac{\sqrt{\min\{T_{0},T_{1}\}}}{|H_{k}|}\sum_{t\in H_{k}}Z_{t}'\Delta_{(k)}+\frac{\sqrt{\min\{T_{0},T_{1}\}}}{|H_{k}|}\left(\sum_{t\in H_{k}}\xi_{t}\right)(\beta'\Delta_{(k)})\nonumber \\
 & =\frac{\sqrt{\min\{T_{0},T_{1}\}}}{|H_{k}|}O_{P}\left(\rho_{T}\gamma_{T}+\sqrt{|H_{k}|}\cdot\|\Delta_{(k)}\|_{2}\right)\nonumber \\
 & \qquad+\frac{\sqrt{\min\{T_{0},T_{1}\}}}{|H_{k}|}O_{P}\left(\left(\sum_{t\in H_{k}}\xi_{t}\right)\cdot\frac{(\log N)^{1/2}}{\sqrt{\sum_{t\in H_{(-k)}}\xi_{t}^{2}}}\right)\nonumber \\
 & \overset{\text{(i)}}{=}o_{P}(1)+T_{1}^{-1/2}O_{P}\left((\log N)^{1/2}\frac{\sum_{t\in H_{k}}\xi_{t}}{\sqrt{\sum_{t\in H_{(-k)}}\xi_{t}^{2}}}\right)\nonumber \\
 & \overset{\text{(ii)}}{=}o_{P}(1)+T_{1}^{-1/2}O_{P}\left((\log N)^{1/2}\cdot T_{1}T^{-1/2}\right)\overset{\text{(iii)}}{=}o_{P}(1),\label{eq: thm nonstat devi 5}
\end{align}
where (i) follows by (\ref{eq: thm nonstat devi 3.3}), $\min\{T_{1},T_{0}\}\asymp T_{1}$,
$|H_{k}|\asymp T_{1}$, $T_{0}\asymp T$, $T_{1}^{-1/2}\rho_{T}\gamma_{T}=o(1)$
(Assumption \ref{assu: weak dependence}), (ii) follows by Assumption \ref{assu: non-stationary deviation},
and (iii) follows by $T_{1}\ll T_{0}/\log N$.

Similarly, we can show that

\[
\sum_{t=T_{0}+1}^{T}Z_{t}'\Delta_{(k)}=O_{P}\left(\rho_{T}\gamma_{T}+\sqrt{T_{1}}\|\Delta_{(k)}\|_{2}\right)
\]
and thus 
\begin{align}
 & \frac{\sqrt{\min\{T_{0},T_{1}\}}}{T_{1}}\sum_{t=T_{0}+1}^{T}X_{t}'\Delta_{(k)}\nonumber \\
 & =\frac{\sqrt{\min\{T_{0},T_{1}\}}}{T_{1}}\sum_{t=T_{0}+1}^{T}Z_{t}'\Delta_{(k)}+\frac{\sqrt{\min\{T_{0},T_{1}\}}}{T_{1}}\left(\sum_{t=T_{0}+1}^{T}\xi_{t}\right)(\beta'\Delta_{(k)})\nonumber \\
 & =\frac{\sqrt{\min\{T_{0},T_{1}\}}}{T_{1}}O_{P}\left(\rho_{T}\gamma_{T}+\sqrt{T_{1}}\|\Delta_{(k)}\|_{2}\right)\nonumber \\
 & \qquad+\frac{\sqrt{\min\{T_{0},T_{1}\}}}{T_{1}}O_{P}\left(\left(\sum_{t=T_{0}+1}^{T}\xi_{t}\right)\frac{(\log N)^{1/2}}{\sqrt{\sum_{t\in H_{(-k)}}\xi_{t}^{2}}}\right)\nonumber \\
 & \overset{\text{(i)}}{=}o_{P}(1)+T_{1}^{-1/2}O_{P}\left((\log N)^{1/2}\frac{\sum_{t=T_{0}+1}^{T}\xi_{t}}{\sqrt{\sum_{t\in H_{(-k)}}\xi_{t}^{2}}}\right)\nonumber \\
 & \overset{\text{(ii)}}{=}o_{P}(1)+T_{1}^{-1/2}O_{P}\left((\log N)^{1/2}\cdot T_{1}T^{-1/2}\right)\overset{\text{(iii)}}{=}o_{P}(1),\label{eq: thm nonstat devi 6}
\end{align}
where (i) follows by (\ref{eq: thm nonstat devi 3.3}), $\min\{T_{1},T_{0}\}\asymp T_{1}$
and $T_{1}^{-1/2}\rho_{T}\gamma_{T}=o(1)$, (ii) follows by Assumption
\ref{assu: non-stationary deviation} and (iii) follows by $T_{1}\ll T_{0}/\log N$.

Now we combine (\ref{eq: thm nonstat devi 3}) with (\ref{eq: thm nonstat devi 5})
and (\ref{eq: thm nonstat devi 6}), obtaining 
\begin{align*}
 & \sqrt{\min\{T_{0},T_{1}\}}(\htau_{k}-\tau)\\
 & =\sqrt{\frac{\min\{T_{0},T_{1}\}}{T_{1}}}\left(T_{1}^{-1/2}\sum_{t=T_{0}+1}^{T}u_{t}\right)-\sqrt{\frac{\min\{T_{0},T_{1}\}}{|H_{k}|}}\left(|H_{k}|^{-1/2}\sum_{t\in H_{k}}u_{t}\right)\\
 & \qquad+\frac{\sqrt{\min\{T_{0},T_{1}\}}}{|H_{k}|}\sum_{t\in H_{k}}X_{t}'\Delta_{(k)}-\frac{\sqrt{\min\{T_{0},T_{1}\}}}{T_{1}}\sum_{t=T_{0}+1}^{T}X_{t}'\Delta_{(k)}\\
 & =\sqrt{\frac{\min\{T_{0},T_{1}\}}{T_{1}}}\left(T_{1}^{-1/2}\sum_{t=T_{0}+1}^{T}u_{t}\right)-\sqrt{\frac{\min\{T_{0},T_{1}\}}{|H_{k}|}}\left(|H_{k}|^{-1/2}\sum_{t\in H_{k}}u_{t}\right)+o_{P}(1).
\end{align*}

Since $T_{1}\ll T_{0}$ and $|H_{k}|=\min\{T_{0}/K,T_{1}\}$, we have
$\min\{T_{0},T_{1}\}/T_{1}\rightarrow1$ and $\min\{T_{0},T_{1}\}/|H_{k}|\rightarrow1$.
The first claim of the theorem follows.

For the second claim, we notice that $(T_{1}^{-1/2}\sum_{t=T_{0}+1}^{T}u_{t},|H_{1}|^{-1/2}\sum_{t\in H_{1}}u_{t},...,|H_{K}|^{-1/2}\sum_{t\in H_{k}}u_{t})'$
converges in distribution to the normal distribution with mean zero
and variance matrix $I_{K+1}\sigma^{2}$, where $\sigma^{2}$ is the
long-run variance of $u_{t}$. Notice that 
\[
\mathbb{T}_{K}=\frac{\sqrt{\frac{KT_{1}}{|H_{1}|K+T_{1}}}\cdot(\htau-\tau)}{\sqrt{\frac{1}{K-1}\sum_{k=1}^{K}(\htau_{(k)}-\htau)^{2}}}.
\]

By the same argument as in the proof of Theorem \ref{thm: t statistic}, asymptotically
the numerator has $N(0,1)$ distribution and is independent of the
denominator, which asymptotically has a $\chi_{K-1}^{2}$ distribution
divided by $K-1$. Therefore, $\mathbb{T}_{K}$ converges in distribution
to the student $t$-distribution with $K-1$ degrees of freedom and, by the argument in Corollary \ref{cor:CI_PI}, $P(\tau \in \mathcal{I}_K(1-\alpha))\rightarrow 1-\alpha$. The
proof is complete.
\end{proof}

\subsection{Proof of Theorem \ref{thm: kellog}}

Before we prove Theorem \ref{thm: kellog}, we introduce the some
notations and establish a preliminary result. For any $i\in\{0,...,N\}$
and any $w\in\mathcal{W}^{SC}$, we have that 
\[
Y_{it}(0)-Y_{-i,t}(0)'w=\theta_{t}+V_{it}-(\theta_{t}\oneb_{N}'+V_{-i,t}')w=V_{it}-V_{-i,t}'w.
\]

Due to the covariance-stationarity of $V_{t}$, the distribution of
the above quantity does not depend on $t$. Hence, we can define the
following time-independent quantity: 
\[
w_{*}^{(i)}\in\underset{w\in\mathcal{W}^{SC}}{\arg\min}\,E(Y_{it}(0)-Y_{-i,t}'(0)w)^{2}=\underset{w\in\mathcal{W}^{SC}}{\arg\min}\,E(V_{it}-V_{-i,t}'w)^{2}.
\]
To state the next result, let $\Delta_{(k)}^{(i)}=\hat{w}_{(k)}^{(i)}- w_{*}^{(i)}$.
\begin{lem}
\label{lem: Kellogg tool 1} Let Assumptions \ref{assu: kellogg part 1}
and \ref{assu: kellogg part 2} hold. Then $\max_{1\leq k\leq K}\max_{0\leq i\leq N}\|\Delta_{(k)}^{(i)}\|_{2}=o_{P}(1)$.
\end{lem}

\begin{proof}[\textbf{Proof of Lemma \ref{lem: Kellogg tool 1}}]
Let $f_{i}(w)$ be the $(N+1)$-dimensional vector whose $i$-th entry
is $1$ and whose other entries are $-w$. Then 
\[
\underset{w\in\mathcal{W}^{SC}}{\arg\min}\,E(V_{it}-V_{-i,t}'w)^{2}=\underset{w\in\mathcal{W}^{SC}}{\arg\min}\,E(V_{t}'f_{i}(w))^{2}=\underset{w\in\mathcal{W}^{SC}}{\arg\min}\,f_{i}(w)'\Sigma f_{i}(w),
\]
where $\Sigma=EV_{t}V_{t}'$. Similarly, 
\begin{multline*}
\underset{w\in\mathcal{W}^{SC}}{\arg\min}\,\sum_{t\in H_{(-k)}}(Y_{it}-Y_{-i,t}'w)^{2}=\underset{w\in\mathcal{W}^{SC}}{\arg\min}\,\sum_{t\in H_{(-k)}}(V_{it}-V_{-i,t}'w)^{2}\\
=\underset{w\in\mathcal{W}^{SC}}{\arg\min}\,\sum_{t\in H_{(-k)}}(V_{t}'f_{i}(w))^{2}=\underset{w\in\mathcal{W}^{SC}}{\arg\min}\,f_{i}(w)'\hat{\Sigma}_{(-k)}f_{i}(w),
\end{multline*}
where $\hat{\Sigma}_{(-k)}=|H_{(-k)}|^{-1}\sum_{t\in H_{(-k)}}V_{t}V_{t}'$.
Let $\Delta_{(k)}^{(i)}=f_{i}(\hat{w}_{(k)}^{(i)})-f_{i}(w_{*}^{(i)})$.
By construction, 
\[
[f_{i}(w_{*}^{(i)})+\Delta_{(k)}^{(i)}]'\hat{\Sigma}_{(-k)}[f_{i}(w_{*}^{(i)})+\Delta_{(k)}^{(i)}]\leq f_{i}(w_{*}^{(i)})'\hat{\Sigma}_{(-k)}f_{i}(w_{*}^{(i)}).
\]

Thus, 
\[
(\Delta_{(k)}^{(i)})'\hat{\Sigma}_{(-k)}\Delta_{(k)}^{(i)}\leq-2f_{i}(w_{*}^{(i)})'\hat{\Sigma}_{(-k)}\Delta_{(k)}^{(i)}.
\]

Similar to the argument of  (\ref{eq: key bnd misspecifiation}) in the proof of Lemma \ref{lem: consistency misspecification},
we have that 
\[
f_{i}(w_{*}^{(i)})'\Sigma\Delta_{(k)}^{(i)}\geq0.
\]

The above two displays imply 
\begin{multline*}
(\Delta_{(k)}^{(i)})'\Sigma\Delta_{(k)}^{(i)}-\|\hat{\Sigma}_{(-k)}-\Sigma\|_{\infty}\cdot\|\Delta_{(k)}^{(i)}\|_{1}^{2}\leq(\Delta_{(k)}^{(i)})'\hat{\Sigma}_{(-k)}\Delta_{(k)}^{(i)}\\
\leq-2f_{i}(w_{*}^{(i)})'\hat{\Sigma}_{(-k)}\Delta_{(k)}^{(i)}\leq2\|f_{i}(w_{*}^{(i)})\|_{1}\|\hat{\Sigma}_{(-k)}-\Sigma\|_{\infty}\|\Delta_{(k)}^{(i)}\|_{1}.
\end{multline*}

By construction, $\|f_{i}(w_{*}^{(i)})\|_{1}=2$ and $\|\Delta_{(k)}^{(i)}\|_{1}=\|\hat{w}_{(k)}^{(i)}-w_{*}^{(i)}\|_{1}\leq2$.
It follows that $(\Delta_{(k)}^{(i)})'\Sigma\Delta_{(k)}^{(i)}\leq12\|\hat{\Sigma}_{(-k)}-\Sigma\|_{\infty}$.
Since the eigenvalues of $\Sigma$ are bounded below by $1/\kappa_{1}$,
this means that 
\[
\max_{0\leq i\leq N,0\leq k\leq K}\|\Delta_{(k)}^{(i)}\|_{2}^{2}\leq12\kappa_{1}\max_{0\leq k\leq K}\|\hat{\Sigma}_{(-k)}-\Sigma\|_{\infty}.
\]

The result follows by the assumption of $\max_{0\leq k\leq K}\|\hat{\Sigma}_{(-k)}-\Sigma\|_{\infty}=o_{P}(1)$. 
\end{proof}
\begin{proof}[\textbf{Proof of Theorem \ref{thm: kellog}}]
Define $\tilde{V}_{it}=V_{it}-EV_{it}$ and $u_{it}=\tilde{V}_{it}-\tilde{V}_{-i,t}'w_{*}^{(i)}$.
Since $\hat{w}_{(k)}^{(i)}\in\mathcal{W}^{SC}$, we have that $Y_{it}(0)-Y_{-i,t}(0)'\hat{w}_{(k)}^{(i)}=V_{it}-V_{-i,t}'\hat{w}_{(k)}^{(i)}$.
Hence, 
\begin{align*}
 & T_{1}^{-1}\sum_{i=0}^{N}\sum_{t=T_{0}+1}^{T}(Y_{it}-Y_{-i,t}'\hat{w}_{(k)}^{(i)})\mathbf{1}\{D=i\}\\
 & =T_{1}^{-1}\sum_{i=0}^{N}\sum_{t=T_{0}+1}^{T}(\alpha_{it}+Y_{it}(0)-Y_{-i,t}(0)'\hat{w}_{(k)}^{(i)})\mathbf{1}\{D=i\}\\
 & =\tau+T_{1}^{-1}\sum_{i=0}^{N}\sum_{t=T_{0}+1}^{T}(Y_{it}(0)-Y_{-i,t}(0)'\hat{w}_{(k)}^{(i)})\mathbf{1}\{D=i\}\\
 & =\tau+T_{1}^{-1}\sum_{i=0}^{N}\sum_{t=T_{0}+1}^{T}(V_{it}-V_{-i,t}'\hat{w}_{(k)}^{(i)})\mathbf{1}\{D=i\}
\end{align*}
and similarly, 
\[
\frac{1}{|H_{k}|}\sum_{i=0}^{N}\sum_{t\in H_{k}}(Y_{it}-Y_{-i,t}'\hat{w}_{(k)}^{(i)})\mathbf{1}\{D=i\}=\frac{1}{|H_{k}|}\sum_{i=0}^{N}\sum_{t\in H_{k}}(V_{it}-V_{-i,t}'\hat{w}_{(k)}^{(i)})\mathbf{1}\{D=i\}.
\]

Therefore, letting $\Delta_{(k)}^{(i)}=\hat{w}_{(k)}^{(i)}-w_{*}^{(i)}$,
we have
\begin{align}
\hat{\tau}_{k}-\tau & =\sum_{i=0}^{N}\left(T_{1}^{-1}\sum_{t=T_{0}+1}^{T}(V_{it}-V_{-i,t}'\hat{w}_{(k)}^{(i)})-\frac{1}{|H_{k}|}\sum_{t\in H_{k}}(V_{it}-V_{-i,t}'\hat{w}_{(k)}^{(i)})\right)\cdot\oneb\{D=i\}\nonumber \\
 & =\sum_{i=0}^{N}\left(T_{1}^{-1}\sum_{t=T_{0}+1}^{T}(\tilde{V}_{it}-\tilde{V}_{-i,t}'\hat{w}_{(k)}^{(i)})-\frac{1}{|H_{k}|}\sum_{t\in H_{k}}(\tilde{V}_{it}-\tilde{V}_{-i,t}'\hat{w}_{(k)}^{(i)})\right)\cdot\oneb\{D=i\}\nonumber \\
 & \quad+\sum_{i=0}^{N}\left(T_{1}^{-1}\sum_{t=T_{0}+1}^{T}(EV_{it}-(EV_{-i,t})'\hat{w}_{(k)}^{(i)})-\frac{1}{|H_{k}|}\sum_{t\in H_{k}}(EV_{it}-(EV_{-i,t})'\hat{w}_{(k)}^{(i)})\right)\cdot\oneb\{D=i\}\nonumber \\
 & \overset{\text{(i)}}{=}\sum_{i=0}^{N}\left(T_{1}^{-1}\sum_{t=T_{0}+1}^{T}(\tilde{V}_{it}-\tilde{V}_{-i,t}'\hat{w}_{(k)}^{(i)})-\frac{1}{|H_{k}|}\sum_{t\in H_{k}}(\tilde{V}_{it}-\tilde{V}_{-i,t}'\hat{w}_{(k)}^{(i)})\right)\cdot\oneb\{D=i\}\nonumber \\
 & =\sum_{i=0}^{N}\left(T_{1}^{-1}\sum_{t=T_{0}+1}^{T}u_{it}-\frac{1}{|H_{k}|}\sum_{t\in H_{k}}u_{it}\right)\cdot\oneb\{D=i\}\nonumber \\
 & \quad+\sum_{i=0}^{N}\left(\frac{1}{|H_{k}|}\sum_{t\in H_{k}}\tilde{V}_{-i,t}-T_{1}^{-1}\sum_{t=T_{0}+1}^{T}\tilde{V}_{-i,t}\right)'\Delta_{(k)}^{(i)}\cdot\oneb\{D=i\},\label{eq: thm kellogg eq 4}
\end{align}
where (i) follows by the covariance-stationarity of $V_{t}$. 

For any $i\in\{0,...,N\}$, let $\bar{\Delta}_{(k)}^{(i)}$ be the
$(N+1)$-dimensional vector whose $i$-th component is zero and other
components correspond to $\Delta_{(k)}^{(i)}$. Then $\tilde{V}_{-i,t}'\Delta_{(k)}^{(i)}=\tilde{V}_{t}'\bar{\Delta}_{(k)}^{(i)}$,
$\|\bar{\Delta}_{(k)}^{(i)}\|_{2}^{2}=\|\Delta_{(k)}^{(i)}\|_{2}^{2}$
and $\|\bar{\Delta}_{(k)}^{(i)}\|_{1}=\|\Delta_{(k)}^{(i)}\|_{1}\leq2$.
Hence, 
\begin{align*}
 & \sum_{i=0}^{N}\left(\frac{1}{|H_{k}|}\sum_{t\in H_{k}}\tilde{V}_{-i,t}-T_{1}^{-1}\sum_{t=T_{0}+1}^{T}\tilde{V}_{-i,t}\right)'\Delta_{(k)}^{(i)}\cdot\oneb\{D=i\}\\
 & =\frac{1}{|H_{k}|}\sum_{t\in H_{k}}\sum_{i=0}^{N}\tilde{V}_{-i,t}'\Delta_{(k)}^{(i)}\cdot\oneb\{D=i\}-T_{1}^{-1}\sum_{t=T_{0}+1}^{T}\sum_{i=0}^{N}\tilde{V}_{-i,t}'\Delta_{(k)}^{(i)}\cdot\oneb\{D=i\}\\
 & =\frac{1}{|H_{k}|}\sum_{t\in H_{k}}\sum_{i=0}^{N}\tilde{V}_{t}'\bar{\Delta}_{(k)}^{(i)}\cdot\oneb\{D=i\}-T_{1}^{-1}\sum_{t=T_{0}+1}^{T}\sum_{i=0}^{N}\tilde{V}_{t}'\bar{\Delta}_{(k)}^{(i)}\cdot\oneb\{D=i\}\\
 & =\frac{1}{|H_{k}|}\sum_{t\in H_{k}}\tilde{V}_{t}'\Delta_{(k)}-T_{1}^{-1}\sum_{t=T_{0}+1}^{T}\tilde{V}_{t}'\Delta_{(k)},
\end{align*}
where $\Delta_{(k)}=\sum_{i=0}^{N}\bar{\Delta}_{(k)}^{(i)}\cdot\oneb\{D=i\}$.
By the same coupling argument as in the proof of Theorem \ref{thm: asy distr}, we have
that 
\[
\left|\frac{1}{|H_{k}|}\sum_{t\in H_{k}}\tilde{V}_{t}'\Delta_{(k)}-T_{1}^{-1}\sum_{t=T_{0}+1}^{T}\tilde{V}_{t}'\Delta_{(k)}\right|=O_{P}\left(\frac{\rho_{T}\gamma_{T}\|\Delta_{(k)}\|_{1}}{\min\{T_{0},T_{1}\}}+\frac{\|\Delta_{(k)}\|_{2}}{\sqrt{\min\{T_{0},T_{1}\}}}\right),
\]
where $\rho_{T}\gamma_{T}\rightarrow0$. Notice that $\|\Delta_{(k)}\|_{1}=\sum_{i=0}^{N}\|\bar{\Delta}_{(k)}^{(i)}\|_{1}\cdot\oneb\{D=i\}\leq2$
and $\|\Delta_{(k)}\|_{2}=\sum_{i=0}^{N}\|\bar{\Delta}_{(k)}^{(i)}\|_{2}\cdot\oneb\{D=i\}\leq\max_{0\leq i\leq N}\|\bar{\Delta}_{(k)}^{(i)}\|_{2}=\max_{0\leq i\leq N}\|\Delta_{(k)}^{(i)}\|_{2}=o_{P}(1)$
due to Lemma \ref{lem: Kellogg tool 1}. Hence, 
\[
\left|\frac{1}{|H_{k}|}\sum_{t\in H_{k}}\tilde{V}_{t}'\Delta_{(k)}-T_{1}^{-1}\sum_{t=T_{0}+1}^{T}\tilde{V}_{t}'\Delta_{(k)}\right|=o_{P}\left(\frac{1}{\sqrt{\min\{T_{0},T_{1}\}}}\right).
\]

Let $\bar{w}_{*}^{(i)}$ be the $(N+1)$-dimensional vector whose
$i$-th entry is $1$ and the other entries correspond to $-w_{*}^{(i)}$.
Then $u_{it}=\tilde{V}_{it}-\tilde{V}_{-i,t}'w_{*}^{(i)}=\tilde{V}_{t}'\bar{w}_{*}^{(i)}$.
Therefore, 
\begin{align*}
 & \sum_{i=0}^{N}\left(T_{1}^{-1}\sum_{t=T_{0}+1}^{T}u_{it}-\frac{1}{|H_{k}|}\sum_{t\in H_{k}}u_{it}\right)\cdot\oneb\{D=i\}\\
 & =\sum_{i=0}^{N}\left(T_{1}^{-1}\sum_{t=T_{0}+1}^{T}\tilde{V}_{t}'\bar{w}_{*}^{(i)}-\frac{1}{|H_{k}|}\sum_{t\in H_{k}}\tilde{V}_{t}'\bar{w}_{*}^{(i)}\right)\cdot\oneb\{D=i\}\\
 & =T_{1}^{-1}\sum_{t=T_{0}+1}^{T}\tilde{V}_{t}'\bar{w}_{*}-\frac{1}{|H_{k}|}\sum_{t\in H_{k}}\tilde{V}_{t}'\bar{w}_{*},
\end{align*}
where $\bar{w}_{*}=\sum_{i=0}^{N}\bar{w}_{*}^{(i)}\oneb\{D=i\}$.
The above two displays and (\ref{eq: thm kellogg eq 4}) imply that
\[
\sqrt{\min\{T_{0},T_{1}\}}\left|\hat{\tau}_{(k)}-\tau-\left(T_{1}^{-1}\sum_{t=T_{0}+1}^{T}\tilde{V}_{t}'\bar{w}_{*}-\frac{1}{|H_{k}|}\sum_{t\in H_{k}}\tilde{V}_{t}'\bar{w}_{*}\right)\right|=o_{P}(1).
\]

Let $A_{1}=\{1,...,T\}\backslash\{T_{0}-2B_{T}+1,...,T_{0}+2B_{T}\}$
and $A_{2}=\{T_{0}-B_{T},...,T_{0}+B_{T}\}$. By Berbee's coupling
(e.g., Theorem 16.2.1 of \citet{athreya2006measure}), on an enlarged probability
space, there exist random variables $\{\bar{V}_{t}\}_{t\in A_{1}}$
such that (1) $\{\bar{V}_{t}\}_{t\in A_{1}}$ and $\{\tilde{V}_{t}\}_{t\in A_{1}}$
have the same distribution, (2) $\{\bar{V}_{t}\}_{t\in A_{1}}$ is
independent of $\{\tilde{V}_{t}\}_{t\in A_{2}}$ and (3) $P(\{\bar{V}_{t}\}_{t\in A_{1}}\neq\{\tilde{V}_{t}\}_{t\in A_{1}})\leq\beta_{{\rm mix}}(B_{T})$.
Define $R_{T,k}=T_{1}^{-1}\sum_{t=T_{0}+2B_{T}+1}^{T}\tilde{V}_{t}'\bar{w}_{*}-\frac{1}{|H_{k}|}\sum_{t\in H_{k}\backslash A_{1}}\tilde{V}_{t}'\bar{w}_{*}$.
Notice that 
\begin{align*}
 & T_{1}^{-1}\sum_{t=T_{0}+1}^{T}\tilde{V}_{t}'\bar{w}_{*}-\frac{1}{|H_{k}|}\sum_{t\in H_{k}}\tilde{V}_{t}'\bar{w}_{*}-R_{T,k}\\
 & =T_{1}^{-1}\sum_{t=T_{0}+1}^{T_{0}+2B_{T}}\tilde{V}_{t}'\bar{w}_{*}-\frac{1}{|H_{k}|}\sum_{t\in H_{k}\bigcap A_{1}}\tilde{V}_{t}'\bar{w}_{*}\\
 & =O_{P}\left(\max_{1\leq t\leq T}\|\tilde{V}_{t}\|_{\infty}\left(T_{1}^{-1}\cdot2B_{T}+|H_{k}|^{-1}\cdot|H_{k}\bigcap A_{1}|\right)\cdot\|\bar{w}_{*}\|_{1}\right)\\
 & \overset{\text{(i)}}{=}O_{P}\left(\rho_{T}\left(T_{1}^{-1}\cdot2B_{T}+|H_{k}|^{-1}\cdot|H_{k}\bigcap A_{1}|\right)\right)\overset{\text{(ii)}}{=}o_{P}\left(\frac{1}{\sqrt{\min\{T_{0},T_{1}\}}}\right),
\end{align*}
where (i) follows by $\|\bar{w}_{*}\|_{1}=\sum_{i=0}^{N}\|\bar{w}_{*}^{(i)}\|_{1}\cdot\oneb\{D=i\}\leq1$
and (ii) follows by $|H_{k}\bigcap A_{1}|\leq2B_{T}$ and the assumption
of $B_{T}^{2}\ll\rho_{T}^{-2}\min\{T_{0},T_{1}\}$. 

By assumption, $\bar{w}_{*}$ is independent of $\{\tilde{V}_{t}\}_{t\in A_{1}}$
conditional on $\{\tilde{V}_{t}\}_{t\in A_{2}}\bigcup\{\theta_{t}\}_{t\in A_{2}}$.
Thus, $\bar{w}_{*}$ is independent of $\{\bar{V}_{t}\}_{t\in A_{1}}$
conditional on $\{\tilde{V}_{t}\}_{t\in A_{2}}\bigcup\{\theta_{t}\}_{t\in A_{2}}$.
Thus, conditional on $\{\bar{V}_{t}\}_{t\in A_{2}}\bigcup\{\theta_{t}\}_{t\in A_{2}}$,
there exists $\sigma^{2}$ representing the long-run variance of $T^{-1}\sum_{t=1}^{T}\bar{V}_{t}'\bar{w}_{*}$
such that
\[
\sqrt{\min\{T_{0},T_{1}\}}\begin{pmatrix}T_{1}^{-1}\sum_{t=T_{0}+2B_{T}+1}^{T}\bar{V}_{t}'\bar{w}_{*}\\
\frac{1}{|H_{1}|}\sum_{t\in H_{1}\backslash A_{1}}\bar{V}_{t}'\bar{w}_{*}\\
\frac{1}{|H_{2}|}\sum_{t\in H_{2}\backslash A_{1}}\bar{V}_{t}'\bar{w}_{*}\\
\vdots\\
\frac{1}{|H_{K}|}\sum_{t\in H_{K}\backslash A_{1}}\bar{V}_{t}'\bar{w}_{*}
\end{pmatrix}/\sigma\overset{d}{\rightarrow}\begin{pmatrix}\sqrt{\min\{c_{0},1\}}\xi_{0}\\
\sqrt{g_{c_{0},K}}\xi_{1}\\
\sqrt{g_{c_{0},K}}\xi_{2}\\
\vdots\\
\sqrt{g_{c_{0},K}}\xi_{K}
\end{pmatrix},
\]
where $\xi_{0},...,\xi_{K}$ are independent $N(0,1)$ random variables.
Since $R_{T,k}=T_{1}^{-1}\sum_{t=T_{0}+2B_{T}+1}^{T}\tilde{V}_{t}'\bar{w}_{*}-\frac{1}{|H_{k}|}\sum_{t\in H_{k}\backslash A_{1}}\tilde{V}_{t}'\bar{w}_{*}$
with probability at least $1-\beta_{\min}(B_{T})$ and $B_{T}\rightarrow\infty$,
the above three displays imply 
\[
\sqrt{\min\{T_{0},T_{1}\}}\begin{pmatrix}\hat{\tau}_{(1)}-\tau\\
\vdots\\
\hat{\tau}_{(K)}-\tau
\end{pmatrix}/\sigma\overset{d}{\rightarrow}\begin{pmatrix}\sqrt{\min\{c_{0},1\}}\xi_{0}-\sqrt{g_{c_{0},K}}\xi_{1}\\
\vdots\\
\sqrt{\min\{c_{0},1\}}\xi_{0}-\sqrt{g_{c_{0},K}}\xi_{K}
\end{pmatrix}.
\]

The desired result follows by the same argument as in Theorem \ref{thm: t statistic}
and Corollary \ref{cor:CI_PI}. 
\end{proof}

\subsection{Proof of Theorem \ref{thm: asy distribution expected effect}}

Using similar arguments as in Lemma \ref{lem: algebra} and Theorem \ref{thm: asy distr}, we obtain
\begin{equation*}
\left|\hat{\tau}_{k}-\tau_{e}-\left(\frac{1}{|\tilde{H}_k|}\sum_{t\in \tilde{H}_k}(u_{t}+\tilde{\alpha}_{t})-\frac{1}{|H_{k}|}\sum_{t\in H_{k}}u_{t}\right)\right|=O_{P}\left(\frac{\rho_{T}\gamma_{T} \|\Delta_{(k)}\|_{1}}{\min\{T_0,T_1\}  }+ \frac{\|\Delta_{(k)}\|_{2}}{\sqrt{\min\{T_0,T_1\}  }} \right).
\end{equation*}

Since $\tilde{\tau}_1,\dots,\tilde{\tau}_K$ are obtained based on non-overlapping blocks of data,
\[
\sqrt{\min\{T_0,T_1\} }\begin{pmatrix}\htau_{1}-\tau_e\\
\vdots\\
\htau_{K}-\tau_e
\end{pmatrix}\overset{d}{\rightarrow}\begin{pmatrix}\xi_{1}\\\vdots\\ \xi_{K}
\end{pmatrix},
\]
where $\xi_1,\dots,\xi_K$ are independent and identically distributed mean zero normal random variables. The result (i) now follows from classical arguments, and (ii) follows from the same arguments as in Corollary \ref{cor:CI_PI}.

\subsection{Proof of Theorem \ref{thm: higher-order t}}
We define the mapping $\psi:\RR^{K}\rightarrow\RR$ by 
\[
\psi(v_{1},...,v_{K})=\frac{\sqrt{K}\bar{v}}{\sqrt{(K-1)^{-1}\sum_{k=1}^{K}(v_{k}-\bar{v})^{2}}}\qquad\text{with}\qquad\bar{v}=K^{-1}\sum_{j=1}^{K}v_{j}.
\]

We define $\Sigma_{T,K}$ to be the covariance matrix of $\sqrt{G}(\hbeta_{1},...,\hbeta_{K})$.
We also define the matrix $\Sigma_{K}$ to be the diagonal matrix
whose diagonal entries are all equal to $TE(\hbeta_{1}^{2})$. Finally,
we define $\xi\sim N(0,I_{K})$. Hence, $\mathbf{T}_K$ has the same distribution
as $\psi(\Sigma_{T,K}^{1/2}\xi)$. Moreover, $\psi(\Sigma_{K}^{1/2}\xi)$
follows student's $t$-distribution with $K-1$ degrees of freedom.
It suffices to show that 
\[
\left|P\left(\psi(\Sigma_{T,K}^{1/2}\xi)\leq a\right)-P\left(\psi(\Sigma_{K}^{1/2}\xi)\leq a\right)\right|=O(T^{-1})\qquad\forall a\in\RR.
\]

Clearly, $\psi(\cdot)$ is Lipschitz. Since the student's $t$-distribution
with $K-1$ degrees of freedom has bounded density, we only need to
show that $\|\Sigma_{T,K}^{1/2}\xi-\Sigma_{K}^{1/2}\xi\|_{2}=O_{P}(T^{-1})$. 

Now we do so by showing that $\Sigma_{T,K}-\Sigma_{K}=O(T^{-1})$.
The diagonal entries in both matrices are the same. We consider the
off-diagonal entries. Fix $k,l\in\{1,...,K\}$ with $l>k\geq1$. Then
\begin{align*}
& \left|\left(\Sigma_{T,K}\right)_{k,l}\right|\\
& =G\left|E(\hbeta_{k}\hbeta_{l})\right|\\
& =G^{-1}\left|\sum_{t_{1}\in H_{k}}\sum_{t_{2}\in H_{l}}Eu_{t_{1}}u_{t_{2}}\right|\\
& =G^{-1}\left|\sum_{t_{1}=(k-1)G+1}^{kG}\sum_{t_{2}=(l-1)G+1}^{lG}Eu_{t_{1}}u_{t_{2}}\right|\\
& =G^{-1}\left|\sum_{t_{1}=(k-1)G+1}^{kG}\sum_{t_{2}=(l-1)G+1}^{lG}\gamma(t_{2}-t_{1})\right|\\
& =G^{-1}\left|\sum_{t_{1}=(k-1)G+1}^{kG}\sum_{h=(l-1)G+1-t_{1}}^{lG-t_{1}}\gamma(h)\right|\\
& =G^{-1}\left|\sum_{h=(l-2)G+1}^{(l-1)G}\gamma(h)\left[\min\left\{ (l-k+1)G-h,G\right\} -\max\left\{ (l-k)G+1-h,1\right\} \right]\right|\\
& \leq G^{-1}\sum_{h=(l-2)G+1}^{(l-1)G}|\gamma(h)|\cdot\left|\left[\min\left\{ (l-k+1)G-h,G\right\} -\max\left\{ (l-k)G+1-h,1\right\} \right]\right|\\
& \leq G^{-1}\sum_{h=(l-2)G+1}^{(l-1)G}|\gamma(h)|h\leq KT^{-1}\sum_{h=1}^{\infty}|\gamma(h)|h.
\end{align*}

Since $\sum_{h=1}^{\infty}|\gamma(h)|h$ is bounded, we have $\left(\Sigma_{T,K}\right)_{k,l}=O(T^{-1})$.
Therefore, $\Sigma_{T,K}-\Sigma_{K}=O(T^{-1})$. The proof is complete.  \qed

\subsection{Proof of Corollary \ref{cor:CI_PI}}

By Theorem \ref{thm: t statistic}, we have that
$$
P\left(\mathbb{T}_K\in [-t_{K-1}(1-\alpha/2),t_{K-1}(1-\alpha/2)]\right)\rightarrow 1-\alpha.$$ 
Note that the event
$$
\left\{\mathbb{T}_K\in [-t_{K-1}(1-\alpha/2),t_{K-1}(1-\alpha/2)]\right\}
$$
is equivalent to the event
$$
\left\{\tau \in \left[\htau -t_{K-1}(1-\alpha/2)\frac{\hat\sigma_{\htau}}{\sqrt{K}},~ \htau +t_{K-1}(1-\alpha/2) \frac{\hat\sigma_{\htau}}{\sqrt{K}}\right]\right\}.
$$
Thus the result follows. \qed

\section{Additional simulations}
\label{app:additional_simulations}

\subsection{Performance with large $T_0$}
\label{app:large_T0}
In this section, we explore the performance of the $t$-test when $T_0=150$. These simulations are motivated by the theoretical results in Section \ref{sec:sparse_deviations}, which require that $T_0$ is much larger than $T_1$. A larger $T_0$ further allows us to explore the impact of choosing a larger $K$.  

Table \ref{tab:empirical_mc_large_T0} shows the results for $(T_0,T_1,N)=(150,16,14)$ and $K\in \{4,6,8\}$. The results for DGP6--DGP7 confirm the  theoretical results in Section \ref{sec:sparse_deviations}, and the results for DGP8--DGP9 suggest that the $t$-test remains quite robust in settings not covered by our theory, provided that $T_0$ is large enough. Table \ref{tab:empirical_mc_large_T0} also shows that increasing $K$ to $K=6$ and $K=8$ reduces the length of the confidence intervals, especially under misspecification, while not affecting coverage accuracy much.

\begin{table}[H]
\setlength{\tabcolsep}{4pt}
\linespread{1.05}
\scriptsize
\centering
\caption{Simulation results with $(T_0,T_1,N)=(150,16,14)$}
\begin{tabular}{lccccccccc}

\toprule
\midrule

&\multicolumn{3}{c}{Bias$\times$10} & \multicolumn{3}{c}{Coverage} & \multicolumn{3}{c}{Avg.\ length CI} \\
\cmidrule(l{5pt}r{5pt}){2-4}\cmidrule(l{5pt}r{5pt}){5-7} \cmidrule(l{5pt}r{5pt}){8-10}
&$K=4$&$K=6$&$K=8$&$K=4$&$K=6$&$K=8$&$K=4$&$K=6$&$K=8$\\
\midrule
DGP1 &  -0.00 & -0.00 & -0.00 & 0.90 & 0.90 & 0.89 & 0.06 & 0.05 & 0.05 \\

\midrule
DGP2 &0.00 & 0.00 & -0.00 & 0.89 & 0.91 & 0.90 & 0.06 & 0.05 & 0.05 \\ 

\midrule
DGP3 &0.01 & 0.01 & 0.02 & 0.89 & 0.90 & 0.90 & 0.48 & 0.41 & 0.38 \\ 

\midrule
DGP4 &-0.01 & -0.00 & -0.02 & 0.90 & 0.90 & 0.91 & 0.48 & 0.41 & 0.39 \\ 

\midrule
DGP5 &0.00 & 0.01 & -0.01 & 0.91 & 0.91 & 0.91 & 0.48 & 0.42 & 0.39 \\ 

\midrule
DGP6 & -0.01 & -0.01 & -0.02 & 0.89 & 0.88 & 0.87 & 0.06 & 0.05 & 0.05 \\ 

\midrule  
DGP7 & 0.00 & 0.00 & 0.00 & 0.89 & 0.89 & 0.89 & 0.06 & 0.05 & 0.05 \\ 

\midrule
DGP8 &0.01 & 0.01 & 0.00 & 0.86 & 0.85 & 0.84 & 0.41 & 0.35 & 0.32 \\ 

\midrule
DGP9 &-0.01 & -0.00 & 0.00 & 0.83 & 0.83 & 0.82 & 0.09 & 0.08 & 0.08 \\ 

\midrule
\bottomrule

\multicolumn{10}{p{11cm}}{\scriptsize{\it Notes:}. Simulations are based on 10,000 repetitions. CI: Confidence interval. Nominal coverage: $1-\alpha= 0.9$. The DGPs are described in Section \ref{sec:simulations}.}
\end{tabular}
\addtolength{\tabcolsep}{4pt}    
\label{tab:empirical_mc_large_T0}
\end{table}
\normalsize
\linespread{1.25}

\subsection{Performance with small $T_1$}
Figures \ref{fig:coverage_small_T1} and \ref{fig:length_small_T1} display the coverage and average length of the confidence intervals based on the $t$-test for DGP1--DGP9 when $T_0=30$ and $T_1\in \{10,12,14,16\}$ (recall that $T_1=16$ in the main text). Given the small sample setting, we choose $K=3$. Overall, the performance of the $t$-test does not change much if the sample size decreases from $T_1=16$ to $T_1=10$, except under DGP8, which is not covered by our theory. This finding is not surprising given that the persistence in the prediction errors is relatively low ($\rho_u=0.31$) and changing $T_1$ does not affect the estimation of the weights. Under DGP8, and to a lesser extent under DGP6, coverage decreases as $T_1$ increases. This result is consistent with our theoretical results, which require $T_0$ to be much larger than $T_1$ when there are deviations from a common nonstationarity.

\begin{figure}[H]
\caption{Coverage with small $T_1$}
	\label{fig:coverage_small_T1}
	\begin{center}
		\includegraphics[width=0.32\textwidth,trim = {0 0cm 0 1cm}]{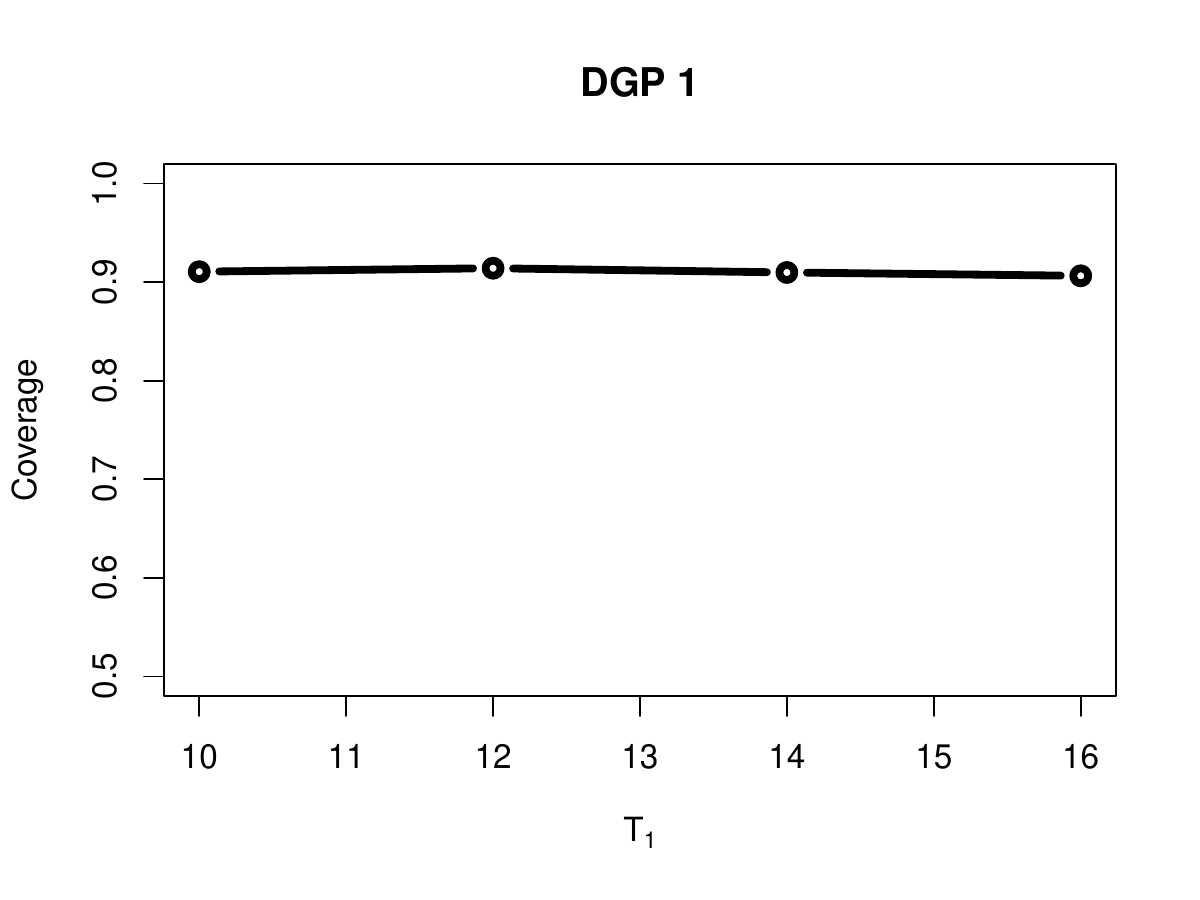}
  		\includegraphics[width=0.32\textwidth,trim = {0 0cm 0 1cm}]{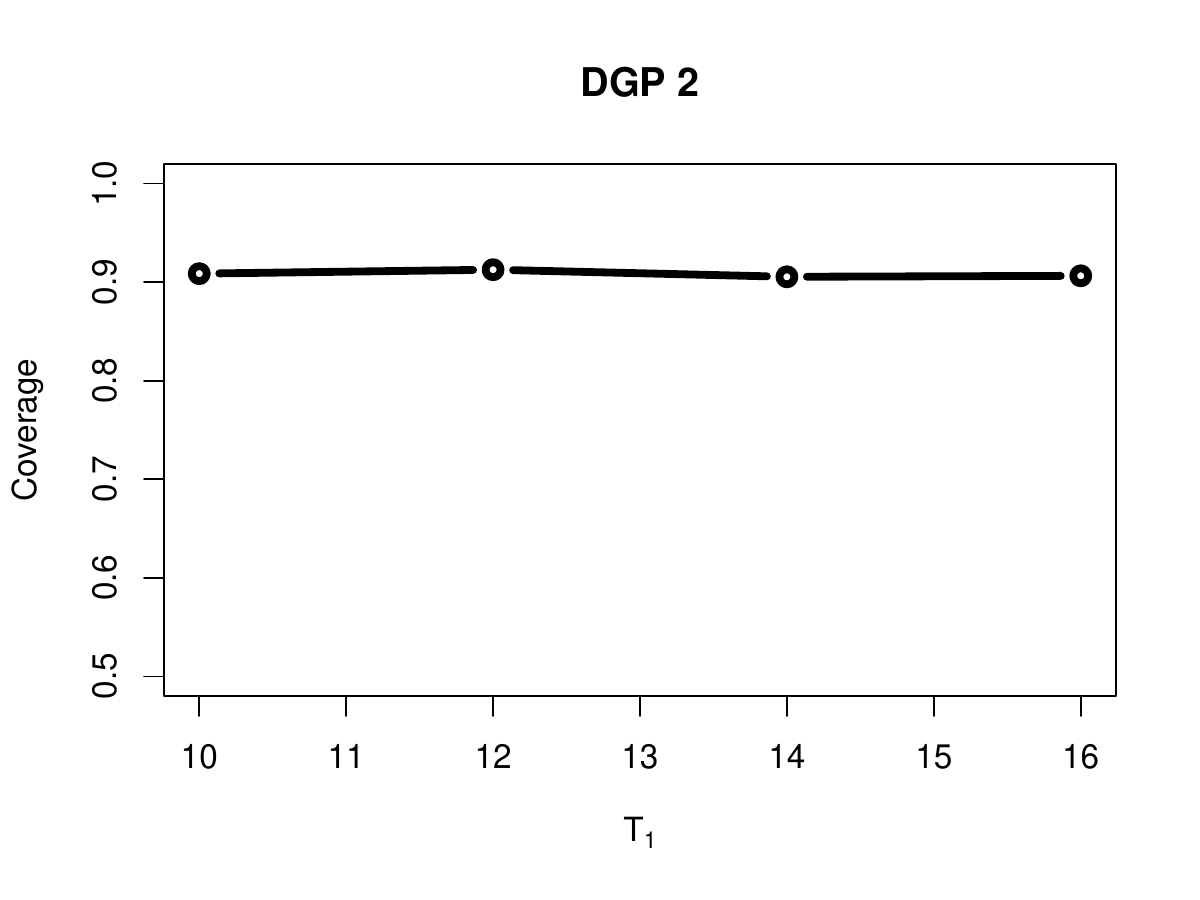}
        \includegraphics[width=0.32\textwidth,trim = {0 0cm 0 1cm}]{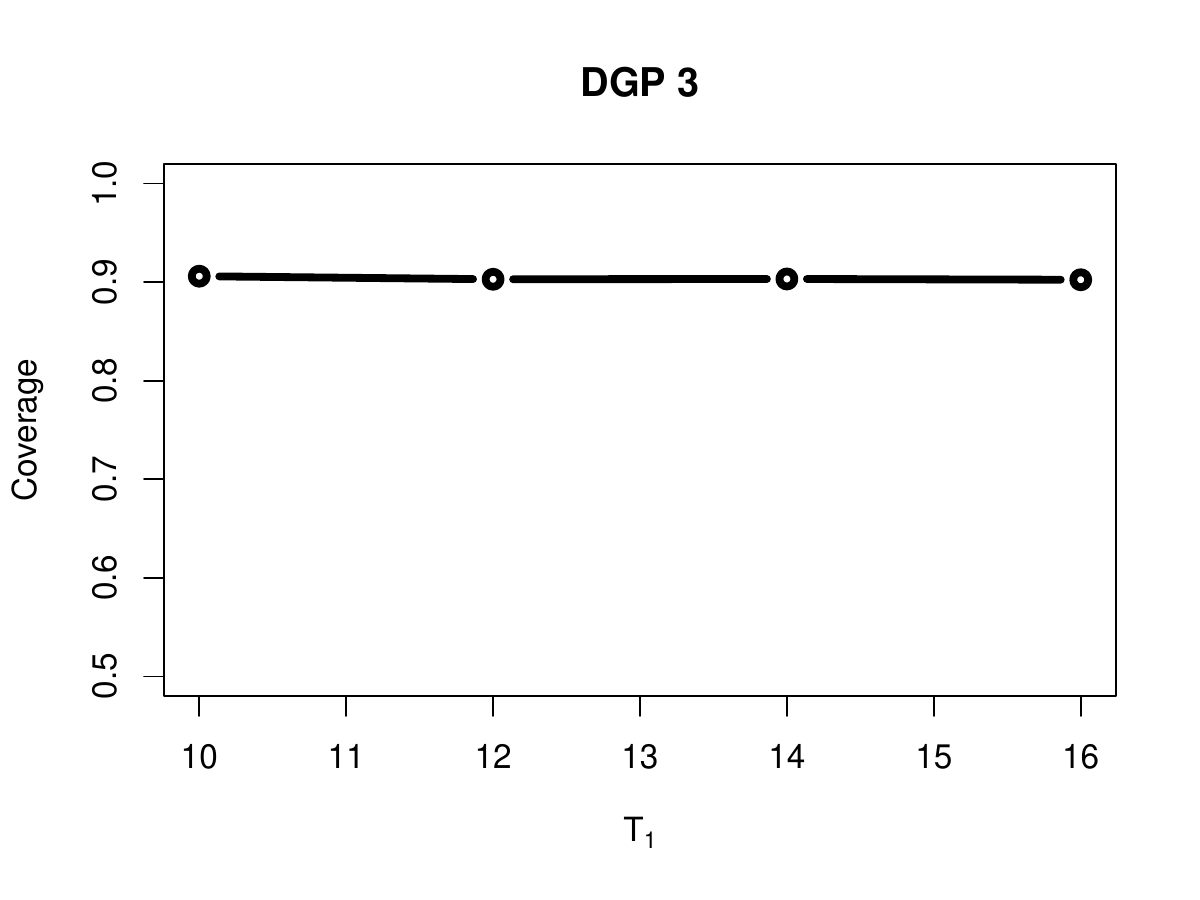}
        
  		\includegraphics[width=0.32\textwidth,trim = {0 0cm 0 1cm}]{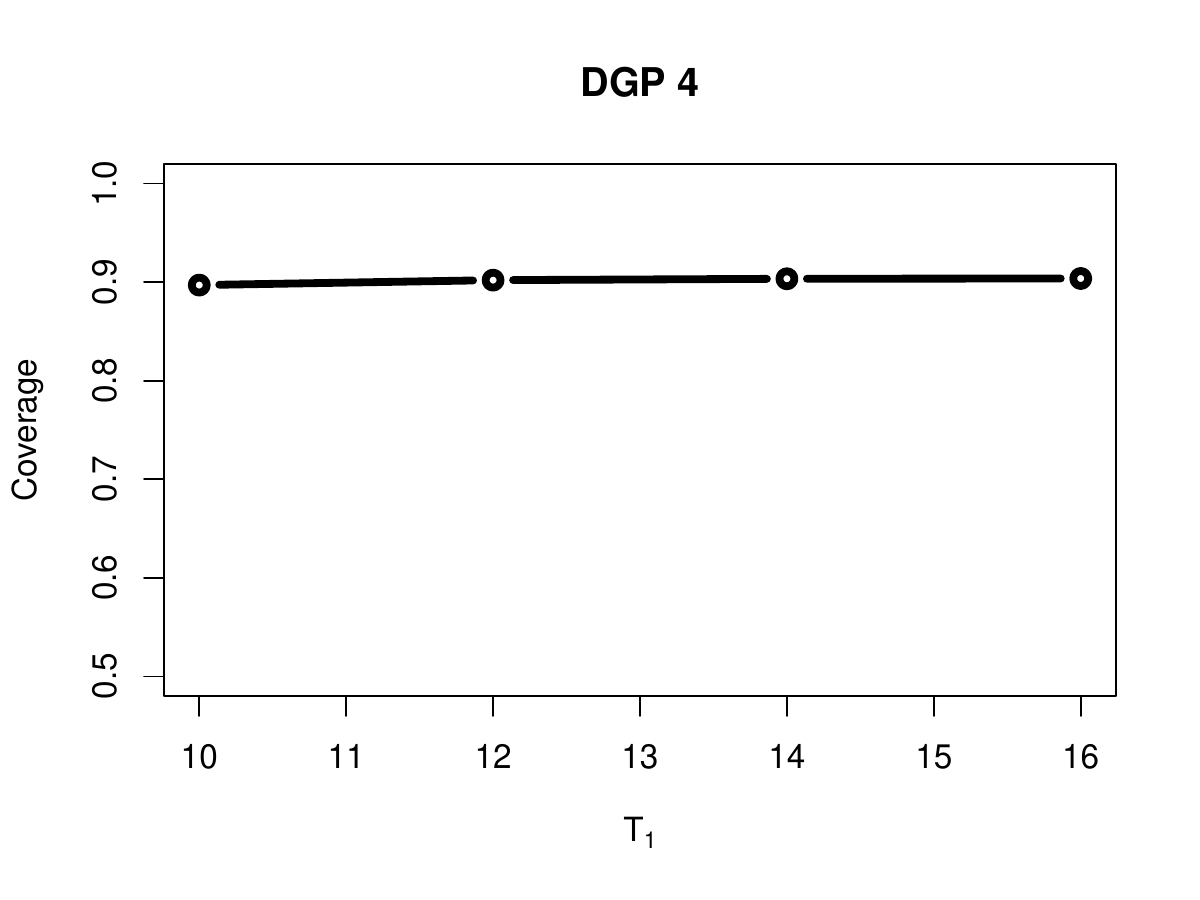}
      	\includegraphics[width=0.32\textwidth,trim = {0 0cm 0 1cm}]{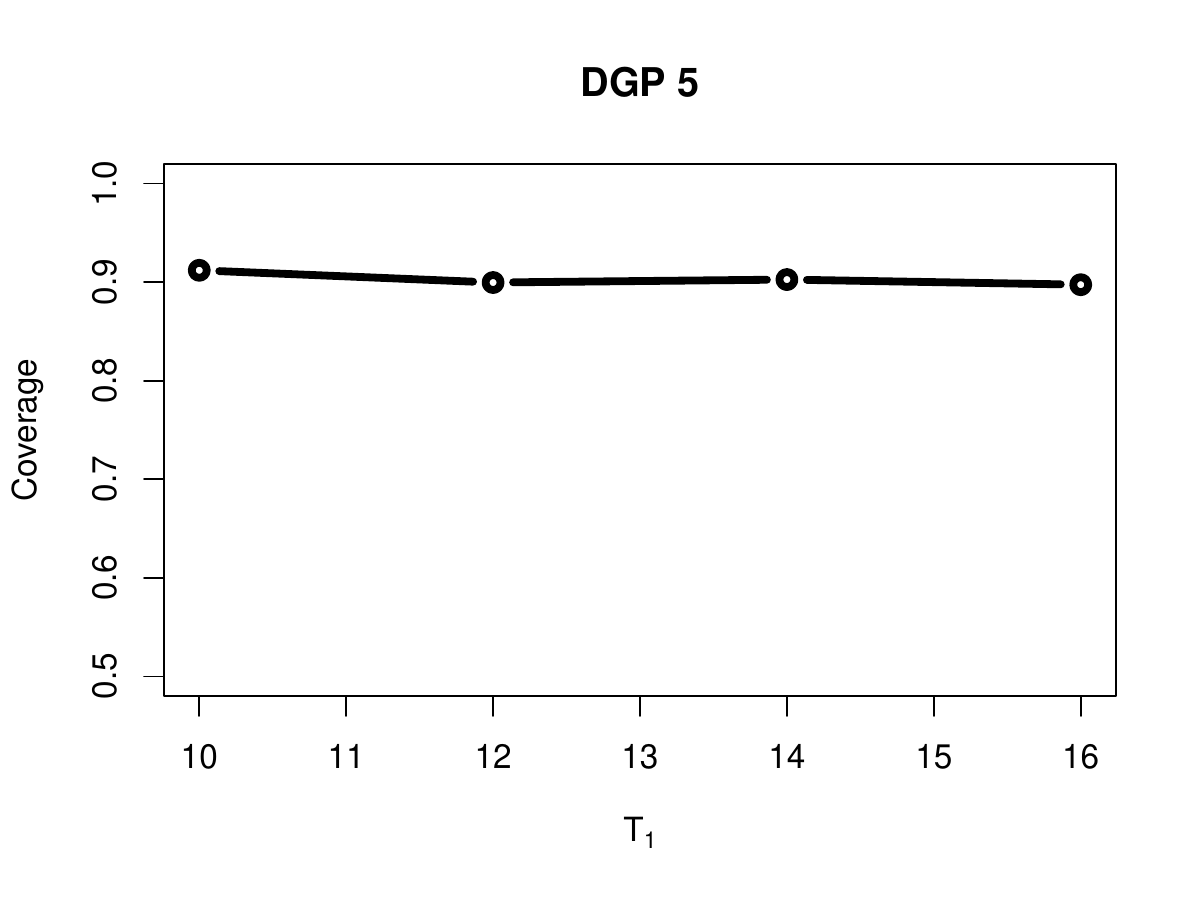}
      	\includegraphics[width=0.32\textwidth,trim = {0 0cm 0 1cm}]{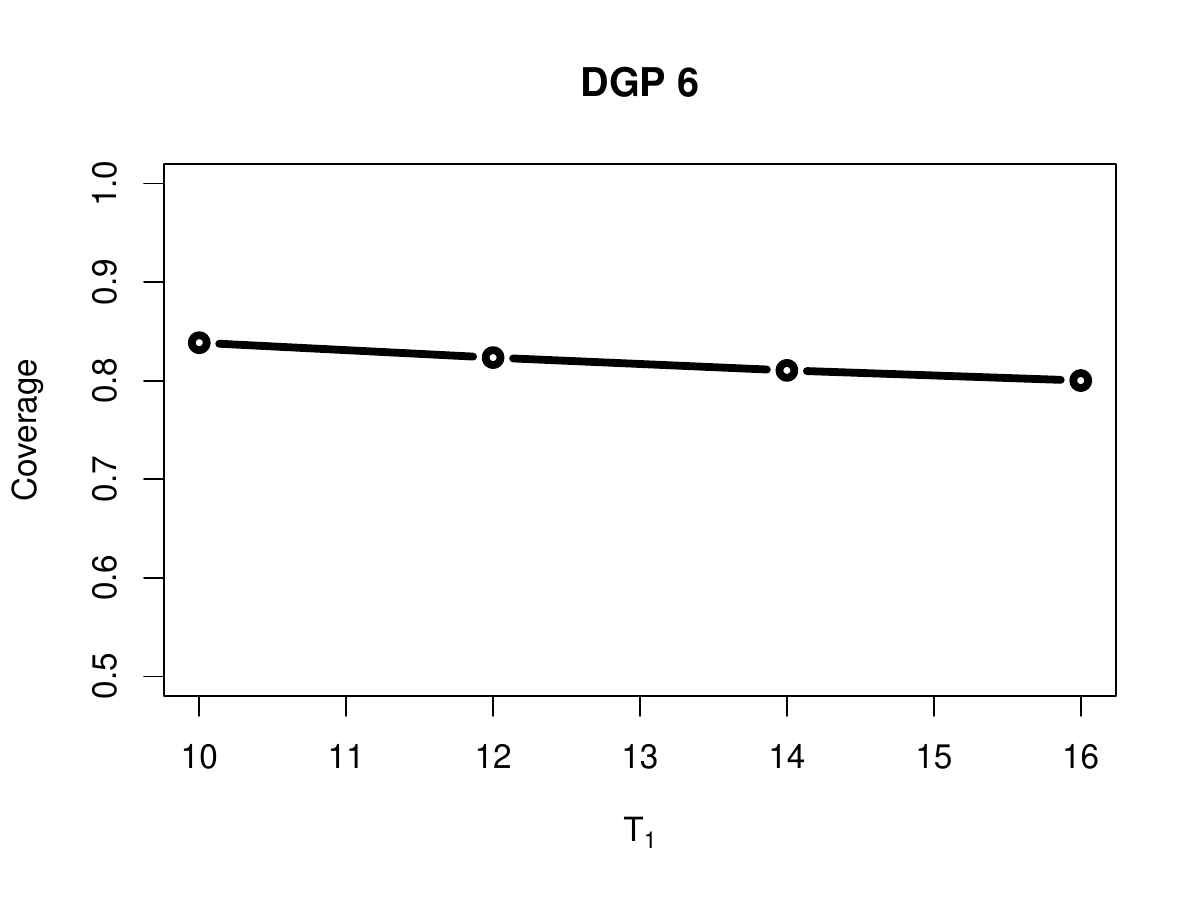}

  		\includegraphics[width=0.32\textwidth,trim = {0 0cm 0 1cm}]{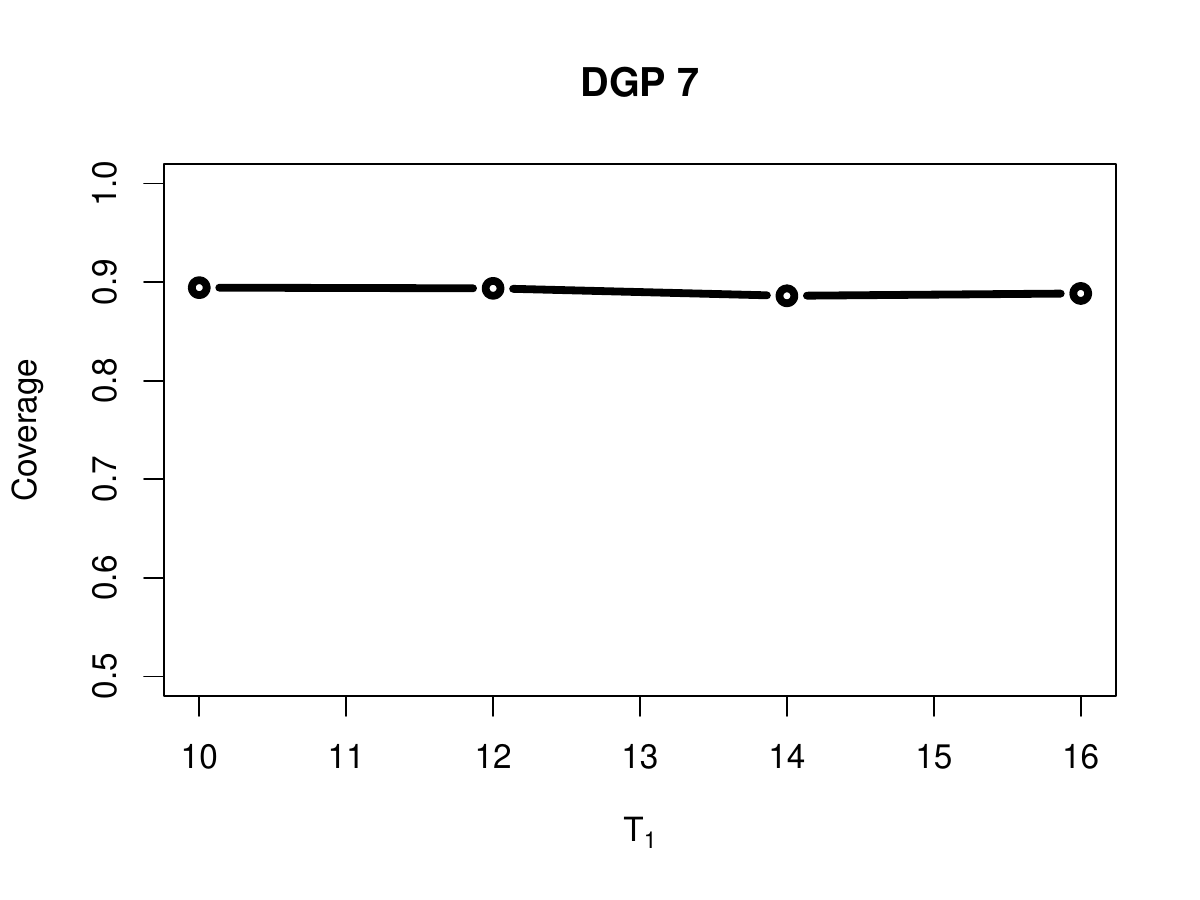}
      	\includegraphics[width=0.32\textwidth,trim = {0 0cm 0 1cm}]{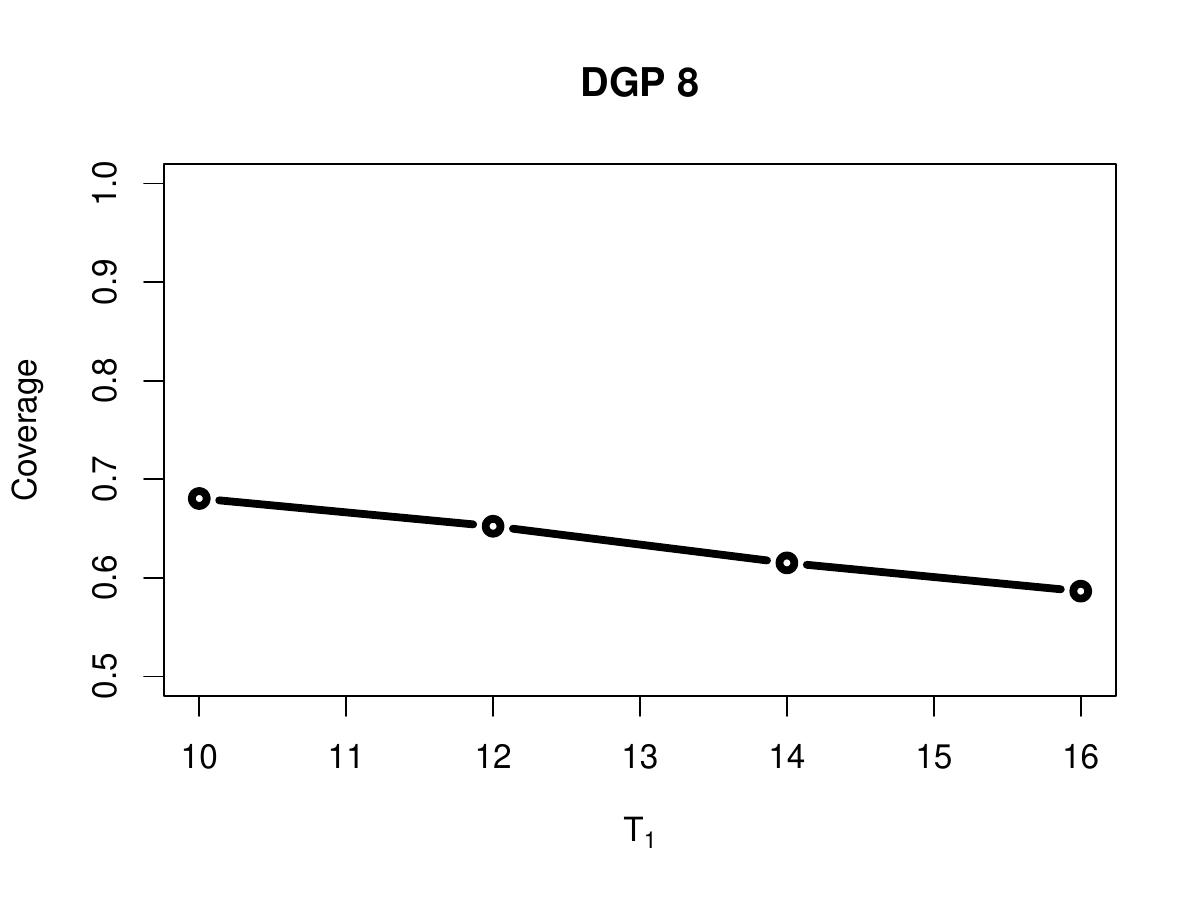}
      	\includegraphics[width=0.32\textwidth,trim = {0 0cm 0 1cm}]{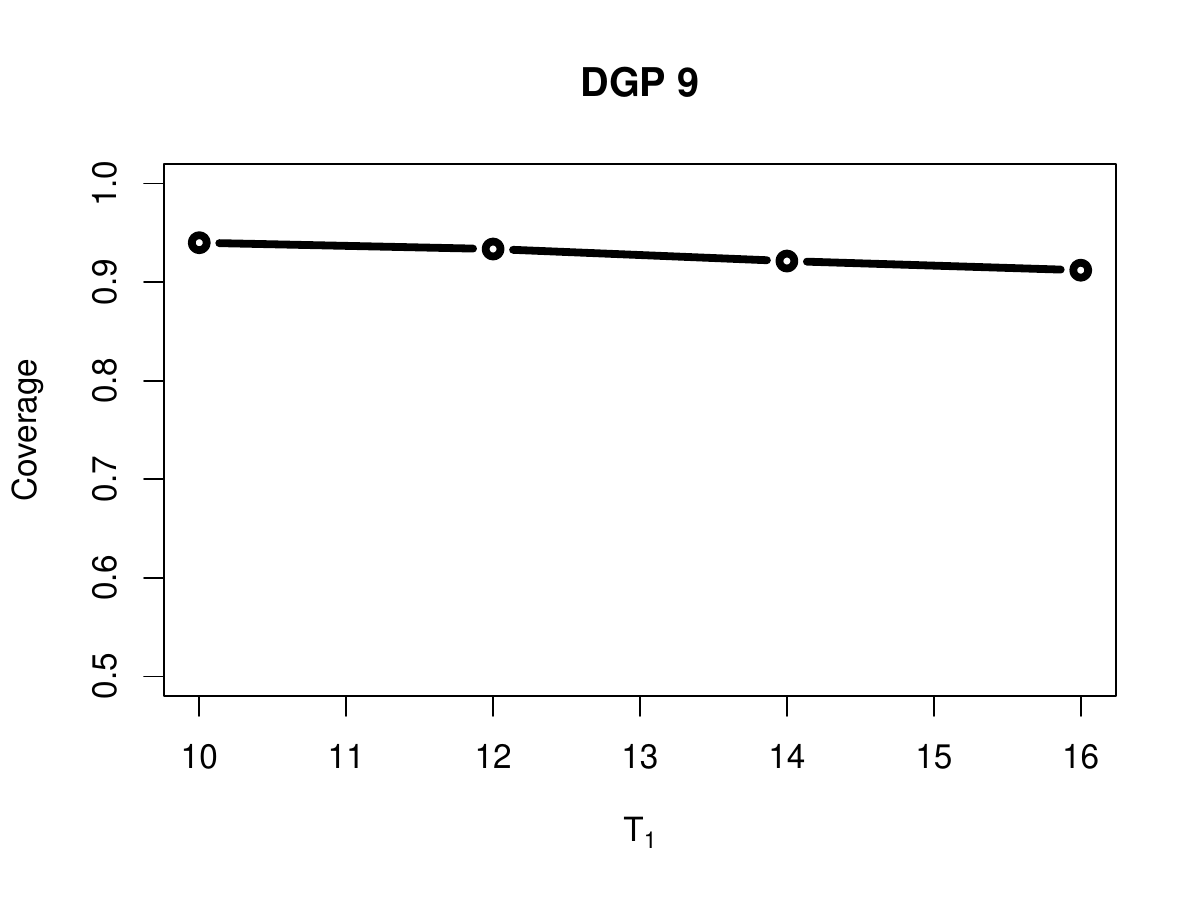}            
   
	\end{center}
\scriptsize{\textit{Notes:} Simulations with 10,000 repetitions. Nominal coverage: $1-\alpha=0.9$. The DGPs are described in Section \ref{sec:simulations}.}	
\end{figure}

\begin{figure}[H]
\caption{Length with small $T_1$}
	\label{fig:length_small_T1}
	\begin{center}
		\includegraphics[width=0.32\textwidth,trim = {0 0cm 0 1cm}]{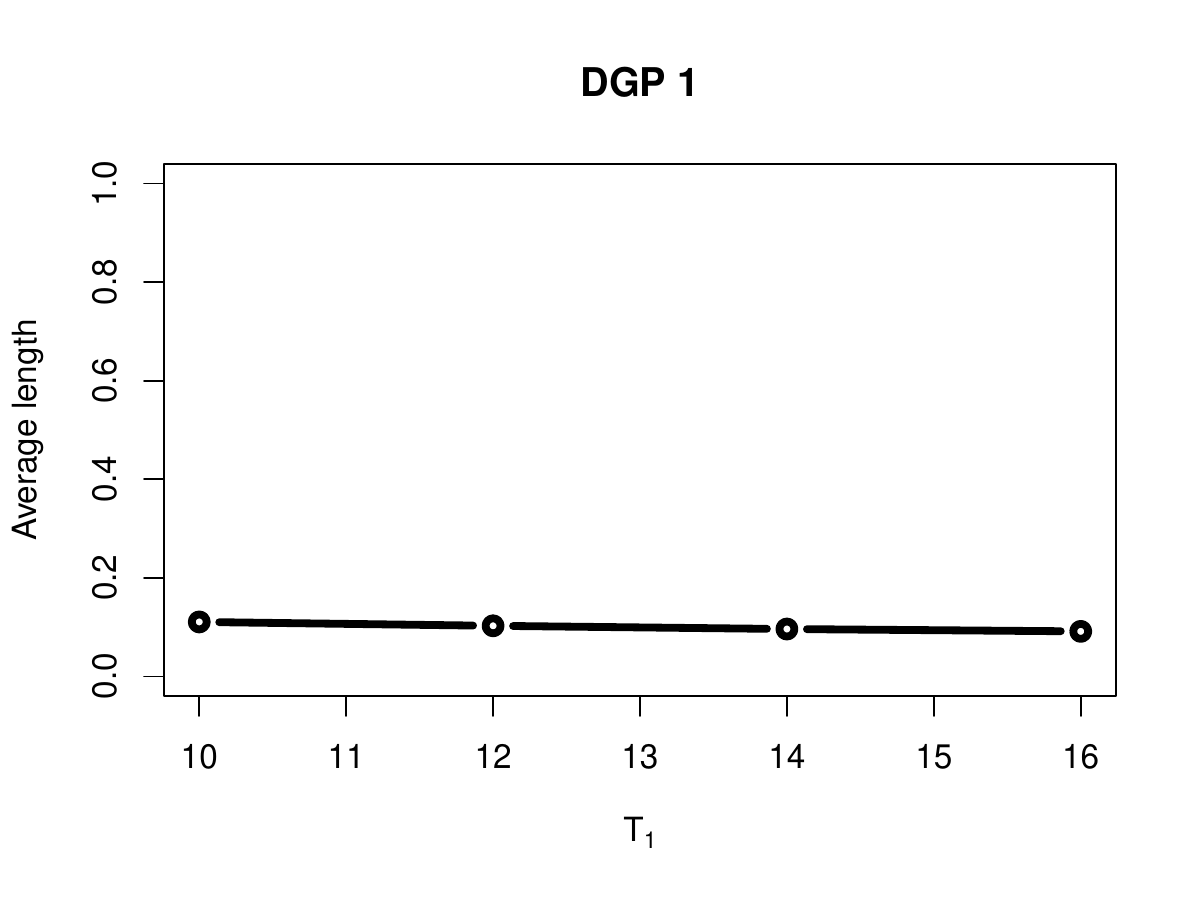}
  		\includegraphics[width=0.32\textwidth,trim = {0 0cm 0 1cm}]{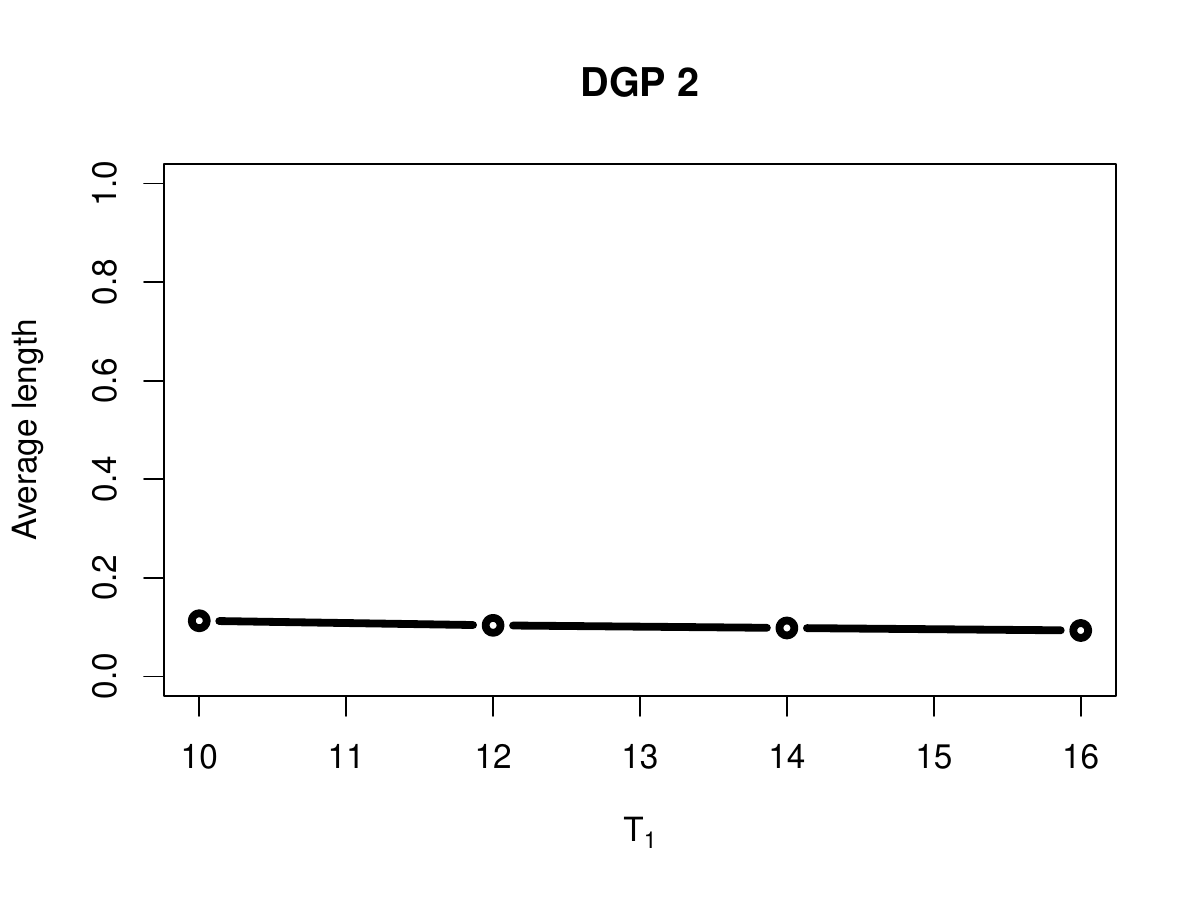}
        \includegraphics[width=0.32\textwidth,trim = {0 0cm 0 1cm}]{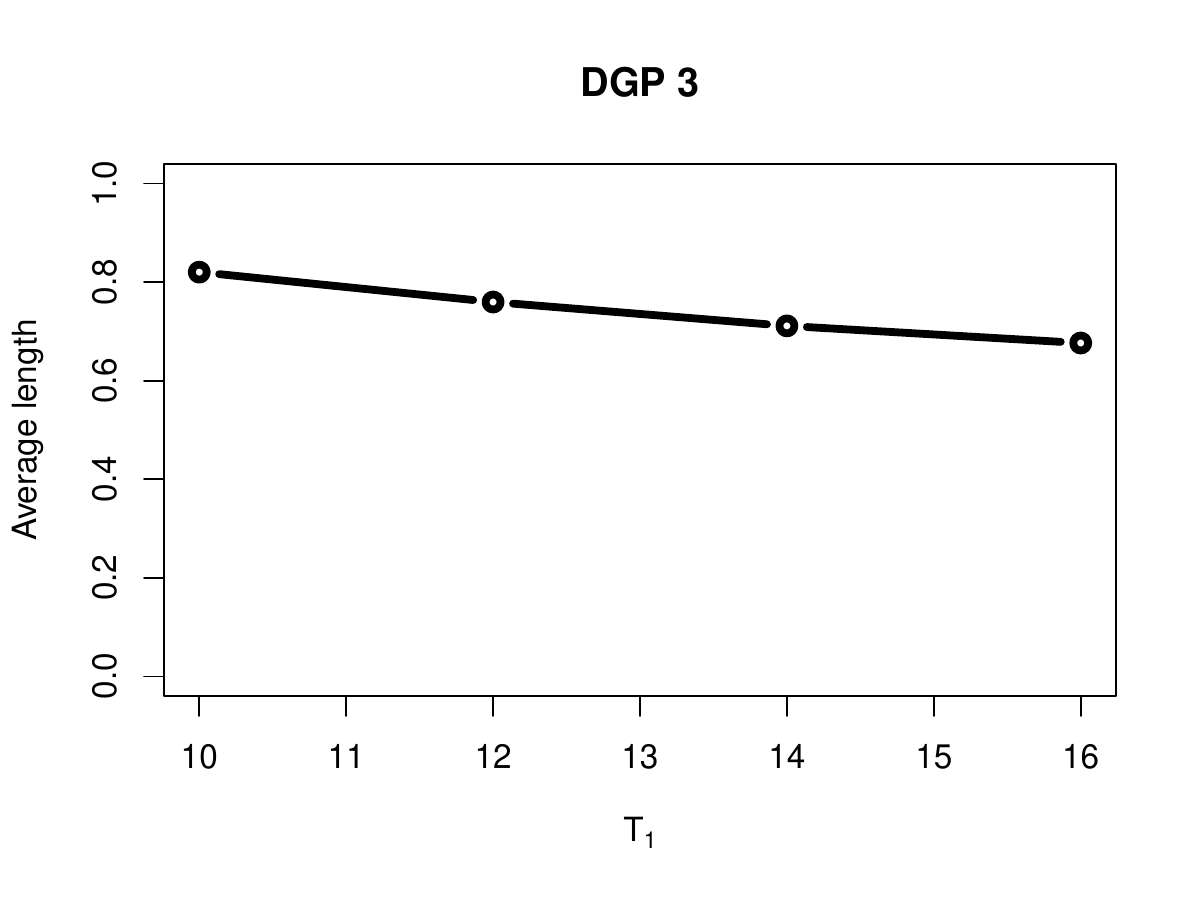}
        
  		\includegraphics[width=0.32\textwidth,trim = {0 0cm 0 1cm}]{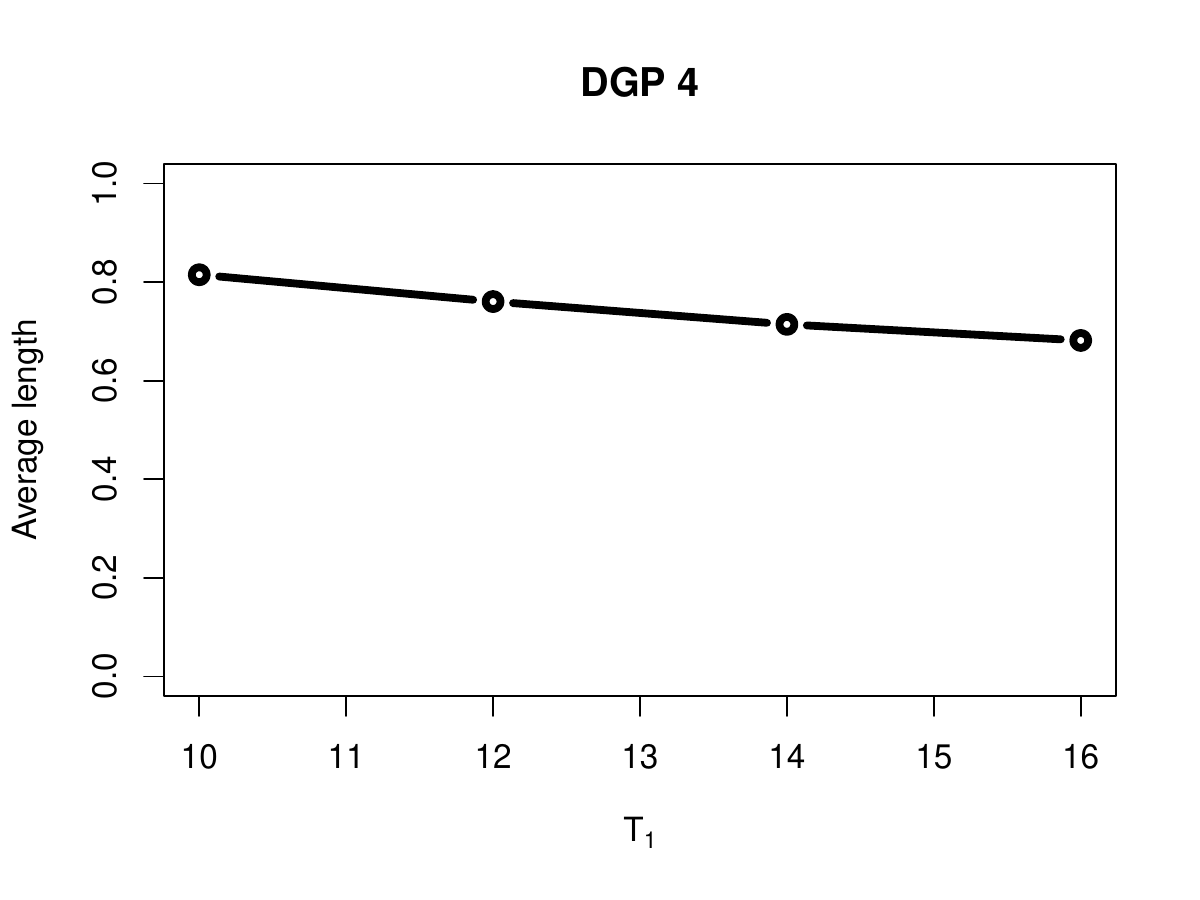}
      	\includegraphics[width=0.32\textwidth,trim = {0 0cm 0 1cm}]{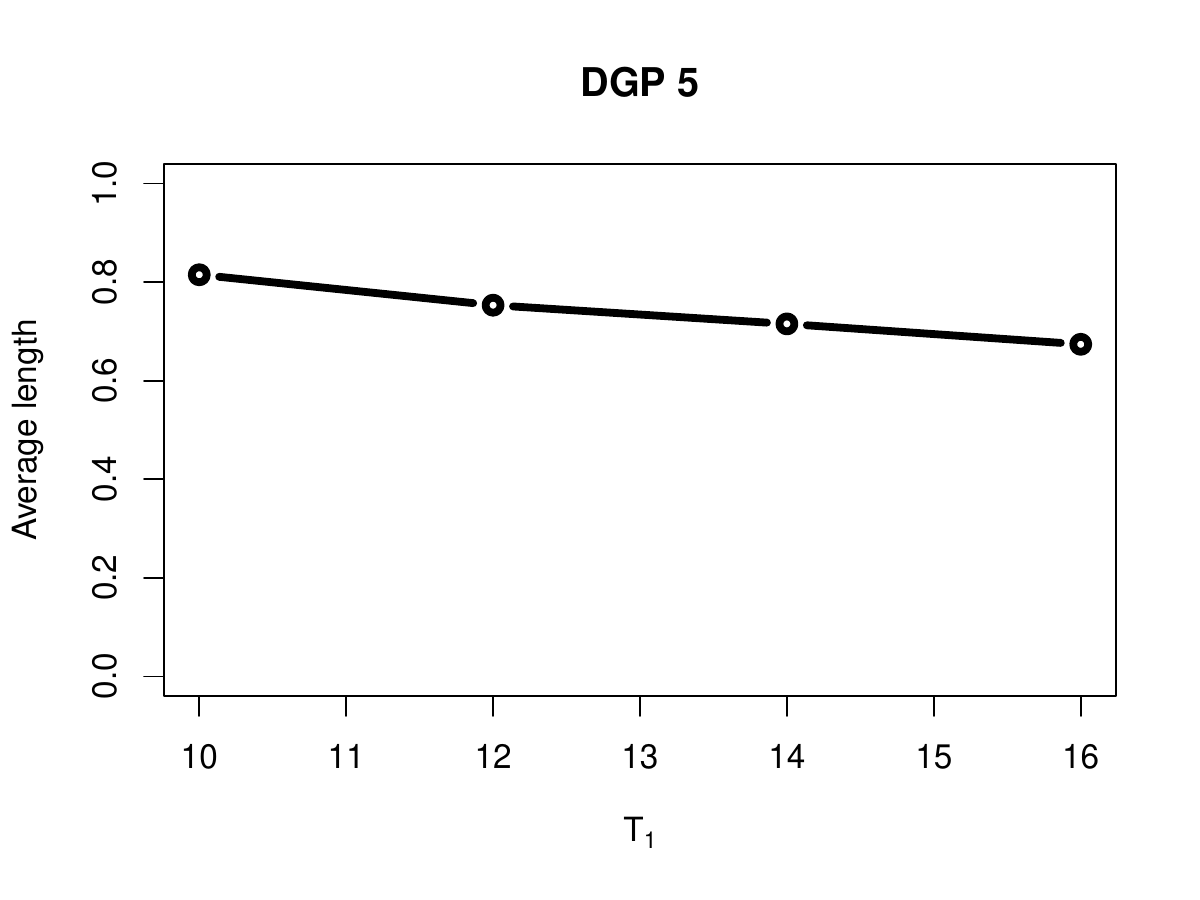}
      	\includegraphics[width=0.32\textwidth,trim = {0 0cm 0 1cm}]{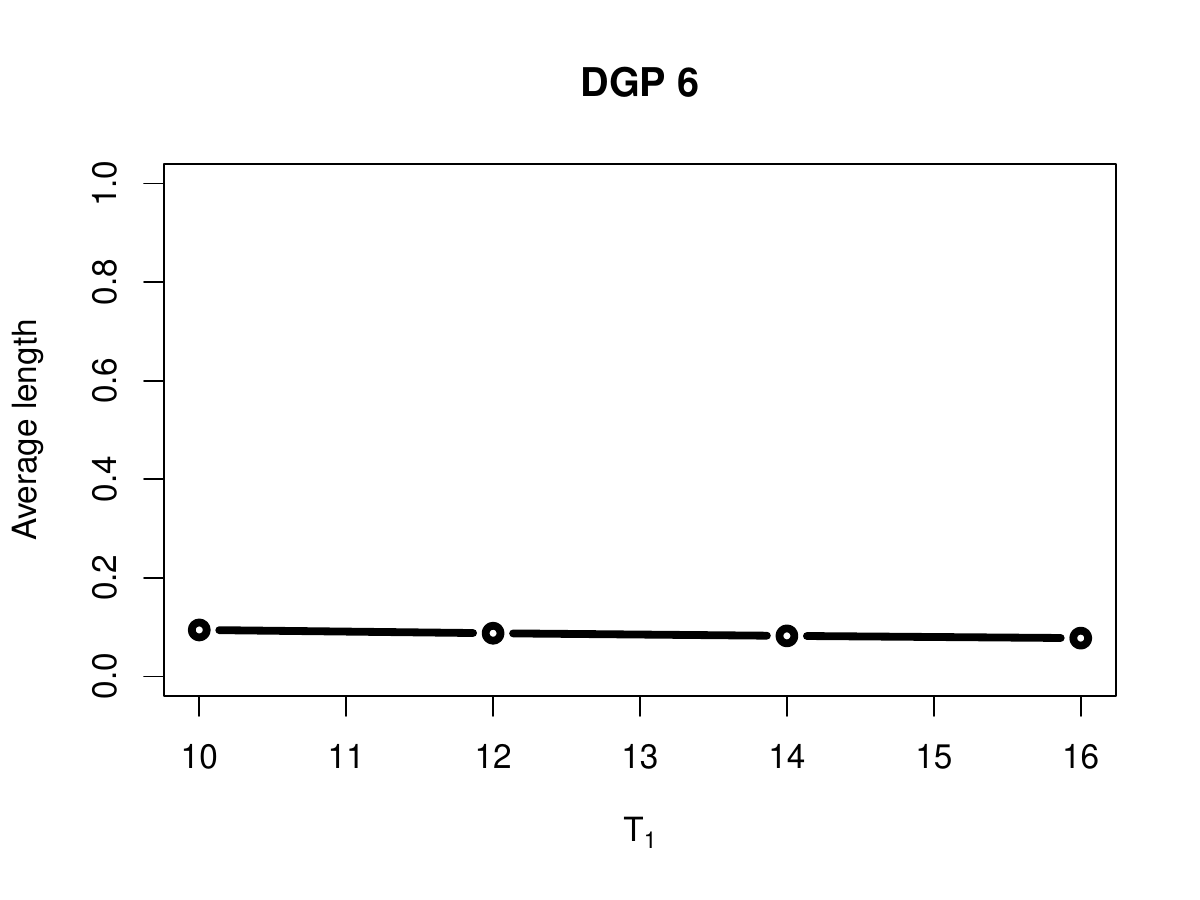}

  		\includegraphics[width=0.32\textwidth,trim = {0 0cm 0 1cm}]{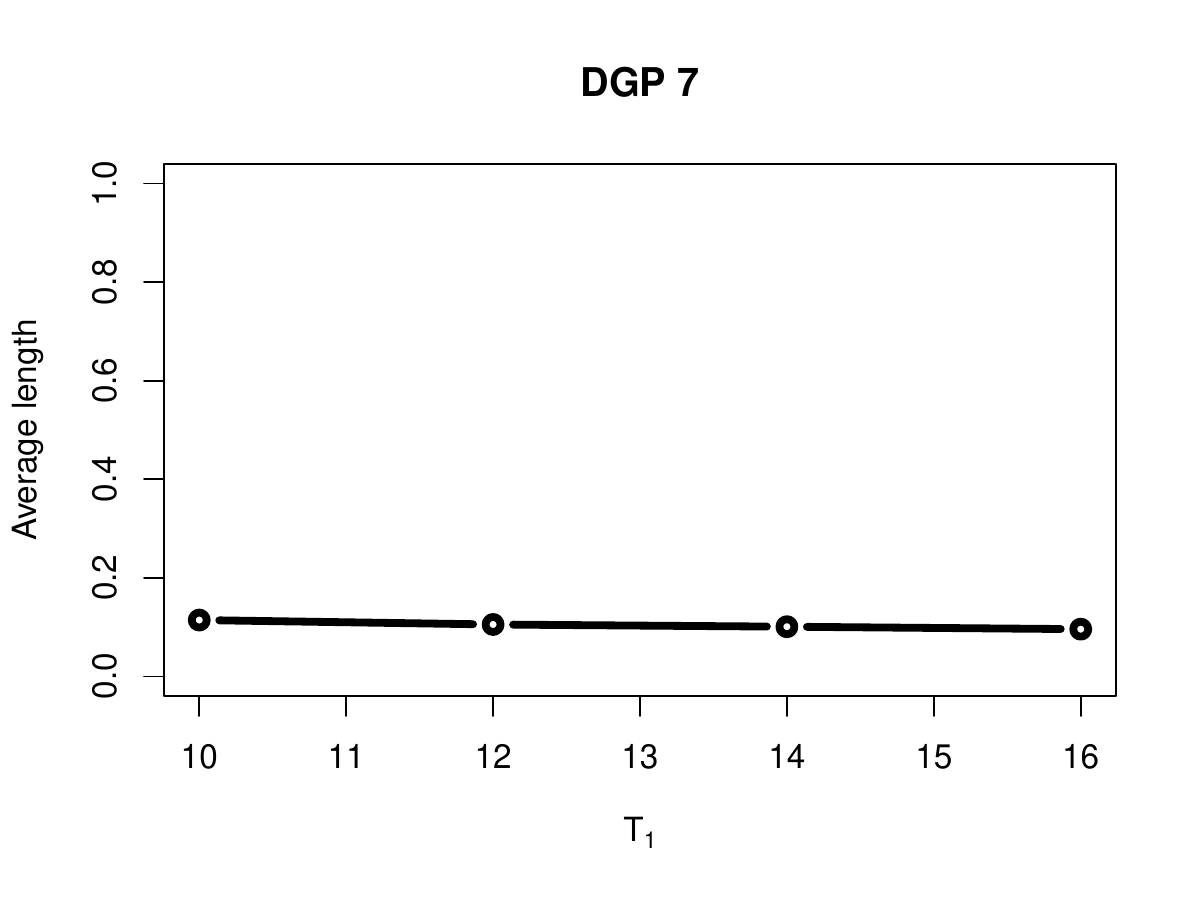}
      	\includegraphics[width=0.32\textwidth,trim = {0 0cm 0 1cm}]{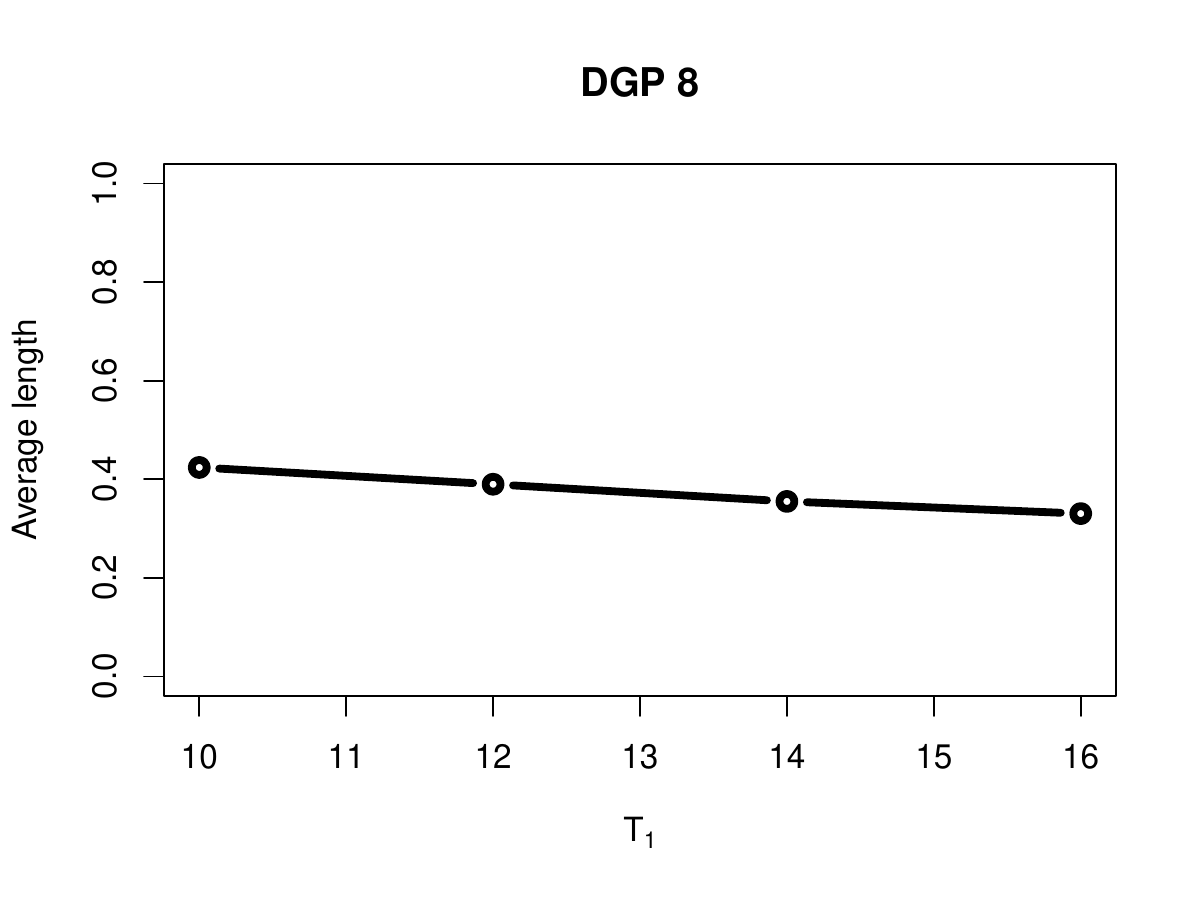}
      	\includegraphics[width=0.32\textwidth,trim = {0 0cm 0 1cm}]{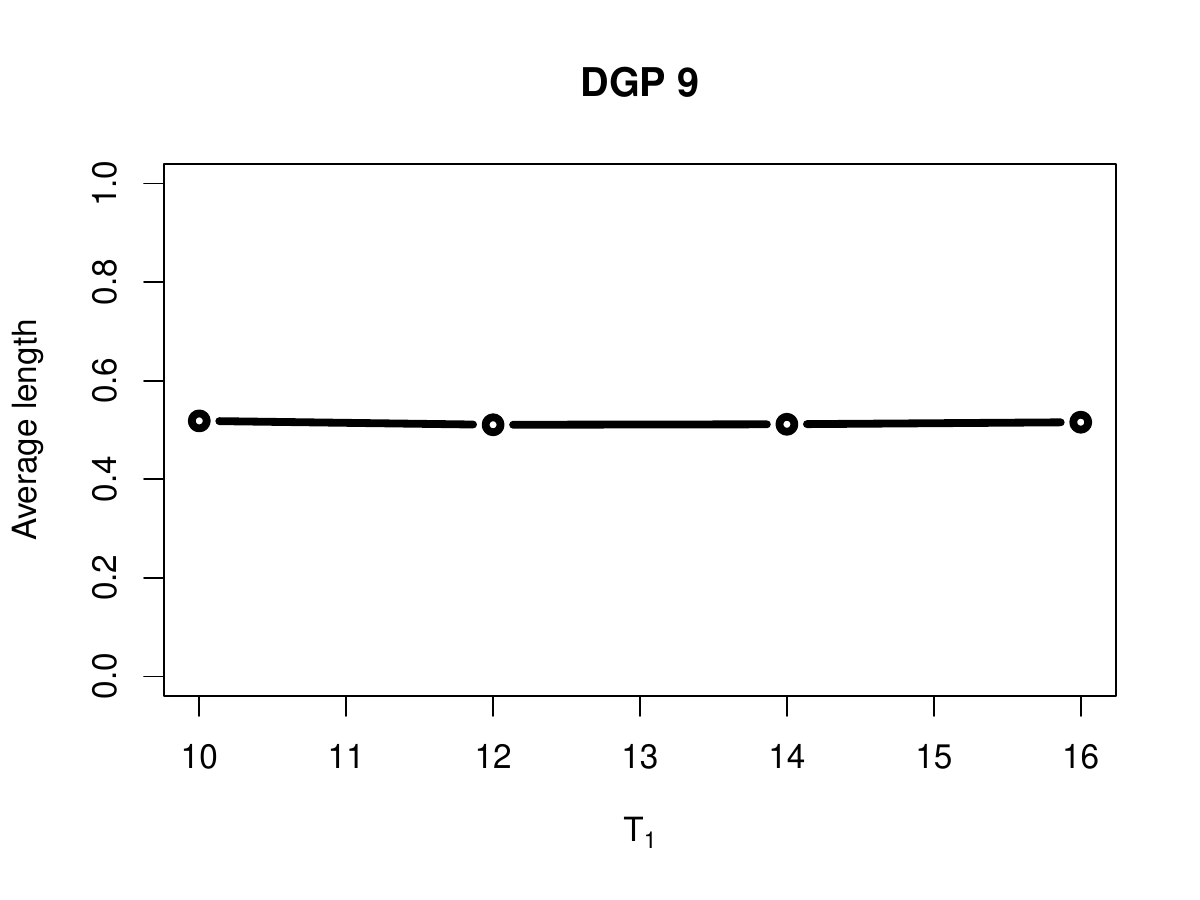}            

	\end{center}
\scriptsize{\textit{Notes:} Simulations with 10,000 repetitions. Nominal coverage: $1-\alpha=0.9$. The DGPs are described in Section \ref{sec:simulations}.}	
\end{figure}

\subsection{Comparison to permutation method of \citet{abadie10sc}}
\label{app:comparison_permutation}

The permutation approach of \citet{abadie10sc} is a popular inference method in SC applications. It is fundamentally different from the $t$-test. First, it is designed to test sharp null hypotheses, such as the null of no effect whatsoever, whereas the $t$-test is designed to make inferences on the ATT. Second, it is typically motivated from a design-based perspective and requires random assignment of the treatment to be valid, whereas the $t$-test is developed within a sampling-based framework and allows for selection on observables and unobservables (see Appendix \ref{app:random_treatment}). 

Given these differences, we focus on size and power in a setting with constant effects to compare both methods. We implement the permutation test as described in Section 3.5 of \citet{abadie2021using}. Figure \ref{fig:comparison_permutation} displays the empirical rejection probabilities for both tests over a grid of alternatives $\alpha_t=a$ for $t\ge T_0$. We consider the same DGPs as in the main text and two different sample sizes $(T_0,T_1,N)=(30,16,14)$ (``small'') and $(T_0,T_1,N)=(150,16,14)$ (``large''). In light of our simulation results, we choose $K=4$ for the small sample case and $K=8$ for the large sample case.

The simulation results show that the $t$-test is more powerful than the permutation approach overall, despite testing a weaker null hypothesis. The permutation test is only comparable to the $t$-test in terms of power under stationarity and correct specification (DGP1 and DGP2) when $T_0=30$. For  the other DGPs and when $T_0=150$, the $t$-test is more powerful, and the permutation test can have essentially no power in some cases. We note that while the permutation test typically underrejects under our specific DGPs, \citet{hahn17inference} show that this test can also substantially overreject under DGPs with sampling-based variation. 

\begin{figure}[H]
\caption{Comparison to permutation approach}
	\label{fig:comparison_permutation}
	\begin{center}
		\includegraphics[width=0.32\textwidth,trim = {0 0cm 0 1cm}]{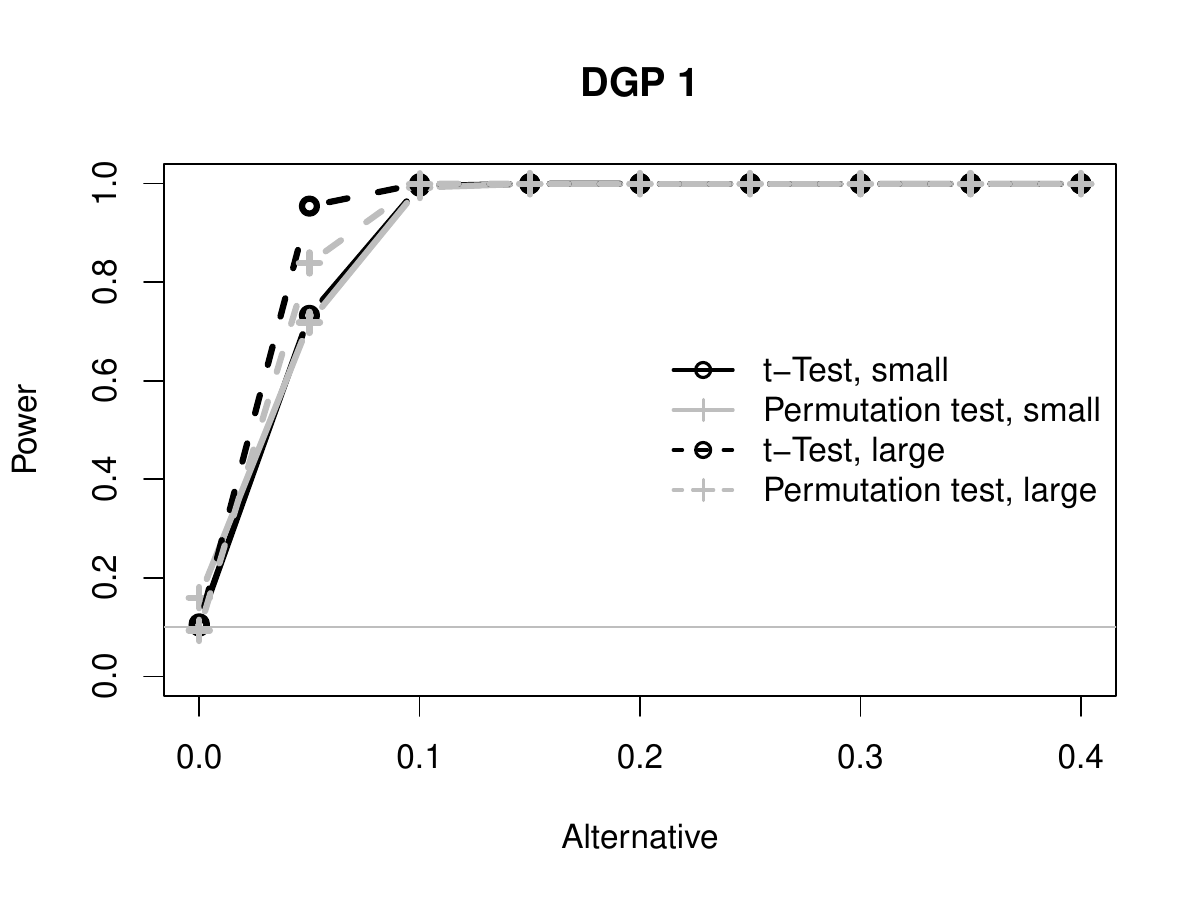}
  		\includegraphics[width=0.32\textwidth,trim = {0 0cm 0 1cm}]{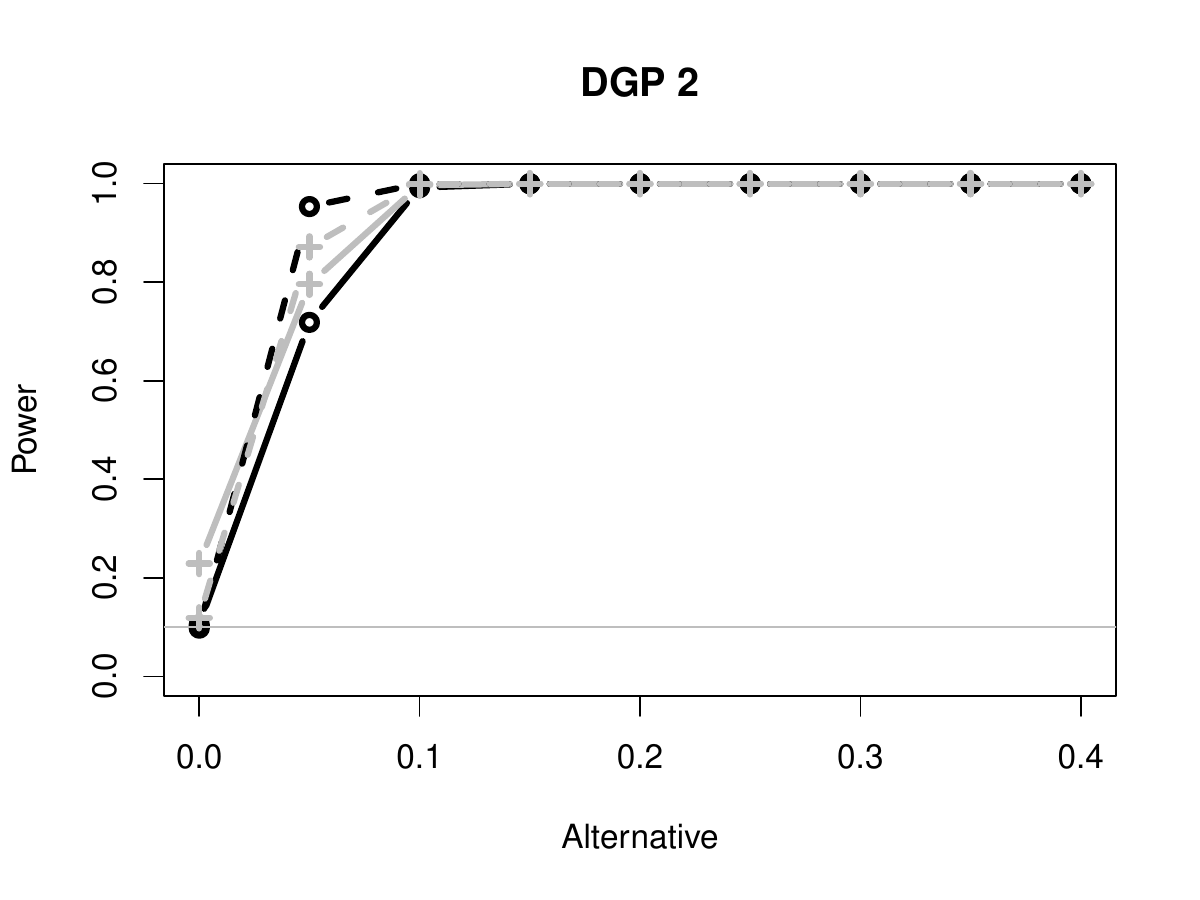}
        \includegraphics[width=0.32\textwidth,trim = {0 0cm 0 1cm}]{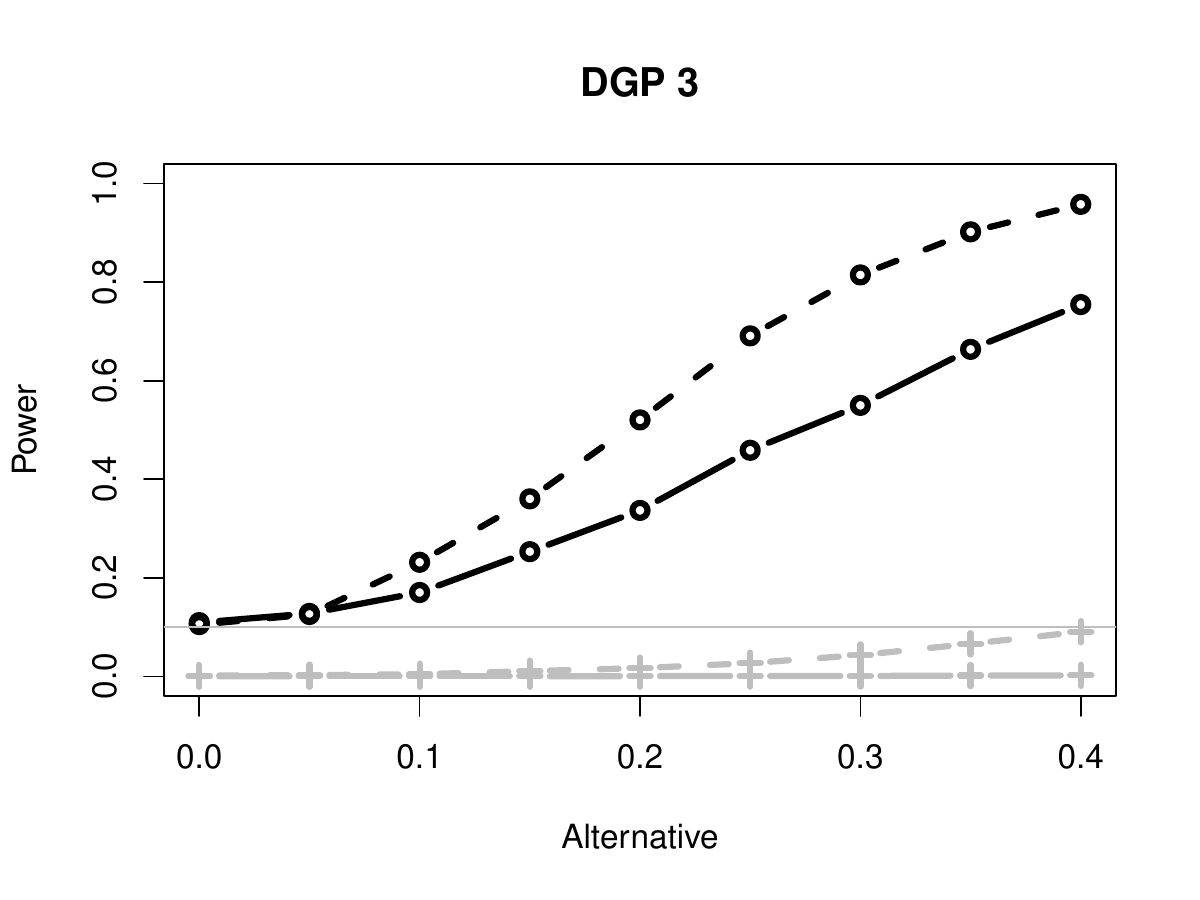}
        
  		\includegraphics[width=0.32\textwidth,trim = {0 0cm 0 1cm}]{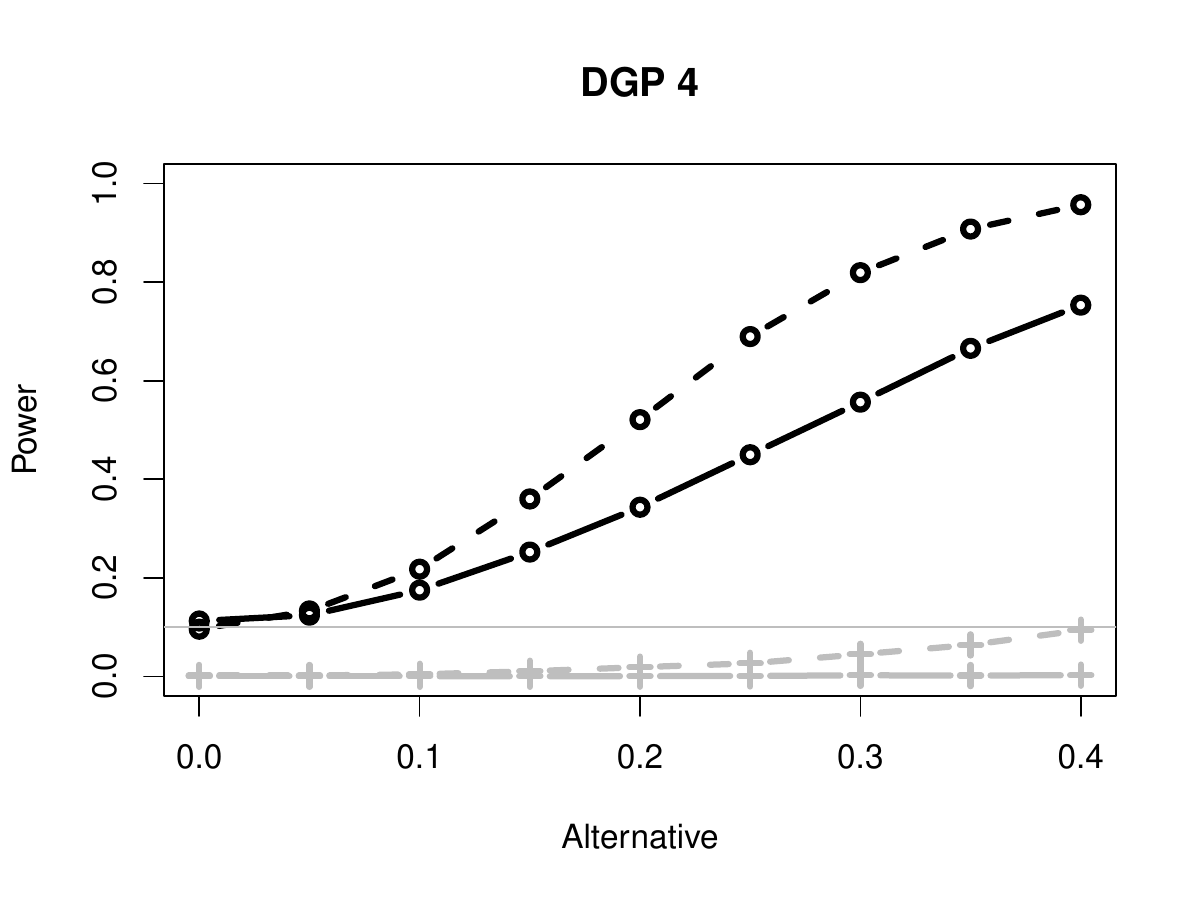}
      	\includegraphics[width=0.32\textwidth,trim = {0 0cm 0 1cm}]{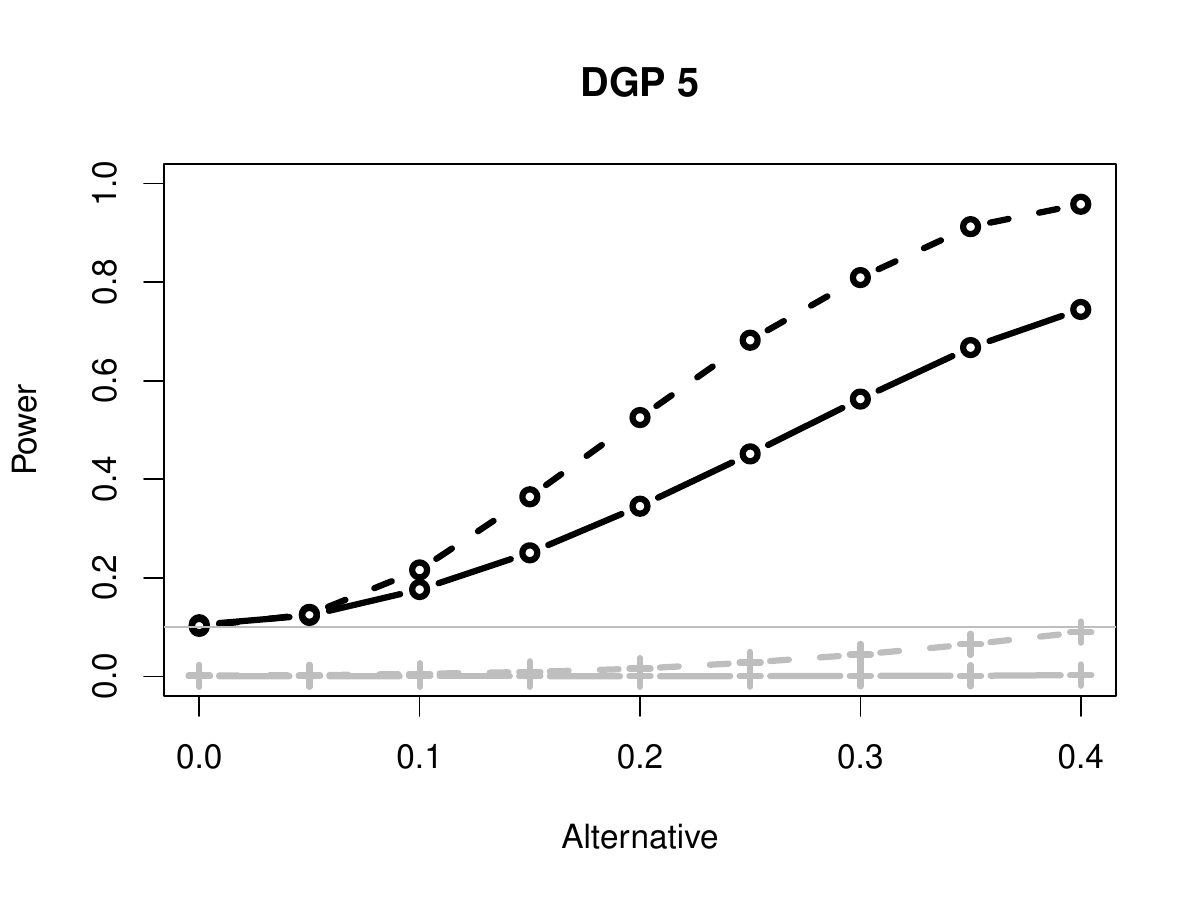}
      	\includegraphics[width=0.32\textwidth,trim = {0 0cm 0 1cm}]{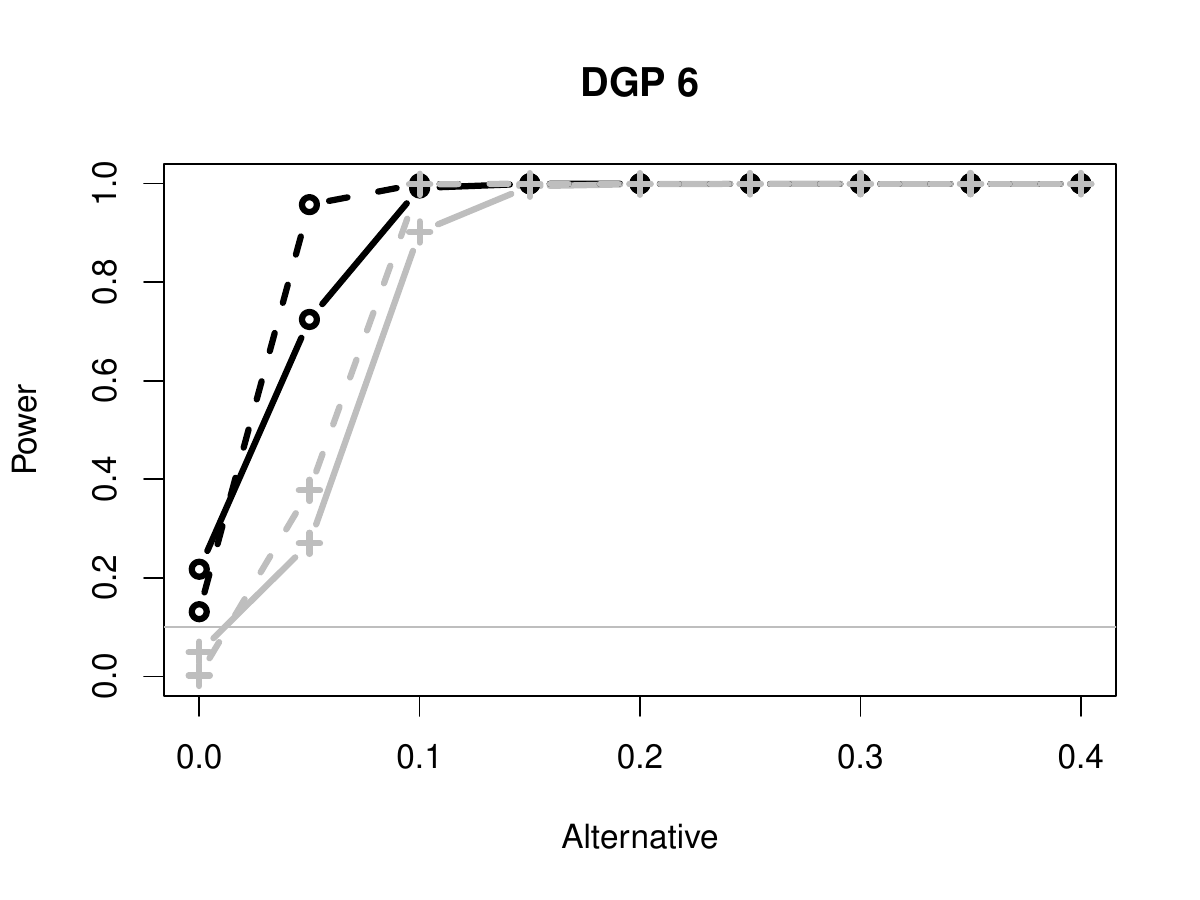}

  		\includegraphics[width=0.32\textwidth,trim = {0 0cm 0 1cm}]{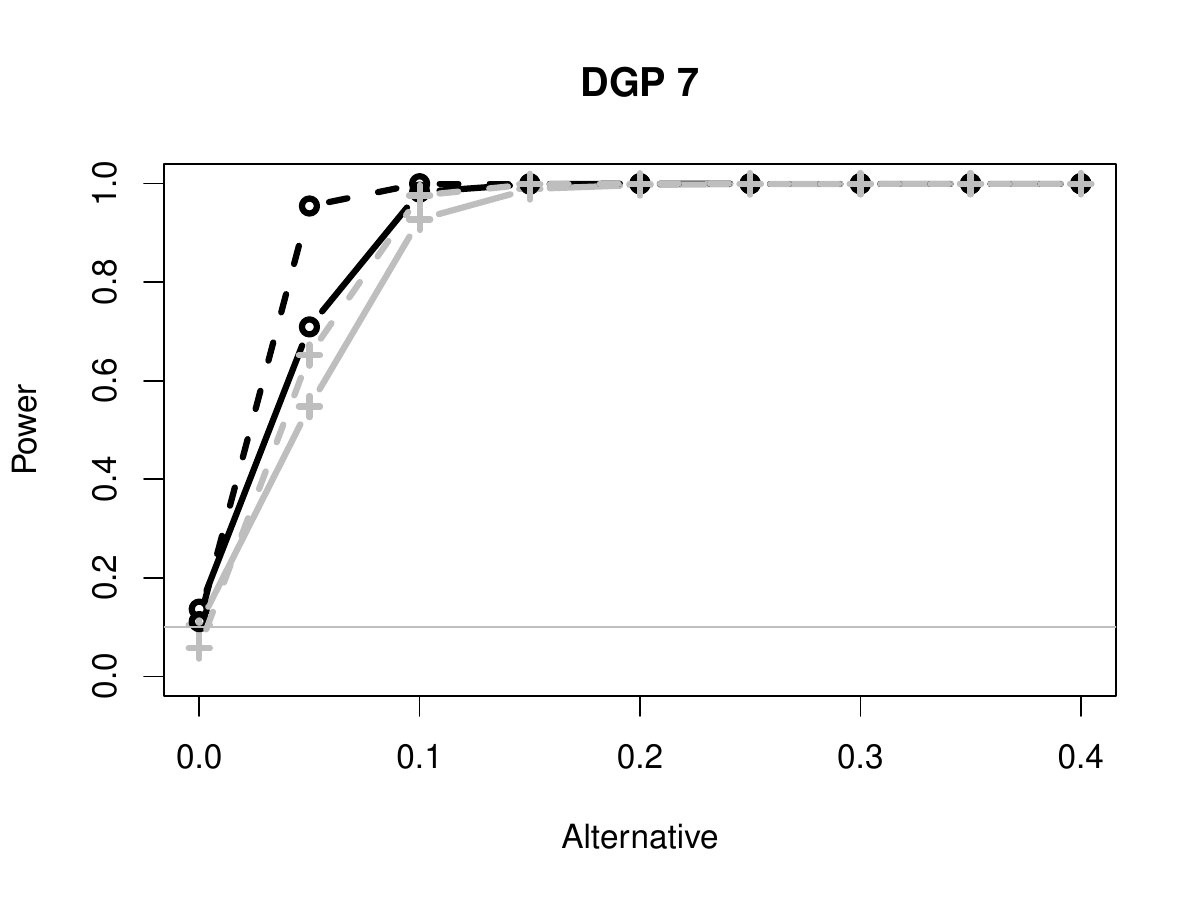}
      	\includegraphics[width=0.32\textwidth,trim = {0 0cm 0 1cm}]{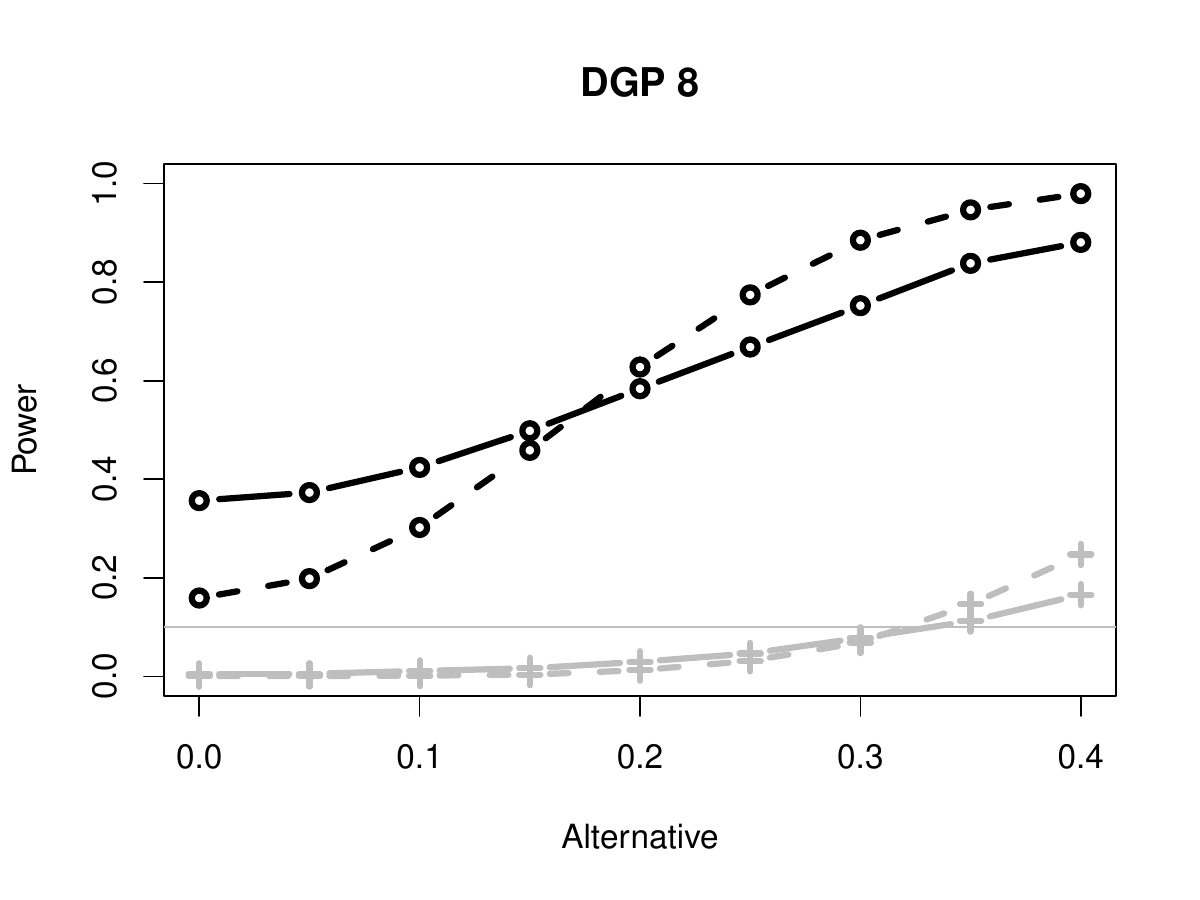}
      	\includegraphics[width=0.32\textwidth,trim = {0 0cm 0 1cm}]{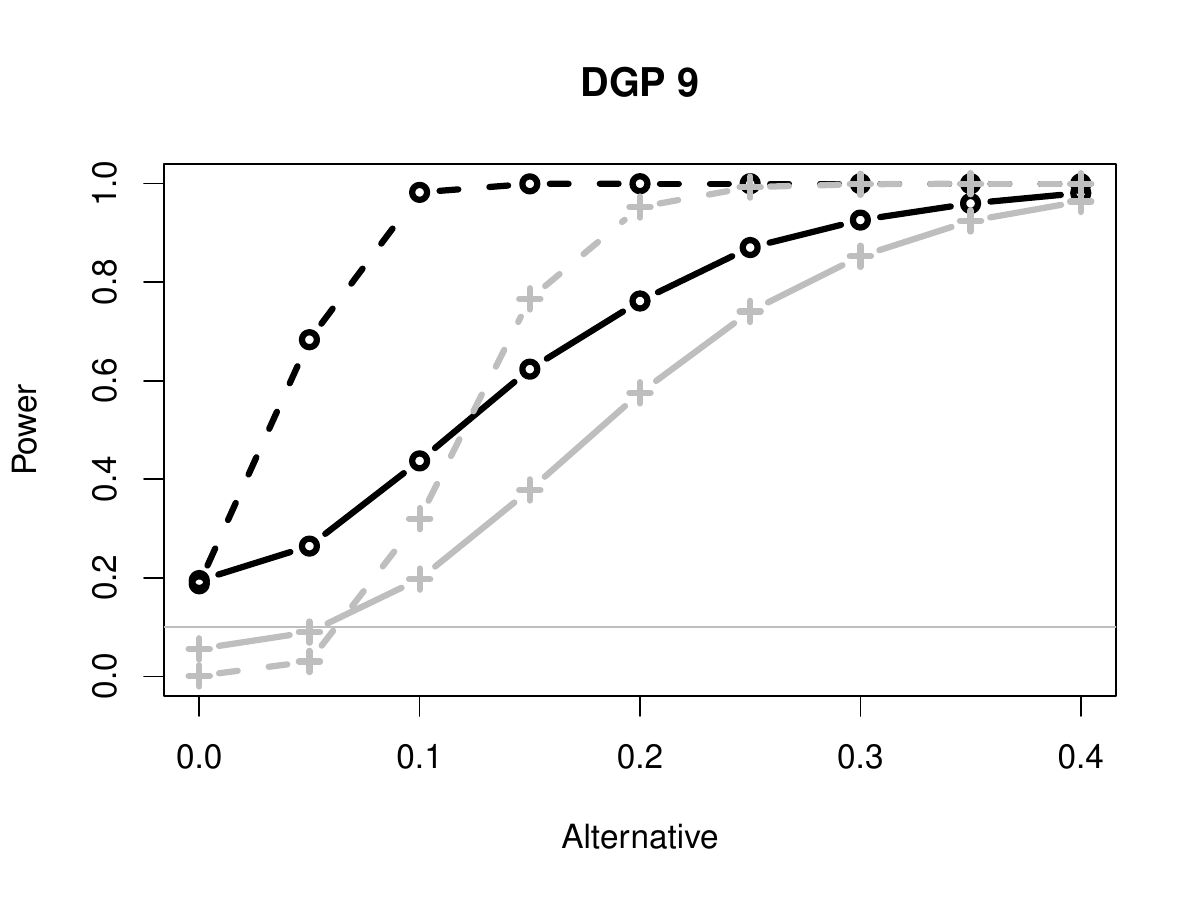}            
        
	\end{center}
\scriptsize{\textit{Notes:} Simulations with 10,000 repetitions. Nominal level: $\alpha=0.1$. Small: $(T_0,T_1,N)=(30,16,14)$ and $K=4$ for the $t$-test.  Large: $(T_0,T_1,N)=(150,16,14)$ and $K=8$ for the $t$-test. The DGPs are described in Section \ref{sec:simulations}.}	
\end{figure}

\subsection{The impact of time-varying weights}
\label{app: time-varying weights}

Here we illustrate the impact of the variability in the weights on the coverage and average length of the confidence intervals. We consider two scenarios. 

\begin{enumerate}
\item \textit{Correct specification.} The weights satisfy the SC constraints in all periods, $w_t\in \mathcal{W}^{SC}$ for all $t$. Specifically, we consider setting where the weights follow a Markov switching process with two regimes we considered in  Section \ref{sec:simulations}, a SC regime and a DID regime.\footnote{Markov-switching processes have a long tradition in economics \citep[e.g.,][]{hamilton1989new,baele2020flights,sims2006were,board2021learning}.} In the SC regime, the weights correspond to $w^{\text{SC}}$ (the SC weights estimated based on the application); in the DID regime, the weights correspond to $w^{\text{DID}}$ ($w^{\text{DID}}=(1/J,\dots,1/J)$). Under this DGP, a larger probability of staying in the same regime, $p_{\text{stay}}$, leads to more persistence and wider confidence intervals.
\item \textit{Misspecification.} The weights do not satisfy the SC constraints in all periods. Specifically, we consider a setting where we add Gaussian noise to $w^{\text{SC}}$, $w_t=w^{\text{SC}}+\zeta_t$, where $\zeta_t\overset{iid}\sim N(0,\sigma_w^2 I)$. The larger the variability in the weights, $\sigma_w$, the wider the confidence intervals.
\end{enumerate}

Figure \ref{fig:time_varying_weights} plots the coverage and average length of the 90\% confidence intervals under both scenarios. We report the length relative to the case where the weights are time-invariant and equal to $w^{\text{SC}}$. The simulation results confirm our theory. The confidence intervals have good coverage accuracy irrespective of the variability in the weights.\footnote{There is some undercoverage under correct specification when $p_{\text{stay}}=0.9$, because $p_\text{stay}$ induces additional persistence.} While it does not affect coverage, variability in the weights increases the length of the confidence intervals. The increase in length is more pronounced under misspecification because the variability in the weights is not restricted by the SC constraints in this case. 

\begin{figure}[ht]
\caption{Coverage and relative length of confidence intervals with time-varying weights}
	\label{fig:time_varying_weights}
	\begin{center}
		\includegraphics[width=0.4\textwidth,trim = {0 1cm 0 1cm}]{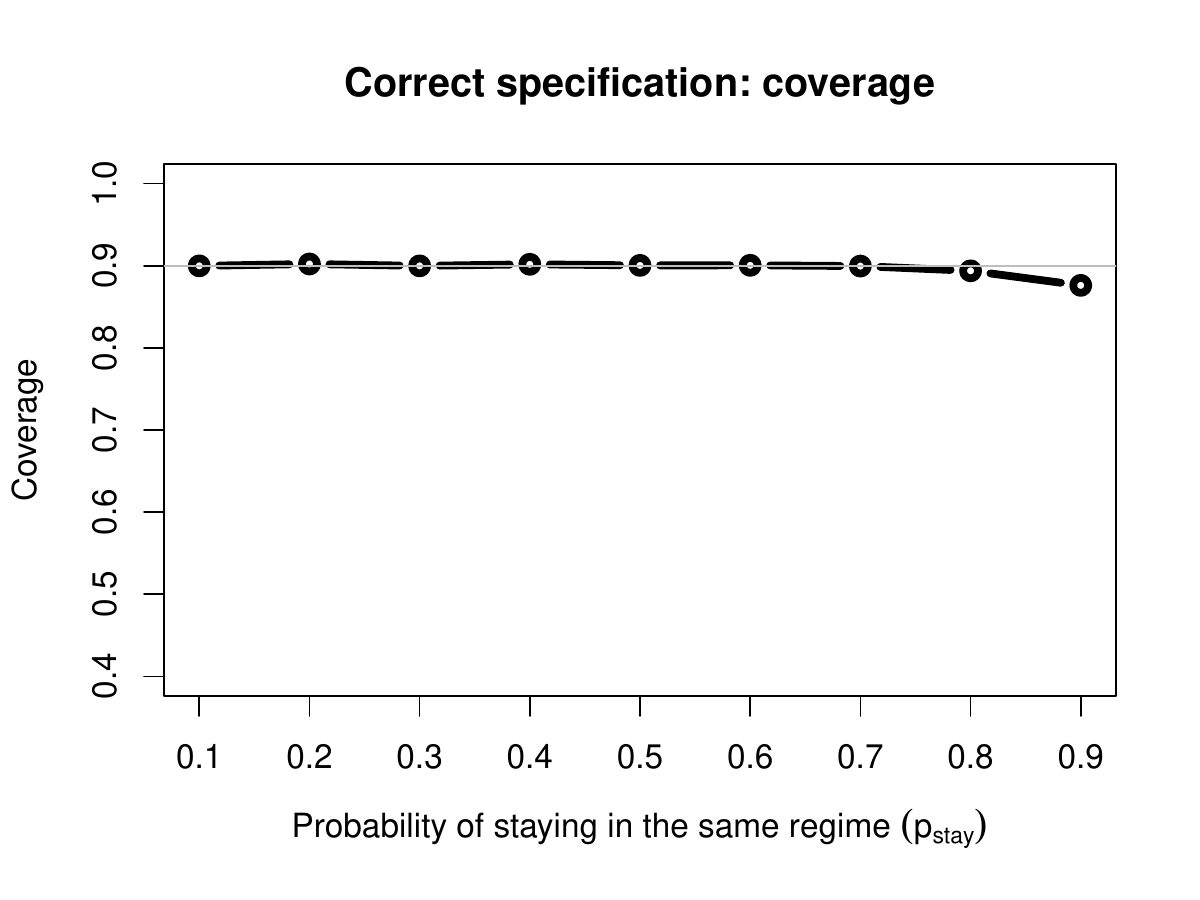}
		\includegraphics[width=0.4\textwidth,trim = {0 1cm 0 1cm}]{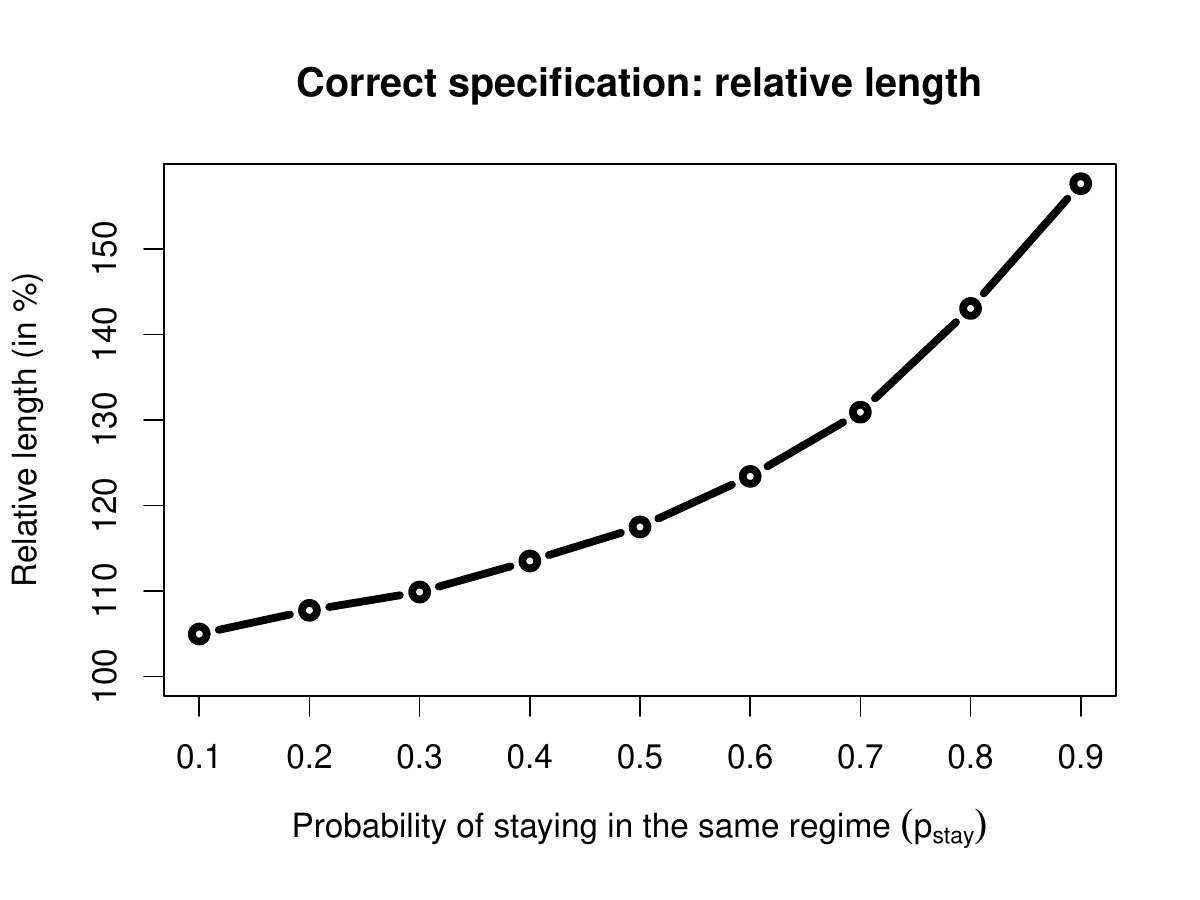}
		\includegraphics[width=0.4\textwidth,trim = {0 1.75cm 0 0cm}]{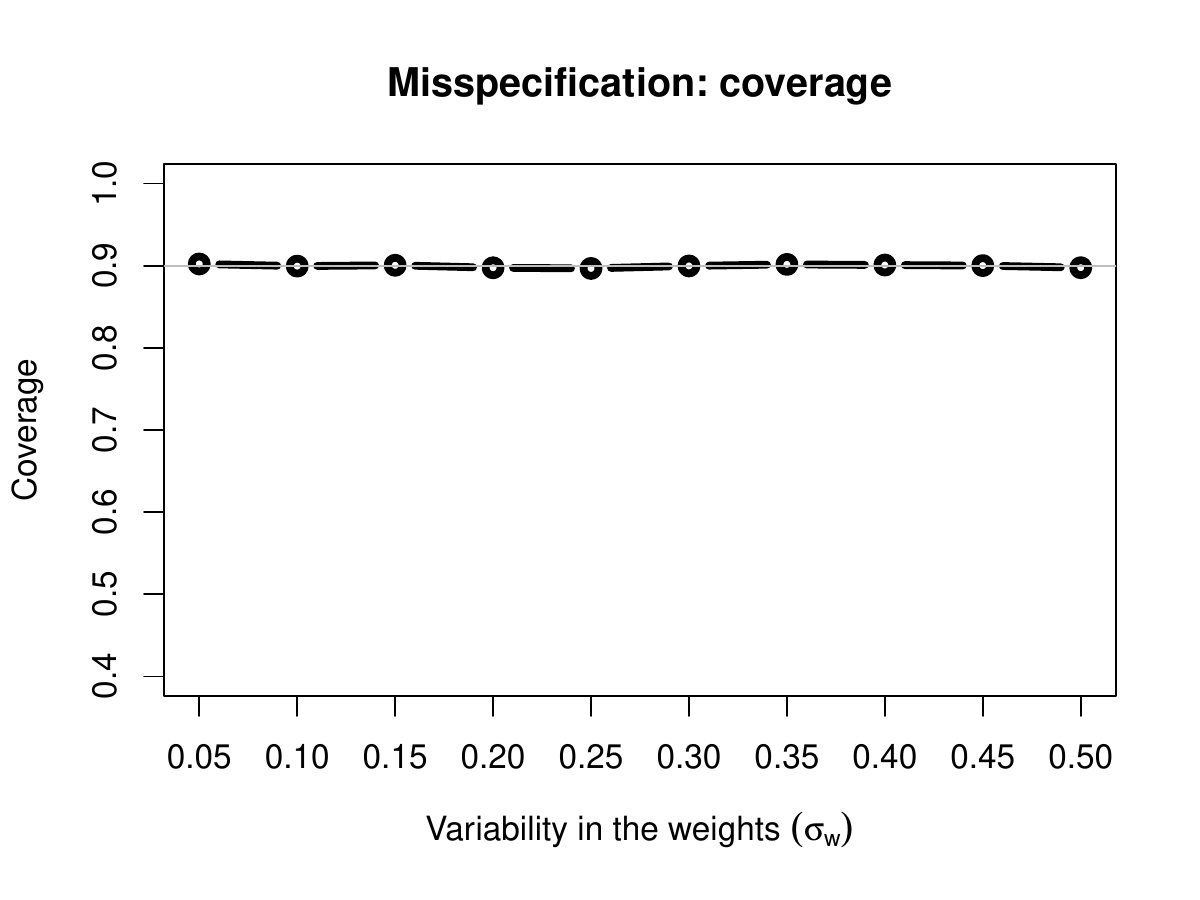}
		\includegraphics[width=0.4\textwidth,trim = {0 1.75cm 0 0cm}]{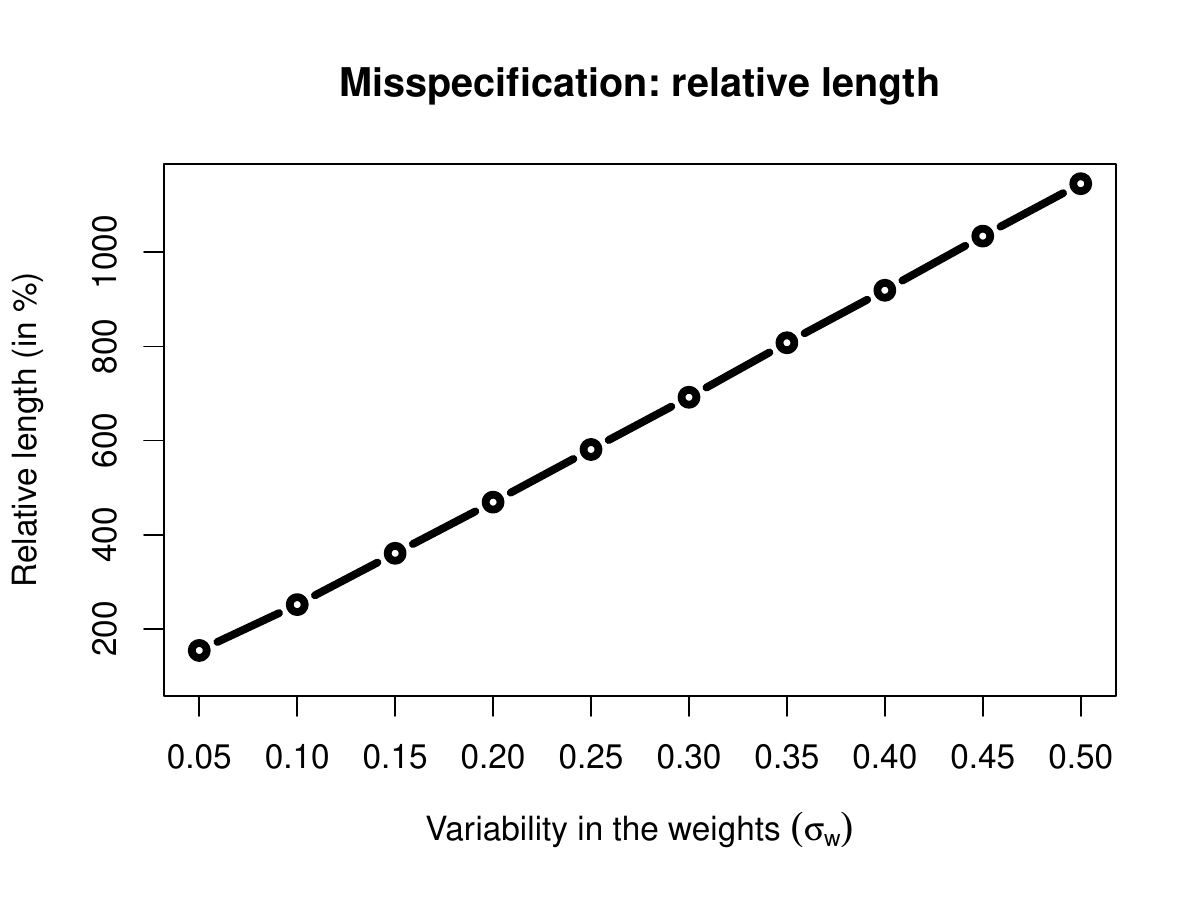}
	\end{center}
\scriptsize{\textit{Notes:} Simulations with 50,000 repetitions. $Y_{0t}(0)=\sum_{i=1}^Nw_{it}Y_{it}(0)+u_t$, $Y_{it}(0)\overset{iid}\sim N(i,1)$, $\{u_t\}$ is a Gaussian AR(1) process with coefficient $0.31$, and $(T_0,T_1,N)=(30,16,14)$ as in the empirical application. The specification of the weights is described in the text. Coverage and relative  length are reported for 90\% confidence intervals. Debiasing is based on $K=3$.}	
\end{figure}

\end{document}